%% file: main.tex
\documentclass[conference, letter, 10pt, times]{IEEEtran}
 
\input{packages}

\input{commands}

\newcommand{\shortVersion}{false} %set to true to generate the short paper version
\newcommand{\hideComments}{true} %set to true to hide comments 

\title{\tool: Principled Detection of \\Speculative Information Flows}

\author{\IEEEauthorblockN{Marco Guarnieri\IEEEauthorrefmark{1},
Boris K\"{o}pf\IEEEauthorrefmark{2}, Jos\'{e} F. Morales\IEEEauthorrefmark{1}, 
Jan Reineke\IEEEauthorrefmark{3}, and 
Andr\'{e}s S\'{a}nchez\IEEEauthorrefmark{1} }
\IEEEauthorblockA{\IEEEauthorrefmark{1}\textit{IMDEA Software Institute}\qquad
\IEEEauthorrefmark{2}\textit{Microsoft Research}\qquad
\IEEEauthorrefmark{3}\textit{Saarland University}}}

\begin{document}

\maketitle

\thispagestyle{plain}
\pagestyle{plain}

\begin{abstract}
Since the advent of \spectre{}, a number of countermeasures have been proposed and deployed. 
Rigorously reasoning about their effectiveness, however, requires a
well-defined notion of security against speculative execution attacks,
which has been missing until now.

In this paper (1) we put forward  \textit{speculative non-interference}, the first semantic notion of security against speculative execution attacks, and (2) we develop \tool{}, an algorithm based on symbolic execution to automatically prove speculative non-interference, or to detect violations. 

We implement \tool{} in a tool, which we use to detect subtle leaks and optimizations opportunities in the way major compilers place \spectre{} countermeasures. A scalability analysis indicates that checking speculative non-interference does not exhibit fundamental bottlenecks beyond those inherited by symbolic execution.
\end{abstract}

\input{intro}

\input{illustration}

\input{language}

\input{security}

\input{automation}

\input{casestudies}

\input{xen}

\input{discussion}

\input{relatedwork}

\input{conclusions}

\bibliographystyle{IEEEtran}
{
\bibliography{bibfile-1}
}

\appendices

\renewcommand{\thesectiondis}[2]{\Alph{section}:}

\input{non-speculative-semantics}

\input{always-mispredict-semantics}

\input{symbolic-semantics-appendix}

\input{examples}

\onecolumn

\input{speculative-and-non-speculative-semantics}

\input{always-mispredict-and-non-speculative-semantics}

\input{always-mispredict-worst-case}

\input{correctness-symbolic-semantics}

\input{proofs}

\end{document}

%% file: packages.tex
\usepackage{stmaryrd}
\usepackage{amsfonts}
\usepackage{graphicx}
\let\ORIGINALincludegraphics\includegraphics
\usepackage{cite}
\usepackage{amsthm}
\usepackage{amssymb}
\usepackage{mathpartir}
\usepackage{verbatim}
\usepackage[hyphens]{url}
\usepackage{paralist}
\usepackage{xcolor}
\usepackage{thmtools}
\usepackage{thm-restate}

\usepackage{listings}
\usepackage{paralist}
\usepackage{subfigure}
\usepackage{xifthen}
\usepackage{mathtools}
\usepackage{multirow}
\usepackage{tikz-cd}
\usepackage{booktabs}
\usepackage{algorithm}
\usepackage[noend]{algpseudocode}
\usepackage{xspace}
\usetikzlibrary{decorations.markings}
\usetikzlibrary{decorations.pathmorphing}
\usetikzlibrary{shapes,fit,arrows.meta}
\usetikzlibrary{calc,trees,positioning,arrows,chains,shapes.geometric,%
    decorations.pathreplacing,decorations.pathmorphing,shapes,%
    matrix,shapes.symbols}
\tikzset{negated/.style={
        decoration={markings,
            mark= at position 0.5 with {
                \node[transform shape] (tempnode) {$\backslash$};
            }
        },
        postaction={decorate}
    }
}

\usepackage{extarrows}
\usepackage{mathtools}
\usepackage{subfigure}
\usepackage[Export]{adjustbox}[2018/04/08]
\usepackage{pdfpages}

\makeatletter
\newcommand{\xMapsto}[2][]{\ext@arrow 0599{\Mapstofill@}{#1}{#2}}
\def\Mapstofill@{\arrowfill@{\Mapstochar\Relbar}\Relbar\Rightarrow}
\providecommand*{\twoheadrightarrowfill@}{%
  \arrowfill@\relbar\relbar\twoheadrightarrow
}
\providecommand*{\xtwoheadrightarrow}[2][]{%
  \ext@arrow 0579\twoheadrightarrowfill@{#1}{#2}%
}
\makeatother

\usepackage{array}
\newcolumntype{C}[1]{>{\centering}m{#1}}
\usepackage{bm}
\usepackage{numprint}
\npthousandsep{\,}

\usepackage{xstring}

%% file: commands.tex
\newcommand{\hide}[1]{%
    \IfEqCase{\hideComments}{%
        {false}{#1}%
        {true}{}%
    }
}

\newcommand{\techReportAppendix}[1]{%
    \IfEqCase{\shortVersion}{%
        {false}{Appendix \ref{#1}}%
        {true}{\cite{techReport}}%
    }
}%

\newcommand{\techReportAppendices}[2]{%
    \IfEqCase{\shortVersion}{%
        {false}{Appendices \ref{#1}--\ref{#2}}%
        {true}{\cite{techReport}}%
    }
}%

\newcommand{\onlyTechReport}[1]{%
    \IfEqCase{\shortVersion}{%
        {false}{#1}%
    }
}%

\newcommand{\onlyShortVersion}[1]{%
    \IfEqCase{\shortVersion}{%
        {true}{#1}%
    }
}%

\newcommand{\lang}{$\mu$\textsc{Asm}}

\newcommand{\para}[1]{\subsubsection*{#1}}
\newcommand{\spectre}{\textsc{Spectre}}
\newcommand{\tool}{\textsc{Spectector}\xspace}
\newcommand{\clang}{\textsc{Clang}}
\newcommand{\icc}{\textsc{Icc}}
\newcommand{\vcc}{\textsc{Visual C++}}

\newcommand{\ciao}{\textsc{Ciao}}

\AtBeginDocument{%
  \mathchardef\mathcomma\mathcode`\,
  \mathcode`\,="8000 
}
{\catcode`,=\active
  \gdef,{\mathcomma\discretionary{}{}{}}
}

\newcommand\fcirc{{\color{black}\bullet}\mathllap{\circ}}

\theoremstyle{definition}
\newtheorem{definition}{Definition}
\newtheorem{exampleInt}{Example}
\theoremstyle{plain}

\newenvironment{example}
 {\begin{exampleInt} } 
 { $\hfill \blacksquare$\end{exampleInt} }

\usepackage{letltxmacro}
\newcommand*{\SavedLstInline}{}
\LetLtxMacro\SavedLstInline\lstinline
\DeclareRobustCommand*{\lstinline}{%
  \ifmmode
    \let\SavedBGroup\bgroup
    \def\bgroup{%
      \let\bgroup\SavedBGroup
      \hbox\bgroup
    }%
  \fi
  \SavedLstInline
}

\lstdefinestyle{Cstyle}
{
	frame = tb,
  belowskip=.4\baselineskip,
  aboveskip=.4\baselineskip,
  	showstringspaces = false,
  	breaklines = true,
  	breakatwhitespace = true,
  	tabsize = 3,
  	numbers = left,
    stepnumber = 1,
    numberstyle = \tiny\color{gray},
    language = {[ANSI]C},
    alsoletter={.\$},
    basicstyle={\ttfamily\color{black}},
    stringstyle={\ttfamily\color{string-color}},
    keywordstyle={\ttfamily\color{Blue3}},
    keywordstyle=[2]{\ttfamily\color{Green4}},
    keywordstyle=[3]{\ttfamily\color{orange}},
    keywordstyle=[4]{\ttfamily\color{violet}},
    otherkeywords = {temp,y,A,B,k,size},
    morekeywords = [2]{A,B},
    morekeywords = [3]{},
    morekeywords = [4]{y,size,k,temp},
}

\lstdefinestyle{ASMstyle}
{
	frame = tb,
  belowskip=.4\baselineskip,
  aboveskip=.4\baselineskip,
  	showstringspaces = false,
              	breaklines = true,
  	breakatwhitespace = true,
  	tabsize = 3,
  	numbers = left,
    stepnumber = 1,
    numberstyle = \tiny\color{gray},
    alsoletter={.\$\%},
    basicstyle={\ttfamily\color{black}},
    stringstyle={\ttfamily\color{string-color}},
    keywordstyle={\ttfamily\color{Blue3}},
    keywordstyle=[2]{\ttfamily\color{Green4}},
    keywordstyle=[3]{\ttfamily\color{orange}},
    keywordstyle=[4]{\ttfamily\color{violet}},
    keywordstyle=[5]{\ttfamily\color{red}}, 
    keywords = {cmova, cmovae, cmovb, cmove, cmovne, cmovbe, mov, xor, or,  jbe, jne, add, jmp,  cmp, shl, and, sar, lfence, jae, clt,  and, lea},
    morekeywords = {},
    morekeywords = [2]{A, B},
    morekeywords = [3]{\%rax, \%rbx, \%rcx, \%eax, \%ecx, \%rbp, \%cl, \%rdi, \%edx, \%esi, \%rsi, \%rdx, \%rsp, \%al, \%rip, \%ax, \%bx, \%cx,  \%bp,  \%di, \%dx, \%si, \%sp, \%ip, \%r8b},
    morekeywords = [4]{k, y, size, temp},
    morekeywords = [5]{END, .L1, .L2}
}

\newcommand{\inlineCcode}[1]{\lstinline[style=Cstyle]|#1|}
\newcommand{\inlineASMcode}[1]{\lstinline[style=ASMstyle]|#1|}

\definecolor{Blue3}{HTML}{0000CD}
\definecolor{Green4}{HTML}{008B00}
\definecolor{Red3}{HTML}{CD0000}
\definecolor{orange}{rgb}{0.8, 0.47, 0.196}

\newcommand\remove[1]{}

 \newcommand\xleadsto[1]{%
   \if\relax\detokenize{#1}\relax
   \rightsquigarrow
   \else
   \mathrel{%
     \begin{tikzpicture}[%
       baseline={(current bounding box.south)}
       ]
       \node[%
       ,inner sep=.44ex
       ,align=center
       ] (tmp) {$\scriptstyle #1$};
       \path[%
       ,draw,<-
       ,decorate,decoration={%
         ,zigzag
         ,amplitude=0.7pt
         ,segment length=1.2mm,pre length=3.5pt
       }
       ] 
       (tmp.south east) -- (tmp.south west);
     \end{tikzpicture}
   }
   \fi
 }

\newcommand{\histories}{\mathcal{H}}
\newcommand{\tup}[1]{\langle #1 \rangle}
\newcommand{\Var}{\mathit{Regs}}
\newcommand{\Val}{\mathit{Vals}}
\newcommand{\Nat}{\mathbb{N}}

\newcommand{\Mem}{\mathit{Mem}}
\newcommand{\Assgn}{\mathit{Assgn}}
\newcommand{\Addr}{\mathit{Addrs}}
\newcommand{\Conf}{\mathit{Conf}}
\newcommand{\Init}{\mathit{InitConf}}
\newcommand{\Final}{\mathit{FinalConf}}

\newcommand{\Obs}{\mathit{Obs}}
\newcommand{\ExtObs}{\mathit{ExtObs}}
\newcommand{\Prg}{\mathit{Prg}}
\newcommand{\bp}{{\mathcal{O}}}
\newcommand{\lbl}{\ell}
\newcommand{\kywd}[1]{\mathbf{#1}}
\newcommand{\skipKywd}{\kywd{skip}}
\newcommand{\storeKywd}{\kywd{store}}
\newcommand{\loadKywd}{\kywd{load}}
\newcommand{\jmpKywd}{\kywd{jmp}}
\newcommand{\jzKywd}{\kywd{beqz}}
\newcommand{\pcKywd}{\kywd{pc}}
\newcommand{\barrierKywd}{\kywd{spbarr}}
\newcommand{\pc}{\pcKywd}
\newcommand{\obsKywd}[1]{\textcolor{Blue3}{\mathtt{#1}}}
\newcommand{\loadObsKywd}{\obsKywd{load}}
\newcommand{\storeObsKywd}{\obsKywd{store}}
\newcommand{\pcObsKywd}{\obsKywd{pc}}
\newcommand{\loadObs}[1]{\loadObsKywd\ #1}
\newcommand{\storeObs}[1]{\storeObsKywd\ #1}
\newcommand{\pcObs}[1]{\pcObsKywd\ #1}
\newcommand{\startObsKywd}[1]{\obsKywd{start}}
\newcommand{\commitObsKywd}[1]{\obsKywd{commit}} 
\newcommand{\rollbackObsKywd}[1]{\obsKywd{rollback}}
\newcommand{\startObs}[1]{\startObsKywd{}\ #1}

\newcommand{\commitObs}[1]{\commitObsKywd{}\ #1}
\newcommand{\rollbackObs}[1]{\rollbackObsKywd{}\ #1}
\newcommand{\unaryOp}[1]{\ominus #1}
\newcommand{\binaryOp}[2]{#1 \otimes #2}
\newcommand{\cmd}[2]{ #1 : #2 }
\newcommand{\pseq}[2]{#1 ; #2}
\newcommand{\pskip}{\skipKywd{}}
\newcommand{\passign}[2]{#1 \leftarrow #2}
\newcommand{\pload}[2]{\loadKywd\ #1, #2}
\newcommand{\pstore}[2]{\storeKywd\ #1, #2}
\newcommand{\pcondassign}[3]{#1 \xleftarrow{#3?} #2}
\newcommand{\pjmp}[1]{\jmpKywd\ #1}
\newcommand{\pjz}[2]{\jzKywd\ #1, #2}
\newcommand{\pbarrier}{\barrierKywd}
\newcommand{\select}[2]{#1(#2)}
\newcommand{\eval}[2]{\xrightarrow{#2}}
\newcommand{\speval}[3]{\xleadsto{#3}}
\newcommand{\exprEval}[2]{\llbracket #1 \rrbracket(#2)}

\newcommand{\specProject}[1]{#1\mathord{\upharpoonright_{se}}}
\newcommand{\nspecProject}[1]{#1\mathord{\upharpoonright_{nse}}}

\newcommand{\specStates}{\mathit{SpecS}}
\newcommand{\ExtConf}{\mathit{ExtConf}}
\newcommand{\exhausted}[1]{\mathit{enabled}(#1)}

\newcommand{\zeroes}[1]{\mathit{zeroes}(#1)}

\newcommand{\id}{\mathit{id}}
\newcommand{\ctr}{\mathit{ctr}}

\newcommand{\expand}[2]{\textsc{expand}(#1,#2)}

\newcommand{\pcSim}{\sim_{\pc}}
\newcommand{\nextSim}{\sim_{\mathit{next}}}

\newcommand{\mbp}{{\mathcal{O}}}

\newcommand{\finiteSequences}[1]{#1^*}

\newcommand{\prefix}[2]{{#1}^{#2}}
\newcommand{\elt}[2]{#1|_{#2}}
\newcommand{\emptysequence}{\varepsilon}
\newcommand{\concat}{\cdot}
\newcommand{\decrement}[1]{\mathit{decr}(#1)}

\newcommand{\indist}[1]{\sim_{#1} }

\newcommand{\sexpr}{\mathit{se}}

\newcommand{\SymVal}{\mathit{SymbVals}}
\newcommand{\SymExpr}{\mathit{SymbExprs}}
\newcommand{\ite}[3]{\mathbf{ite}(#1,#2,#3)}

\newcommand{\symEval}[2]{\xrightarrow{#2}_{\mathit{s}}}
\newcommand{\symExprEval}[2]{#2(#1)}

\newcommand{\symTraces}[2]{\mathit{symTraces}_{#1}(#2)}

\newcommand{\symb}[1]{#1_s}
\newcommand{\symWrite}[3]{\kywd{write}(#1,#2,#3)}
\newcommand{\symRead}[2]{\kywd{read}(#1,#2)}

\newcommand{\policy}{P}

\newcommand{\symPcObs}[1]{\obsKywd{symPc}(#1)}
\newcommand{\cstrs}[1]{\mathit{obsEqv}(#1)}
\newcommand{\pathCond}[1]{\mathit{pthCnd}(#1)}
\newcommand{\policyEqv}[1]{\mathit{polEqv}(#1)}
\newcommand{\memcheck}{\textsc{MemLeak}}
\newcommand{\pccheck}{\textsc{CtrlLeak}}

\algdef{SE}[DOWHILE]{Do}{doWhile}{\algorithmicdo}[1]{\algorithmicwhile\ #1}%
\newcommand{\sameSymPc}[1]{\mathit{sameSymbPc}(#1)}
\newcommand{\decrTop}[1]{\mathit{decr}'(#1)}
\newcommand{\decrementTop}[1]{\decrTop{#1}}
\newcommand{\window}[1]{\mathit{wndw}(#1)} 
\newcommand{\exhaustedTop}[1]{\mathit{enabled}'(#1)}
\newcommand{\zeroesTop}[1]{\mathit{zeroes}'(#1)}

\newcommand{\unp}{\textsc{Unp}}
\newcommand{\fen}{\textsc{Fen}}
\newcommand{\slh}{\textsc{Slh}}
\newcommand{\opt}{\texttt{-O2}}
\newcommand{\unopt}{\texttt{-O0}}

\newcommand{\lbbar}{\{\kern-0.5ex|}
\newcommand{\rbbar}{|\kern-0.5ex\}}
\newcommand{\lanbar}{\langle \kern-0.5ex|}
\newcommand{\ranbar}{|\kern-0.5ex\rangle}
\newcommand{\specEval}[3]{\llbracket #1 \rrbracket_{#2}(#3)}
\newcommand{\nspecEval}[2]{\llparenthesis #1 \rrparenthesis(#2)}
\newcommand{\amEval}[3]{\lbbar #1 \rbbar_{#2}(#3)}
\newcommand{\ameval}[1]{\xLongrightarrow{#1}}

\newcommand{\symSpeval}[3]{\xLongrightarrow{#3}_{\mathit{s}}}

\newcommand{\nspecTraces}[1]{\llparenthesis #1 \rrparenthesis}
\newcommand{\specTraces}[2]{\llbracket #1 \rrbracket_{#2}}
\newcommand{\amTraces}[2]{\lbbar #1 \rbbar_{#2}}
\newcommand{\symbTraces}[2]{\amTraces{#1}{#2}^\mathit{symb}} 

%% file: intro.tex
\section{Introduction}\label{sec:intro}

\textit{Speculative execution} avoids expensive pipeline stalls by predicting the outcome of branching (and other) decisions, and by
speculatively executing the corresponding instructions. 
If a prediction is incorrect, the processor aborts the speculative execution and rolls back the effect of the 
speculatively executed instructions on the architectural (ISA) state, which consists of registers, flags, and main memory. 

However, the speculative execution's effect on the microarchitectural state, which comprises the content of the cache, is not (or only partially) rolled back. 
This side effect can leak information about the speculatively accessed data and thus violate confidentiality.  
The family of \spectre{} attacks~\cite{Maisuradze:2018:RSE:3243734.3243761, Kocher2018spectre,schwarz2018netspectre,220586,kiriansky2018speculative} demonstrates that this vulnerability affects all modern general-purpose processors and poses a serious threat for platforms with multiple tenants.

Since the advent of \spectre{}, a number of countermeasures have been
proposed and deployed. 
At the software-level, these include, for instance, the insertion of serializing instructions~\cite{Intel}, the use of branchless bounds checks~\cite{Webkit}, and speculative load hardening~\cite{spec-hard}.
Several compilers support the automated insertion of these countermeasures during compilation~\cite{Intel-compiler,microsoft,spec-hard}, and the first static analyses to help identify vulnerable code patterns are emerging~\cite{oo7}.

However, we still lack a precise characterization of  {\em security against speculative execution attacks}.
Such a characterization is a prerequisite for reasoning about the
effectiveness of countermeasures, and for making principled decisions about
their placement.
It would enable one, for example, to identify cases where countermeasures do not prevent all attacks, or where they are unnecessary.

\subsubsection*{Our approach}
We develop a novel, principled approach for detecting information flows introduced by speculative execution, and for reasoning about software defenses against \spectre{}-style attacks.
Our approach is backed by a semantic notion of security against speculative execution attacks, and it comes with an algorithm, based on symbolic execution, for proving the absence of speculative leaks.

\subsubsection*{Defining security}
The foundation of our approach is \textit{speculative non-interference}, a novel semantic notion of security against speculative execution attacks.
Speculative non-interference is based on comparing a
program with respect to two different semantics:
\begin{asparaitem}
\item The first is a standard, {\em non-speculative semantics}. We use this semantics as a proxy for the intended program
  behavior.
\item The second is a novel, {\em speculative semantics} that can
  follow mispredicted branches for a bounded number of steps before
  backtracking. We use this semantics to capture the effect of
  speculatively executed instructions.
\end{asparaitem}

In a nutshell, speculative non-interference requires that {\em
  speculatively executed instructions do not leak more information
  into the microarchitectural state than what the intended behavior does},
i.e., than what is leaked by the standard, non-speculative semantics.\looseness=-1

To capture ``leakage into the microarchitectural state'', we
consider an observer of the program execution that sees the locations
of memory accesses and jump targets. This observer model is commonly used for
characterizing ``side-channel free'' or ``constant-time''
code~\cite{MolnarPSW05,AlmeidaBBDE16} in the absence of detailed
models of the microarchitecture.

Under this observer model, an adversary may distinguish two initial program states if they yield different traces of memory locations and jump targets.
{\em Speculative non-interference} (SNI)  requires that two initial program states can be distinguished under the speculative semantics only if they can also be distinguished under the standard, non-speculative semantics.\looseness=-1

The speculative semantics, and hence SNI, depends on the decisions taken by a branch predictor. 
We show that one can abstract from the specific predictor by considering a worst-case predictor that mispredicts every branching decision. SNI w.r.t. this worst-case predictor implies SNI w.r.t. a large class of real-world branch predictors, without introducing false alarms.

\subsubsection*{Checking speculative non-interference}

We propose \tool{}, an algorithm to automatically prove that programs satisfy SNI.
Given a program $p$, \tool{} uses symbolic execution with respect to the speculative semantics and the worst-case branch predictor to derive a concise representation of the traces of memory accesses and jump targets during execution along all possible program paths. 

Based on this representation, \tool creates an SMT
formula that captures that, whenever two initial program states produce the
same memory access patterns in the standard semantics, they also
produce the same access patterns in the speculative
semantics. Validity of this formula for each program path implies
speculative non-interference.

\subsubsection*{Case studies}
We implement a prototype of \tool{}, with a front end for parsing (a subset of) x86 assembly and a back end for solving SMT formulas using the Z3 SMT solver. 
We perform two case studies where we evaluate the precision and scalability of \tool{}.\looseness=-1

\begin{asparaitem}
\item For evaluating precision, we analyze the 15 variants of \spectre{} v1~by Kocher~\cite{Kocher2018examples}.
We create a corpus of 240 microbenchmarks by compiling the 15 programs with the \clang{}, \textsc{Intel} \icc{}, and Microsoft \vcc{} compilers, using different levels of optimization and protection against \spectre{}.
Using \tool{}, we successfully (1) detect all leaks pointed out
in~\cite{Kocher2018examples}, (2) detect novel, subtle leaks that are
out of scope of existing approaches that check for known vulnerable code
patterns~\cite{oo7}, and (3) identify cases where compilers
unnecessarily inject countermeasures, i.e., opportunities for
optimization without sacrificing security.
  \item For evaluating scalability, we apply \tool{} to the codebase of the Xen Project Hypervisor. Our evaluation indicates that the cost of checking speculative non-interference is comparable to that of discovering symbolic paths, which shows that our approach does not exhibit bottlenecks beyond those inherited by symbolic execution.
  \end{asparaitem}

 \subsubsection*{Scope}
We focus on leaks introduced by speculatively executed instructions resulting from mispredicted branch outcomes, such as those exploited in \spectre{} v1~\cite{Kocher2018spectre}.
For an in-depth discussion of our approach's scope, see Section~\ref{sec:discussion}.

\subsubsection*{Summary of contributions}
Our contributions are both theoretical and practical.
On the theoretical side, we present \textit{speculative non-interference}, the first semantic notion of security against speculative execution attacks.
On the practical side, we develop \tool{}, an automated technique for
detecting speculative leaks (or prove their absence), and we use it to
detect subtle leaks -- and optimization opportunities -- in the way 
 compilers inject \spectre{} countermeasures.

\tool{} is available at \texttt{\url{https://spectector.github.io}}. \onlyShortVersion{An extended version of this paper containing proofs of the technical results is available at~\cite{techReport}.}

%% file: illustration.tex
\section{Illustrative example}\label{sec:illustration}

\begin{figure}
	
\begin{lstlisting}[style=Cstyle]
if (y < size)
 temp &= B[A[y] * 512];
\end{lstlisting}
\caption{\spectre{} variant 1 - C code}\label{figure:spectrev1:c-code}
\end{figure}

\begin{figure}
\begin{lstlisting}[style=ASMstyle]
 mov	size, %rax
 mov	y, %rbx
 cmp	%rbx, %rax
 jbe	END
 mov	A(%rbx), %rax
 shl	$9, %rax
 mov	B(%rax), %rax
 and	%rax, temp
\end{lstlisting}
\caption{\spectre{} variant 1 - Assembly code}\label{figure:spectrev1:asm-code}
\end{figure}

To illustrate our approach, we show how  \tool{} applies to the \spectre{}  v1 example~\cite{Kocher2018spectre} shown in Figure~\ref{figure:spectrev1:c-code}.

\para{Spectre v1}
The program checks whether the index stored in the variable \inlineCcode{y} is less than the size of the array \inlineCcode{A}, stored in the variable \inlineCcode{size}.
If that is the case, the program retrieves \inlineCcode{A[y]}, amplifies it with a multiple (here: \inlineCcode{512}) of the cache line size, and uses the result as an address for accessing the array \inlineCcode{B}.\looseness=-1

If \inlineCcode{size} is not cached, evaluating the branch condition requires traditional processors to wait until \inlineCcode{size} is fetched from main memory.
Modern processors instead speculate on the condition's outcome and continue the computation.
Hence, the memory accesses in line 2 may be executed even if $\inlineCcode{y} \geq \inlineCcode{size}$.\looseness=-1

When \inlineCcode{size} becomes available, the processor checks whether the speculated branch is the correct one.
If it is not, it rolls back the  architectural (i.e., ISA) state's changes and executes the correct branch.
However, the speculatively executed memory accesses leave a footprint
in the microarchitectural state, in particular in the cache, which enables an adversary to retrieve 
\inlineCcode{A[y]}, even for $\inlineCcode{y} \geq \inlineCcode{size}$, by probing the array \inlineCcode{B}.

\para{Detecting leaks with \tool{}}
\tool{} automatically detects leaks introduced by speculatively executed instructions, or proves their absence.
Specifically, \tool{} detects a leak whenever executing the program
under the speculative semantics, which captures that the execution can
go down a mispredicted path for a bounded number of steps, leaks more
information into the microarchitectural state than executing the program
under a non-speculative semantics. 

To illustrate how \tool{} operates, we consider the x86 assembly\footnote{We use a simplified AT\&T syntax without operand sizes} translation of Figure~\ref{figure:spectrev1:c-code}'s program (cf.~Figure~\ref{figure:spectrev1:asm-code}).

\tool{} performs symbolic execution with respect to the speculative
semantics to derive a concise representation of the concrete traces of memory
accesses and program counter values along each path of the program.
These symbolic traces capture the program's effect on the microarchitectural state. 

During speculative execution, the speculatively executed parts are determined by the predictions of the branch predictor.
As shown in Section~\ref{sec:automation:always-mispred}, leakage due to speculative execution is maximized under a branch predictor that mispredicts every branch. 
The code in Figure~\ref{figure:spectrev1:asm-code} yields two
symbolic traces w.r.t. the speculative semantics that mispredicts every branch:\footnote{For simplicity of presentation, the example traces capture only
  loads but not the program counter. }
\begin{equation}\label{tr:inbounds}
\startObsKywd{} \concat\rollbackObsKywd{}\concat\tau \quad \text{when} \quad
  \inlineASMcode{y} < \inlineASMcode{size}
\end{equation}\vspace{-15pt}
\begin{equation} \label{tr:outbounds}
\startObsKywd{} \concat \tau  \concat
    \rollbackObsKywd{} \quad \text{when} \quad \inlineASMcode{y} \geq \inlineASMcode{size}
\end{equation}
where
$\tau=\loadObs{(\inlineCcode{A} + \inlineASMcode{y} )} \concat
\loadObs{(\inlineASMcode{B} +\inlineASMcode{A[y] * 512})}$.  Here, the
argument of $\loadObsKywd$ is visible to the observer, while
$\startObsKywd{}$ and $\rollbackObsKywd{}$ denote the start and the
end of a misspeculated execution. The traces of the {\em
  non-speculative} semantics are obtained from those of the
speculative semantics by removing all observations in between
$\startObsKywd{}$ and $\rollbackObsKywd{}$. 

Trace~\ref{tr:inbounds} shows that whenever \inlineASMcode{y} is in
bounds (i.e., $\inlineASMcode{y} < \inlineASMcode{size}$) the
observations of the speculative semantics and the non-speculative
semantics coincide (i.e. they are both $\tau$). In contrast,
Trace~\ref{tr:outbounds} shows that whenever
$\inlineASMcode{y} \geq \inlineASMcode{size}$, the speculative
execution generates observations $\tau$ that depend on \inlineASMcode{A[y]}
whose value is not visible in the non-speculative
execution.
This is flagged as a leak by \tool{}.

\para{Proving security with \tool{}}
The \clang{} 7.0.0 C++
compiler implements a countermeasure, called
speculative load hardening~\cite{spec-hard}, that applies conditional
masks to addresses to prevent leaks into the
microarchitectural state. Figure~\ref{figure:spectrev1:asm-code-slh}
depicts the protected output of \clang{} on the program from
Figure~\ref{figure:spectrev1:c-code}.

\begin{figure}[h]
\begin{lstlisting}[style=ASMstyle]
 mov    size, %rax
 mov    y, %rbx
 mov    $0, %rdx              
 cmp    %rbx, %rax
 jbe    END
 cmovbe $-1, %rdx             
 mov    A(%rbx), %rax
 shl    $9, %rax
 or     %rdx, %rax 
 mov    B(%rax), %rax
 or     %rdx, %rax             
 and    %rax, temp
\end{lstlisting}
%\vspace{-10pt}
\caption{\spectre{} variant 1 - Assembly code with speculative load hardening.
\clang{} inserted instructions 3, 6, 9, and 11.}\label{figure:spectrev1:asm-code-slh}
\end{figure}

The symbolic execution of the speculative semantics produces, as before,
 Trace~\ref{tr:inbounds} and Trace~\ref{tr:outbounds}, but with 
\begin{equation*}
\tau=\loadObs{ (\inlineCcode{A}
    + \inlineASMcode{y} )} \concat \loadObs{(\inlineASMcode{B}
    +(\inlineASMcode{A[y] * 512}) \,|\, \mathit{mask})},
\end{equation*}
where
$\mathit{mask} = \ite{\inlineASMcode{y} <
  \inlineASMcode{size}}{\texttt{0x0}}{\texttt{0xFF..FF}}$ corresponds
to the conditional move in line $6$ and $|$ is a bitwise-or operator. 
Here, $\ite{\inlineASMcode{y} <
  \inlineASMcode{size}}{\texttt{0x0}}{\texttt{0xFF..FF}}$ is a symbolic if-then-else expression evaluating to $\texttt{0x0}$ if $\inlineASMcode{y} <
  \inlineASMcode{size}$ and to $\texttt{0xFF..FF}$ otherwise.

The analysis of Trace 1 is as before. For Trace 2, however, \tool{}
determines (via a query to Z3~\cite{z3}) that, for all
$\inlineASMcode{y}\ge \inlineASMcode{size}$ %-- which we assume is public -- 
there is exactly {\em one} observation that the adversary can
make during the speculative execution, namely
$\loadObs{ (\inlineCcode{A} + \inlineASMcode{y} )} \concat
\loadObs{(\inlineASMcode{B} + \texttt{0xFF..FF})}$, from which it concludes that no
information leaks into the microarchitectural state, i.e., the
countermeasure is effective in securing the program. See Section~\ref{sec:casestudies}
for examples where \tool{} detects that countermeasures are not applied effectively.

%% file: language.tex
\section{Language and Semantics}\label{sec:language}
We now introduce \lang{}, a core assembly language which we use for defining speculative non-interference and describing \tool{}.

\subsection{Syntax}
The syntax of \lang{} is defined in Figure \ref{figure:language:syntax}.
Expressions are built from a set of register identifiers $\Var$, which contains a designated element $\pc$ representing the program counter, and a set $\Val$ of values, which consists of the natural numbers and $\bot$.
\lang{} features eight kinds of instructions: a $\pskip$ instruction, (conditional) assignments, load and store instructions, branching instructions, indirect jumps, and speculation barriers $\pbarrier{}$.
Both conditional assignments and speculation barriers are commonly used to implement \spectre{} countermeasures~\cite{Intel,spec-hard}.\looseness=-1

A \lang{} program is a sequence of pairs $\cmd{n}{i}$, where $i$ is an instruction and $n \in \Nat$ is a value representing the instruction's label.
We say that a program is \textit{well-formed} if (1) it does not contain duplicate labels, (2) it contains an instruction labeled with $0$, i.e., the initial instruction, and (3) it does not contain branch instructions of the form $\cmd{n}{\pjz{x}{n+1}}$.
In the following we consider only well-formed programs.

We often treat programs $p$ as partial functions from natural numbers to instructions.
Namely, given a program $p$ and a number $n \in \Nat$, we denote by $\select{p}{n}$ the instruction labelled with $n$ in $p$ if it exists, and $\bot$ otherwise.

\begin{example}\label{example:language:spectrev1}
	The \spectre{} v1 example from Figure~\ref{figure:spectrev1:c-code} can be expressed in \lang{} as follows:
	\begin{align*}
		& 	\cmd{0}{\passign{x}{\mathtt{y} < \mathtt{size}}}\\
		&	\cmd{1}{\pjz{x}{\bot}}\\
		&	\cmd{2}{\pload{z}{\mathtt{A} + \mathtt{y}}}\\
		&	\cmd{3}{\passign{z}{z * 512}}\\
		&	\cmd{4}{\pload{w}{\mathtt{B} + z}}\\
		&	\cmd{5}{\passign{\mathtt{temp}}{\mathtt{temp}\ \&\ w}}
	\end{align*}
	Here, registers $\mathtt{y}$, $\mathtt{size}$, and $\mathtt{temp}$ store the respective variables.
	Similarly, registers $\mathtt{A}$ and $\mathtt{B}$ store the memory addresses of the first elements of the arrays \inlineCcode{A} and \inlineCcode{B}.
\end{example}

\begin{figure}
\begin{tabular}{llcl}
\multicolumn{4}{l}{\bf Basic Types} \\
\textit{(Registers)} 	&  $x$		& $\in$ & $\Var$ \\
\textit{(Values)} 		&  $n, \lbl$ 		& $\in$ & $\Val = \Nat \cup \{\bot\}$  \\\\
\multicolumn{4}{l}{\bf Syntax} \\
\textit{(Expressions)} 	&  $e$		& $:=$ & $n \mid x \mid \unaryOp{e} \mid \binaryOp{e_1}{e_2}$ \\
\textit{(Instructions)} 	&  $i$ 		& $:=$ & $\pskip{} \mid \passign{x}{e} \mid \pload{x}{e} \mid $ \\
						&			&	   & $\pstore{x}{e} \mid \pjmp{e} \mid \pjz{x}{\lbl} \mid$ \\
						&			&	   & $ \pcondassign{x}{e}{e'} \mid \pbarrier{}$ \\
\textit{(Programs)}		&  $p$		& $:=$ & $\cmd{n}{i} \mid \pseq{p_1}{p_2}$
\end{tabular}
\caption{\lang{} syntax}\label{figure:language:syntax}
\end{figure}

\subsection{Non-speculative Semantics}\label{sect:language:non-speculative}
The standard, non-speculative semantics models the execution of \lang{} programs on a platform without speculation. 
This semantics is formalized as a ternary relation $\sigma \eval{p}{o}\sigma'$ mapping a configuration $\sigma$ to a configuration $\sigma'$, while producing an observation $o$.
Observations are used to capture what an adversary can see about a given execution trace. 
We describe the individual components of the semantics below.\looseness=-1

\begin{figure*}
	{
	\small
	\begin{mathpar}
			
	\inferrule[Load]
	{
	\select{p}{a(\pc)} = \pload{x}{e} \\
	x \neq \pc \\
	n = \exprEval{e}{a}
	}
	{
	\tup{m, a} \eval{p}{\loadObs{n}} \tup{m, a[\pc \mapsto a(\pc)+1, x \mapsto m(n)]}
	}

	\inferrule[Store]
	{
	\select{p}{a(\pc)} = \pstore{x}{e} \\
	n = \exprEval{e}{a}
	}
	{
	\tup{m, a} \eval{p}{\storeObs{n}} \tup{ m[n \mapsto a(x)], a[\pc \mapsto a(\pc)+1]}
	}

	\inferrule[CondUpdate-Sat]
	{
		p(a(\pc)) = \pcondassign{x}{e}{e'}\\
		\exprEval{e'}{a} = 0\\
		x \neq \pc
	}
	{
		\tup{m,a} \eval{p}{} \tup{m,a[\pc \mapsto a(\pc) + 1, x \mapsto \exprEval{e}{a}]}
	}

	\inferrule[Beqz-Sat]
	{
	\select{p}{a(\pc)} = \pjz{x}{\lbl} \\
	a(x) = 0
	}
	{
	\tup{m, a} \eval{p}{\pcObs{\lbl}} \tup{ m, a[\pc \mapsto \lbl]}
	}

	\inferrule[Jmp]
	{
	\select{p}{a(\pc)} = \pjmp{e} \\
	\lbl = \exprEval{e}{a}
	}
	{
	\tup{m, a} \eval{p}{\pcObs{\lbl}} \tup{ m, a[\pc \mapsto \lbl]}
	}
	\end{mathpar}
	}
	\caption{Standard semantics for \lang{} program $p$ -- selected rules}\label{figure:language:semantics}
	\end{figure*}

\para{Configurations}
A \textit{configuration} $\sigma$ is a pair $\tup{m, a}$ of a \textit{memory} $m\in\Mem$ and a {\em register assignment}  $a\in \Assgn$, modeling the state of the computation.
Memories $m$ are functions mapping memory addresses, represented by natural numbers, to values in $\Val$. 
Register assignments $a$ are functions mapping register identifiers to values.
We require that $\bot$ can only be assigned to the program counter $\pc$, signaling termination.
A configuration $\tup{m,a}$ is \textit{initial} (respectively \textit{final}) if $a(\pc)=0$ (respectively $a(\pc)=\bot$). 
We denote the set $\Mem \times \Assgn$ of all configurations by $\Conf$.

\para{Adversary model and observations}
We consider an adversary that observes the program counter and the locations of memory accesses during computation.
This adversary model is commonly used to formalize timing side-channel free code~\cite{MolnarPSW05,AlmeidaBBDE16}, without requiring microarchitectural models.
In particular, it captures leakage through caches without requiring an explicit cache model.

We model this adversary in our semantics by annotating transactions with observations  $\loadObs{n}$ and $\storeObs{n}$, which expose read and write accesses to an address $n$, and observations $\pcObs{n}$, which expose the value of the program counter.
We denote the set of all observations by $\Obs$.

\para{Evaluation relation}
We describe the execution of \lang{} programs using the evaluation relation $\eval{p}{} \subseteq \Conf \times \Obs \times \Conf$. 
Most of the rules defining $\eval{p}{}$ are fairly standard, which is why Figure~\ref{figure:language:semantics} presents only a selection. 
We refer the reader to Appendix~\ref{appendix:non-speculative-semantics} for the remaining rules.

The rules \textsc{Load} and \textsc{Store} describe the behavior of instructions $\pload{x}{e}$ and $\pstore{x}{e}$ respectively.
The former assigns to the register $x$ the memory content at the address $n$ to which expression $e$ evaluates; the latter stores the content of $x$ at that address.
Both rules expose the address $n$ using observations and increment the program counter.

The rule \textsc{CondUpdate-Sat} describes the behavior of a conditional update $\pcondassign{x}{e}{e'}$ whose condition $e'$ is satisfied.
It first checks that the condition $e'$ evaluates to  $0$.
It then updates the register assignment $a$ by storing in $x$ the value of $e$, and by incrementing $\pc$.

The rule \textsc{Beqz-Sat} describes the effect of the instruction $\pjz{x}{\lbl}$ when the branch is taken.
Under the condition that $x$ evaluates to $0$, it sets the program counter to $\lbl$ and exposes this change using the observation $\pcObs{\lbl}$.

Finally, the rule \textsc{Jmp} executes $\pjmp{e}$ instructions.
The rule stores the value of $e$ in the program counter and records this change using the observation $\pcObs{\lbl}$.

\para{Runs and traces}
The evaluation relation captures individual steps in the execution of a program.
Runs capture full executions of the program. We formalize them as triples $\tup{\sigma, \tau, \sigma'}$ consisting of an initial configuration~$\sigma$, a trace of observations~$\tau$, and a final configuration~$\sigma'$.
Given a program~$p$, we denote by~$\nspecTraces{p}$ the set of all possible runs of the non-speculative semantics, i.e., it contains all triples $\tup{\sigma, \tau, \sigma'}$ corresponding to executions $ \sigma \eval{p}{\tau}^* \sigma'$.
Finally, we denote by~$\nspecEval{p}{\sigma}$ the trace~$\tau$ such that there is a final configuration~$\sigma'$ for which~$\tup{\sigma, \tau, \sigma'} \in \nspecTraces{p}$.
In this paper, we only consider terminating programs.
Extending the definitions and algorithms to non-terminating programs is future work.

\section{Speculative semantics}\label{sect:language:speculative}

This section introduces a model of speculation that captures the execution of \lang{} programs on speculative in-order microarchitectures. 
We first informally explain this model in Section~\ref{modspec} before formalizing it in the rest of the section. 

\subsection{Modeling speculation}\label{modspec}
Non-branching instructions are executed as in the standard semantics. 
Upon reaching a branching instruction, the {\em prediction oracle}, which is a parameter of our model, is queried to obtain a branch prediction that is used to decide which of the two branches to execute speculatively. 

To enable a subsequent rollback in case of a misprediction, a snapshot of the current program configuration is taken, before starting a  \emph{speculative transaction}.
In this speculative transaction, the program is executed speculatively along the predicted branch for a bounded number of computation steps.
Computing the precise length $w$ of a speculative transactions would (among other aspects) require a detailed model of the memory hierarchy. To abstract from this complexity, in our model $w$ is also provided by the prediction oracle.

At the end of a speculative transaction, the correctness of the prediction is evaluated:
\begin{asparaitem}
\item If the prediction was {\em correct}, the transaction is committed and the computation continues using the current  configuration.
\item If the prediction was {\em incorrect}, the transaction is aborted, the original  configuration is restored, and the computation continues on the correct branch. 
\end{asparaitem}

In the following we formalize the behavior intuitively described above in the {\em speculative semantics}.
The main technical challenge lies in catering for nested branches and transactions.

\subsection{Prediction oracles}
In our model, prediction oracles serve two distinct purposes: (1) predicting branches, and (2) determining the speculative transactions' lengths. 
A \textit{prediction oracle} $\bp$ is a partial function that takes as input a program~$p$, a branching history~$h$, and a label~$\lbl$ such that $p(\lbl)$ is a branching instruction, and that returns as output a pair $\tup{\lbl',w} \in \Val \times \Nat$, where $\lbl'$ represents the predicted branch (i.e., $\lbl' \in \{\lbl+1, \lbl''\}$ where $p(\lbl) = \pjz{x}{\lbl''}$) and $w$ represents the speculative transaction's length.

Taking into account the {\em branching history} enables us to capture history-based branch predictors, a general class of branch predictors that base their decisions on the sequence of branches leading up to a branching instruction.
Formally, a branching history is a sequence of triples $\tup{\lbl, \id, \lbl'}$, where $\lbl \in \Val$ is the label of a branching instruction, $\lbl' \in \Val$ is the label of the predicted branch, and $\id\in\Nat$ is the identifier
 of the transaction in which the branch is executed. 

A prediction oracle $\bp$ has \textit{speculative window at most} $w$ if the length of the transactions generated by its predictions is at most $w$, i.e., for all programs $p$, branching histories $h$, and labels $\lbl$, $\bp(p, h, \lbl) = \tup{\lbl',w'}$, for some $\lbl'$ and with $w' \leq w$.

\begin{example}\label{example:speculative-semantics:btfnt-predictor}
The ``backward taken forward not taken'' (BTFNT) branch predictor, implemented in early CPUs~\cite{hennessy2011computer}, predicts the branch as taken if the target instruction address is lower than the program counter.
It can be formalized as part of a prediction oracle $\mathit{BTFNT}$, for a fixed speculative window~$w$, as follows:
$\mathit{BTFNT}(p, h, \lbl) = \tup{\mathit{min}(\lbl+1, \lbl'),w}$, where  $p(\lbl) = \pjz{x}{\lbl'}$.
\end{example}

Dynamic branch predictors, such as simple 2-bit predictors and more complex correlating or tournament predictors~\cite{hennessy2011computer}, can also be formalized using prediction oracles.

\subsection{Speculative transactions}
To manage each ongoing speculative transaction\footnote{Due to nesting, multiple transactions may be happening simultaneously.}, the speculative semantics needs to remember a snapshot $\sigma$ of the configuration prior to the start of the transaction, the length~$w$ of the transaction (i.e., the number of instructions left to be executed in this transaction), the branch prediction~$\lbl$ used at the start of the transaction, and the transaction's identifier~$\id$.
We call such a 4-tuple $\tup{\sigma,\id, w, \lbl} \in \Conf \times \Nat \times \Nat \times \Val$, a \textit{speculative state}, and we denote by $\specStates$ the set  of all speculative states.\looseness=-1

Nested transactions are represented by sequences of speculative states.
We use standard notation for sequences: $\finiteSequences{S}$ is the set of all finite sequences over the set $S$, $\emptysequence$ is the  empty sequence, and  $s_1 \concat s_2$ is the concatenation of sequences $s_1$ and $s_2$.\looseness=-1

We use the following two helper functions to manipulate sequences of speculative states $s \in \specStates^*$:
\begin{asparaitem}
	\item $\mathit{decr} : \specStates^* \to \specStates^*$ decrements by $1$ the length of all transactions in the sequence.
	\item  $\mathit{zeroes} : \specStates^* \to \specStates^*$ sets to $0$ the length of all transactions in the sequence.
	\item The predicate $\exhausted{s}$ holds if and only if none of the transactions in $s$ has remaining length~$0$.
\end{asparaitem}

In addition to branch and jump instructions, speculative transactions can also modify the program counter:
rolling back a transaction results in resetting the program counter to the one in the correct branch.
To expose such changes to the adversary, we extend the set $\Obs$ of observations with elements of the form $\startObs{\id}$, $\commitObs{\id}$, and $\rollbackObs{\id}$, to denote start,  commit, and  rollback of a speculative transaction $\id$.
$\ExtObs$ denotes the set of extended observations.

\subsection{Evaluation relation}
The speculative semantics operates on \textit{extended configurations}, which are 4-tuples $\tup{\ctr, \sigma, s,h} \in \ExtConf$ consisting of a global counter $\ctr \in \Nat$ for generating transaction identifiers, a configuration $\sigma \in \Conf$,  a sequence $s$ of speculative states representing the ongoing speculative transactions, and a branching history $h$. 
Along the lines of the standard semantics, we describe the speculative semantics of \lang{} programs under a prediction oracle $\bp$ using the relation $\speval{p}{\bp}{} \subseteq \ExtConf \times \ExtObs^* \times \ExtConf$. 
The rules are given in Figure~\ref{figure:language:speculative-execution} and are explained below:

\textsc{Se-NoBranch} captures the behavior of non-branching instructions as long as the length of all speculative states in $s$ is greater than $0$, that is, as long as $\exhausted{s}$ holds.
In this case, $\speval{p}{\bp}{}$ mimics the behavior of the non-speculative semantics~$\eval{p}{}$.
If the instruction is not a speculation barrier, the lengths of all speculative transactions are decremented by~1 using $\mathit{decr}$.
In contrast, if the instruction is a speculation barrier $\pbarrier$, the length of all transactions is set to $0$ using $\mathit{zeroes}$. 
In this way, $\pbarrier$ forces the termination (either with a commit or with a rollback) of all ongoing speculative transactions.

\textsc{Se-Branch} models the behavior of branch instructions.
The rule
\begin{inparaenum}[(1)]
\item queries the prediction oracle $\bp$ to obtain a prediction $\tup{\lbl,w}$ consisting of the predicted next instruction address~$\lbl$ and the length of the transaction $w$,
\item sets the program counter to $\lbl$,
\item decrements the length of the transactions in $s$,
\item increments the transaction counter $\ctr$, 
\item appends a new speculative state with  configuration $\sigma$, identifier $\ctr$, transaction's length $w$, and predicted instruction address $\lbl$, and
\item updates the branching history by appending an entry $\tup{a(\pc), \ctr, \lbl}$ modeling the prediction.
\end{inparaenum}
The rule also records the start of the speculative execution and the change of the program counter through observations.

\textsc{Se-Commit} captures a speculative transaction's commit.
It is executed whenever a speculative state's  remaining length  reaches $0$.
Application of the rule requires that the prediction made for the transaction is correct, which is checked by comparing the predicted address $\lbl$ with the one obtained by executing one step of the non-speculative semantics starting from the configuration $\sigma'$.
The rule records the transaction's commit through an observation, and it updates the branching history according to the branch decision that has been taken.

\textsc{Se-Rollback} captures a speculative transaction's rollback.
The rule checks that the prediction is incorrect (again by comparing the predicted address $\lbl$ with the one obtained from the non-speculative semantics), and it restores the configuration stored in $s$.
Rolling back a transaction also terminates the speculative execution of all the nested transactions.
This is modeled by dropping the portion $s'$ of the speculative state associated with the nested transactions.
The rule also produces observations recording the transaction's rollback and the change of the program counter, and it updates the branching history by recording the branch instruction's correct outcome. % of the branch instruction.

\para{Runs and traces}
Runs and traces are defined analogously to the non-speculative case:
Given a program $p$ and an oracle~$\bp$, we denote by~$\specTraces{p}{\bp}$ the set of all possible runs of the speculative semantics.
By $\specEval{p}{\bp}{\sigma}$ we denote the trace~$\tau$ such that there is a final configuration~$\sigma'$ for which  $\tup{\sigma, \tau, \sigma'} \in \specTraces{p}{\bp}$.

\begin{example}\label{example:speculative-semantics:spectrev1}

For illustrating the speculative semantics, we execute the program from Ex.~\ref{example:language:spectrev1} with the  oracle from Ex.~\ref{example:speculative-semantics:btfnt-predictor} and a  configuration $\tup{0,\tup{m,a}, \emptysequence, \emptysequence}$ where $a(\mathtt{y})\geq a(\mathtt{size})$.\looseness=-1
First, the rule \textsc{Se-NoJump} is applied to execute the assignment  $\passign{x}{\mathtt{y} < \mathtt{size}}$.
Then, the branch instruction $\pjz{x}{\bot}$ is reached and so rule \textsc{Se-Jump} applies.
This produces the observations $\startObs{0}$, modeling the beginning of a speculative transaction with id~0, and $\pcObs{2}$, representing the program counter's change.
Next, rule \textsc{Se-NoJump} applies three times to execute the instructions 2--5, thereby producing the observations $\loadObs{v_1}$ and $\loadObs{v_2}$ that record the memory accesses.
Finally, rule \textsc{Se-Rollback} applies, which terminates the speculative transaction and rolls back its effects.
This rule produces the observations  $\rollbackObs{0}$ and $\pcObs{\bot}$.
Thus, executing the program produces the  trace:
$$\tau := \startObs{0} \concat \pcObs{2} \concat \loadObs{v_1} \concat \loadObs{v_2} \concat \rollbackObs{0} \concat \pcObs{\bot}\ ,$$ where $v_1 = a(\mathtt{A}) + a(\mathtt{y})$ and $v_2 = a(\mathtt{B}) + m(v_1) * 512$.
\end{example}

\begin{figure*}[!h]
{
\small
\begin{mathpar}
\inferrule[Se-NoBranch]
{
p(\sigma(\pc)) \neq \pjz{x}{\lbl}\\
\sigma \eval{p}{\tau} \sigma'\\
 \exhausted{s} \\\\
 s' = {
 	\begin{cases}
 		\decrement{s} & \text{if}\ p(\sigma(\pc)) \neq \pbarrier\\
		\zeroes{s}  & 	\text{otherwise}
 	\end{cases}
 }
}
{
\tup{\ctr, \sigma, s, h} \speval{p}{\bp}{\tau} \tup{\ctr, \sigma', s', h}
}	

\inferrule[Se-Branch]
{
p(\sigma(\pc)) = \pjz{x}{\lbl'}\\
\bp(p,h,\sigma(\pc)) = \tup{\lbl,w}\\
\sigma = \tup{m ,a} \\\\
 \exhausted{s} \\
 s' = \decrement{s} \concat \tup{\sigma,\ctr, w, \lbl}\\
 \id = \ctr
}
{
\tup{\ctr, \sigma, s, h} \speval{p}{\bp}{\startObs{\id} \concat \pcObs{\lbl}} \tup{ \ctr +1 , \tup{ m, a[\pc \mapsto \lbl]}, s', h \concat \tup{a(\pc), \id, \lbl}}
}

\inferrule[Se-Commit]
{
\sigma' \eval{p}{\tau} \tup{ m, a}\\
\lbl = a(\pc)\\
\exhausted{s'}\\\\
h' =  h \concat \tup{\sigma'(\pc), \id, a(\pc)}
}
{
\tup{\ctr, \sigma, s \concat \tup{\sigma',\id, 0, \lbl} \concat s', h} \speval{p}{\bp}{\commitObs{\id}} \tup{\ctr, \sigma, s \concat s', h'}
}

\inferrule[Se-Rollback]
{
\sigma' \eval{p}{\tau} \tup{ m, a}\\
\lbl \neq a(\pc)\\
\exhausted{s'}\\\\
h' =  h \concat \tup{\sigma'(\pc), \id, a(\pc)}
}
{
\tup{\ctr, \sigma, s \concat \tup{\sigma',\id, 0, \lbl} \concat s', h} \speval{p}{\bp}{\rollbackObs{\id} \concat \pcObs{a(\pc)}} \tup{\ctr, \tup{m,a}, s,h'}
}
\end{mathpar}
}
\caption{Speculative execution for \lang{} for a program $p$ and a prediction oracle $\bp$}\label{figure:language:speculative-execution}
\end{figure*}

\subsection{Speculative and Non-speculative Semantics}

We conclude this section by connecting  the speculative and non-speculative semantics. 
For this, we introduce two projections of speculative traces $\tau$:
\begin{asparaitem}
	\item  the {\em non-speculative projection} $\nspecProject{\tau}$ is the trace obtained by removing from $\tau$ (1) all substrings that correspond to rolled-back transactions, i.e. all substrings $\startObs{\id} \concat \tau' \concat \rollbackObs{\id}$, and (2) all extended observations.
	\item the {\em speculative projection} $\specProject{\tau}$ is the trace   produced by rolled-back transactions, i.e. the complement of $\nspecProject{\tau}$. 
\end{asparaitem}
We lift projections $\specProject{}$ and $\nspecProject{}$ to sets of  runs  in the natural way.
Then, a program's non-speculative behavior can be obtained from its speculative behavior by dropping all speculative observations, i.e., by applying $\nspecProject{\tau}$ to all of its runs~$\tau$:
\begin{restatable}{prop}{speculativeAndNonSpeculative}
\label{proposition:speculative-and-non-speculative}
Let $p$ be a program and $\bp$ be a prediction oracle.
Then, $\nspecTraces{p} = \nspecProject{\specTraces{p}{\bp}}$.
\end{restatable}

\onlyTechReport{The proof of Proposition~\ref{proposition:speculative-and-non-speculative} is given in~\techReportAppendix{appendix:speculative-and-non-speculative}.}

%% file: security.tex
\section{Speculative Non-interference}\label{sec:security}

This section introduces \emph{speculative non-interference} (SNI), a semantic notion of security characterizing those information leaks that are introduced by speculative execution.

\subsection{Security policies}\label{sec:secpolicy}
Speculative non-interference is parametric in a policy that specifies which parts of the configuration are known or controlled by an adversary, i.e., ``public'' or ``low'' data.

Formally, a security policy $\policy$ is a finite subset of $\Var \cup \Nat$ specifying the low register identifiers and memory addresses.
Two configurations $\sigma, \sigma' \in \Conf$ are {\em indistinguishable} with respect to a policy $\policy$, written $\sigma \indist{\policy} \sigma'$, iff they agree on all registers and memory locations in $\policy$.

\begin{example}\label{example:security-condition:policy}
A policy $\policy$ for the program from Example~\ref{example:language:spectrev1} may state that the content of the registers $\mathtt{y}$, $\mathtt{size}$, $\mathtt{A}$, and $\mathtt{B}$ is non-sensitive, i.e., $\policy = \{\mathtt{y}, \mathtt{size},\mathtt{A}, \mathtt{B}\}$.
\end{example}
Policies need not be manually specified but can in principle be inferred from the context in which a piece of code executes, e.g., whether a variable is reachable from public input or not.

\subsection{Speculative non-interference}\label{sec:specnonint}

Speculative non-interference requires that executing a program under the speculative semantics  does not leak more information than executing the same program under the non-speculative semantics. 
Formally, whenever two indistinguishable configurations produce the same non-speculative traces, then they must also produce the same speculative traces.

\begin{definition}\label{def:gsni}
A program $p$ satisfies \textit{speculative non-interference} for a prediction oracle $\bp$ and a security policy $\policy$ iff 
for all initial configurations $\sigma,\sigma' \in \Init$,
if $\sigma \indist{\policy} \sigma'$ and $\nspecEval{p}{\sigma} = \nspecEval{p}{\sigma'}$, then $\specEval{p}{\bp}{\sigma} = \specEval{p}{\bp}{\sigma'}$.
\end{definition}

Speculative non-interference is a variant of non-interference. 
While  non-interference compares what is leaked by a program with a policy specifying the allowed leaks, speculative non-interference compares the program leakage under two semantics, the non-speculative and the speculative one.
The security policy and the non-speculative semantics, together, specify what the program may leak under the speculative semantics.\footnote{Conceptually, the non-speculative semantics can be seen as a declassification assertion for the speculative semantics~\cite{sabelfeld2009declassification}.}

\begin{example}
The program $p$ from Example~\ref{example:language:spectrev1} does not satisfy speculative non-interference for the BTFNT oracle from Example~\ref{example:speculative-semantics:btfnt-predictor} and the policy $\policy$ from Example~\ref{example:security-condition:policy}.
Consider two initial configurations $\sigma:= \tup{m,a}, \sigma':=\tup{m',a'}$ that agree on the values of $\mathtt{y}$, $\mathtt{size}$, $\mathtt{A}$, and $\mathtt{B}$ but disagree on the value of $\mathtt{B}[\mathtt{A}[\mathtt{y}] * 512]$.
Say, for instance, that $m(a(\mathtt{A}) + a(\mathtt{y})) = 0$ and $ m'(a'(\mathtt{A}) + a'(\mathtt{y})) = 1$.
Additionally, assume that $\mathtt{y} \geq \mathtt{size}$.

Executing the program under the non-speculative semantics produces the trace $\pcObs{\bot}$ when starting from $\sigma$ and $\sigma'$.
Moreover, the two initial configurations are indistinguishable with respect to the policy $\policy$.
However, executing $p$ under the speculative semantics produces two distinct traces $\tau = \startObs{0} \concat \pcObs{3} \concat \loadObs{v_1} \concat \loadObs{(a'(\mathtt{B})+0)} \concat \rollbackObs{0} \concat \pcObs{\bot}$ and $\tau' = \startObs{0} \concat \pcObs{3} \concat \loadObs{v_1} \concat \loadObs{(a'(\mathtt{B})+1)} \concat \rollbackObs{0} \concat \pcObs{\bot}$, where $v_1 = a(\mathtt{A}) + a(\mathtt{y}) = a'(\mathtt{A}) + a'(\mathtt{y})$.
Therefore, $p$ does not satisfy speculative non-interference.
\end{example}

\subsection{Always-mispredict speculative semantics}\label{sec:automation:always-mispred}
The speculative semantics and SNI are parametric in the prediction oracle $\bp$. 
Often, it is desirable obtaining guarantees w.r.t. {\em any} prediction oracle, since branch prediction models in modern CPUs are unavailable and as different CPUs employ different predictors.
To this end, we introduce a variant of the speculative semantics that facilitates such an analysis.

Intuitively, leakage due to speculative execution is maximized under a branch predictor that mispredicts every branch.
This intuition holds true unless speculative transactions are nested, where a correct prediction of a nested branch sometimes yields more leakage than a misprediction.

\begin{example}\label{ex:mispred}
Consider the following variation of the \spectre{} v1 example~\cite{Kocher2018spectre} from Figure~\ref{figure:spectrev1:c-code}, and assume that the function \verb!benign()! runs for longer than the speculative window and does not leak any information.
\begin{lstlisting}[style=Cstyle]
if (y < size)
	if (y-1 < size)
		benign();
	temp &= B[A[y] * 512];
\end{lstlisting}
Then, under a branch predictor that mispredicts every branch, the speculative transaction corresponding to the outer branch will be rolled back before reaching line 4.
On the other hand, given a correct prediction of the inner branch, line~4 would be reached and a speculative leak would be present.
\end{example}

A simple but inefficient approach to deal with this challenge would be to consider both cases, correct and incorrect prediction, upon every branch.
This, however, would result in an exponential explosion of the number of paths to consider.

To avoid this, we  introduce the {\em always-mispredict semantics} that differs from the speculative semantics in three key ways:
\begin{asparaenum}[(1)]
	\item It mispredicts every branch, hence its name. In particular, it is not parametric in the prediction oracle. 
	\item It initializes the length of every {\em non-nested}  transaction to~$w$, and the length of every {\em nested}  transaction to the remaining length of its enclosing  transaction decremented by $1$.\looseness=-1
	\item Upon executing instructions, only the remaining length of the innermost transaction is decremented.
\end{asparaenum}
The consequence of these modifications is that nested transactions do not reduce the number of steps that the semantics may explore the correct path for, after the nested transactions have been rolled back.
In Example~\ref{ex:mispred}, after rolling back the nested speculative transaction, the outer transaction continues as if the nested branch had been correctly predicted in the first place, and thus the speculative leak in line~4 is reached.

Modifications (1)-(3) are formally captured in the three rules \textsc{Am-NoBranch}, \textsc{Am-Branch}, and \textsc{Am-Rollback} given  in Appendix~\ref{appendix:always-mispredict}.
Similarly to $\specEval{p}{\bp}{\sigma}$, we denote by $\amEval{p}{w}{\sigma}$ the trace of observations obtained by executing the program $p$, starting from initial configuration $\sigma$ according to the always-mispredict evaluation relation with speculative window $w$.

Theorem~\ref{theorem:always-mispredict} \onlyTechReport{(proved in Appendix~\ref{appendix:always-mispredict-worst-case}) } states that checking SNI w.r.t. the always- mispredict semantics is sufficient to obtain security guarantees w.r.t. all prediction oracles.

\begin{restatable}{thm}{alwaysMispredict}
\label{theorem:always-mispredict}
A program $p$ satisfies SNI for a security policy $\policy$ and all prediction oracles~$\bp$ with speculative window at most $w$ iff 
for all initial configurations $\sigma,\sigma' \in \Init$,
if $\sigma \indist{\policy} \sigma'$ and $\nspecEval{p}{\sigma} = \nspecEval{p}{\sigma'}$, then $\amEval{p}{w}{\sigma} = \amEval{p}{w}{\sigma'}$.
\end{restatable}

In our case studies in Sections~\ref{sec:casestudies} and \ref{sec:xen}, we use $w = 200$. 
This is motivated by typical sizes of the reorder buffer~\cite{IntelReorderBuffer}, which limits the lengths of speculative transactions in modern microarchitectures.

%% file: automation.tex
\section{Detecting speculative information flows}

We now present \tool{}, an approach to detect speculative leaks, or to prove their absence.
\tool{} symbolically executes the program~$p$ under analysis to derive a concise representation of $p$'s behavior as a set of symbolic traces.
It analyzes each symbolic trace using an SMT solver to detect possible speculative leaks through memory accesses or control-flow instructions.
If neither memory nor control leaks are detected, \tool{} reports the program as secure.

\subsection{Symbolically executing \lang{} programs}\label{sec:automation:symbolic-semantics}

We symbolically execute programs w.r.t. the {\em always mispredict} semantics, which enables us to derive security guarantees that hold for arbitrary prediction oracles, see Theorem~\ref{theorem:always-mispredict}. 
Our symbolic execution engine relies on the following components:
\begin{asparaitem}
	\item A \textit{symbolic expression} $\sexpr$ is a concrete value $n \in \Val$, a symbolic value $s \in \SymVal$, an if-then-else expression $\ite{\sexpr}{\sexpr'}{\sexpr''}$, or the application of unary or binary operators to symbolic expressions.

	\item A \textit{symbolic memory} is a term in the standard theory of arrays~\cite{bradley2007calculus}.
	A memory update  $\symWrite{sm}{\sexpr}{\sexpr'}$ updates the symbolic memory $sm$ by assigning the symbolic value $\sexpr'$ to the symbolic address $\sexpr$.
	We extend symbolic expressions with memory reads $\symRead{sm}{\sexpr}$, which retrieve the value of the symbolic address $\sexpr$ from the symbolic memory $sm$.
\item A {\em symbolic trace} $\tau $ is a sequence of symbolic observations of the form $\loadObs{\sexpr}$ or $\storeObs{\sexpr}$, symbolic branching conditions of the form $\symPcObs{\sexpr}$, and transaction-related observations of the form $\startObs{n}$ and $\rollbackObs{n}$, for natural numbers $n$ and symbolic expressions $\sexpr$. 
\item The path condition $\pathCond{\tau}\!\!=\!\!\bigwedge_{\symPcObs{\sexpr} \in \tau} \sexpr$ of trace~$\tau$ is the conjunction of all symbolic branching conditions in $\tau$.\looseness=-1
\item The symbolic execution derives symbolic runs $\tup{\sigma, \tau, \sigma'}$, consisting of symbolic configurations $\sigma,\sigma'$ and a symbolic trace $\tau$. The set of all symbolic runs forms the symbolic semantics, which  we denote by $\symbTraces{p}{w}$. The derivation rules are fairly standard and are given in Appendix~\ref{appendix:symbolic-semantics}.
\item The value of  an expression $\sexpr$ depends on a \textit{valuation} $\mu : \SymVal \to \Val$  mapping symbolic values to concrete ones.
The evaluation $\mu(\sexpr)$  of $\sexpr$ under $\mu$ is standard and formalized in Appendix~\ref{appendix:symbolic-semantics}.
\item A symbolic expression $\sexpr$ is \textit{satisfiable}, written $\mu \models \sexpr$, if there is a valuation $\mu$ such that $\symExprEval{\sexpr}{\mu} \neq 0$.
 Every valuation that satisfies a symbolic run's path condition maps the run to a concrete run. We denote by
$\gamma(\tup{\sigma, \tau, \sigma'})$ the set $\{  \tup{\mu(\sigma), \mu(\tau), \mu(\sigma')} \mid  \mu \models \pathCond{\tau} \}$  of  $\tup{\sigma, \tau, \sigma'}$'s  concretizations, and we lift it to $\symbTraces{p}{w}$. The concretization of the symbolic runs  yields the set of all concrete runs:
\end{asparaitem}

\begin{restatable}{prop}{alwaysMispredSymbolic}\label{proposition:symbolic-execution}
	Let $p$ be a program and $w \in \Nat$ be a speculative window.
	Then, $\amTraces{p}{w} = \gamma(\symbTraces{p}{w})$.
\end{restatable}

The proof of Proposition~\ref{proposition:symbolic-execution} is given in~\techReportAppendix{appendix:concorete-and-symbolic}.

\begin{example}\label{example:symbolic-semantics}
		Executing the program from Example~\ref{example:language:spectrev1} under the symbolic speculative semantics with speculative window~$2$ yields the following two symbolic traces:
		$\tau_1 := \symPcObs{{\mathtt{y}} < {\mathtt{size}}} \concat \startObs{0} \concat \pcObs{2} \concat \pcObs{10} \concat \rollbackObs{0} \concat \pcObs{3} \concat \loadObs{{\mathtt{A}} + {\mathtt{y}}}  \concat \loadObs{{\mathtt{B}} + \symRead{sm}{ ({\mathtt{A}} + {\mathtt{y}}) }*512}$, and $\tau_2 := \symPcObs{{\mathtt{y}} \geq {\mathtt{size}}} \concat  \startObs{0} \concat \pcObs{3} \concat \loadObs{{\mathtt{A}} + {\mathtt{y}}}  \concat \loadObs{{\mathtt{B}} + \symRead{sm}{ ({\mathtt{A}} + {\mathtt{y}}) }*512} \concat \rollbackObs{0} \concat \pcObs{2} \concat \pcObs{10}$.
\end{example}

\subsection{Checking speculative non-interference}

\begin{algorithm}
    \caption{\tool{}}
    \label{algorithm:tool}
    \begin{algorithmic}[1]
	\Require A program $p$, a security policy $\policy$, a speculative window  $w \in \Nat$.
	\Ensure \textsc{Secure} if $p$ satisfies speculative non-interference with respect to the policy $\policy$; \textsc{Insecure} otherwise
		\Statex{}
	\Procedure{\tool}{$p,\policy,w$}
		\For{each symbolic run $\tup{\sigma,\tau,\sigma'} \in \symbTraces{p}{w}$}
			\If{$\memcheck(\tau,P)  \vee \pccheck(\tau,P)$}
				\State{\Return{\textsc{Insecure}}}
			\EndIf
		\EndFor
		\State{\Return{\textsc{Secure}}}
        \EndProcedure
    \Statex{}
    \Procedure{$\memcheck$}{$\tau, \policy$}
    		\State{$\psi \gets \pathCond{\tau}_{1 \wedge 2} \wedge \policyEqv{\policy} \wedge $}
    		\Statex{$\qquad \qquad \cstrs{\nspecProject{\tau}} \wedge \neg \cstrs{\specProject{\tau}}$}
			\State{\Return{$\textsc{Satisfiable}(\psi)$}}

    \EndProcedure
    \Statex{}
    \Procedure{$\pccheck$}{$\tau, \policy$}
    		\For{each prefix $\nu \concat \symPcObs{\mathit{se}}$ of $\specProject{\tau}$}
				\State{$\psi \gets \pathCond{\nspecProject{\tau}\concat \nu}_{1 \wedge 2} \wedge \policyEqv{\policy} \wedge$}
				\Statex{$\qquad \qquad \quad \cstrs{\nspecProject{\tau}} \wedge \neg \sameSymPc{se}$}
				\If{$\textsc{Satisfiable}(\psi)$}
					\State{\Return{$\top$}}
				\EndIf
			\EndFor
    		\State{\Return{$\bot$}}
    \EndProcedure
    \end{algorithmic}
\end{algorithm}

\tool{} is given in Algorithm~\ref{algorithm:tool}.
It relies on two procedures: \memcheck{} and \pccheck{}, to detect leaks resulting from memory  and  control-flow instructions, respectively.
We start by discussing the \tool{} algorithm and next explain the \memcheck{} and \pccheck{} procedures.

\tool{} takes as input a program $p$, a policy $\policy$ specifying the non-sensitive information, and a speculative window $w$.
The algorithm iterates over all symbolic runs produced by the symbolic always-mispredict speculative semantics (lines 2-4).
For each  run $\tup{\sigma,\tau,\sigma'}$, the algorithm checks whether $\tau$ speculatively leaks information through memory accesses or control-flow instructions.
If this is the case, then \tool{} has found a witness of a speculative leak and it reports $p$ as \textsc{Insecure}.
If none of the traces contain speculative leaks, the algorithms terminates returning \textsc{Secure} (line 5).

\para{Detecting leaks caused by memory accesses}
The procedure \memcheck{} takes as input a  trace~$\tau$ and a policy $\policy$, and it determines whether $\tau$ leaks information through symbolic $\loadObs{}$ and $\storeObs{}$ observations. 
The check is expressed as a satisfiability check of a  constraint $\psi$. The construction of~$\psi$ is inspired by self-composition~\cite{barthe2004secure}, which reduces reasoning about {\em pairs} of program runs to reasoning about single runs by replacing each symbolic variable $x$ with two copies $x_1$ and~$x_2$. We lift the subscript notation to symbolic expressions.

The constraint $\psi$ is the conjunction of four formulas:
\begin{asparaitem}
\item $\pathCond{\tau}_{1 \wedge 2}$ stands for $\pathCond{\tau}_{1} \wedge \pathCond{\tau}_{2}$, which ensures that both runs follow the path associated with $\tau$.\looseness=-1
\item $\policyEqv{\policy}$ introduces constraints $x_1 = x_2$ for each register $x\in \policy$ and $\symRead{sm_1}{n} = \symRead{sm_2}{n}$ for each memory location $n \in P$, which ensure that both runs agree on all non-sensitive inputs.
\item $\cstrs{\nspecProject{\tau}}$ introduces a constraint $\sexpr_1 = \sexpr_2$ for each $\loadObs{\sexpr}$ or $\storeObs{\sexpr}$ in $\nspecProject{\tau}$, which ensures that the non-speculative observations associated with memory accesses are the same in both runs. 
\item $\neg \cstrs{\specProject{\tau}}$  ensures that speculative observations associated with memory accesses differ among the two runs. 
\end{asparaitem}

If $\psi$ is satisfiable, there are two $\policy$-indistinguishable configurations that produce the same non-speculative traces (since $\pathCond{\tau}_{1 \wedge 2} \wedge \policyEqv{\policy} \wedge \cstrs{\nspecProject{\tau}}$ is satisfied) and whose speculative traces differ in a memory access observation (since $\neg \cstrs{\specProject{\tau}}$ is satisfied), i.e. a violation of SNI.

\para{Detecting leaks caused by control-flow instructions}
To detect leaks caused by control-flow instructions, \pccheck{} checks whether there are two traces in $\tau$'s concretization that agree on the outcomes of all non-speculative branch and jump instructions, while differing in the outcome of at least one speculatively executed branch or jump instruction.

In addition to $\pathCond{\tau}$, $\cstrs{\tau}$, and $\policyEqv{\policy}$, the procedure relies on the  function $\sameSymPc{\sexpr}$ that introduces the constraint $\sexpr_1 \leftrightarrow \sexpr_2$ ensuring that  $\sexpr$ is satisfied in one concretization iff it is satisfied in the other.

\pccheck{} checks, for each prefix $\nu \concat \symPcObs{\sexpr}$ in $\tau$'s speculative projection $\specProject{\tau}$,  the satisfiability of the conjunction of $\pathCond{\nspecProject{\tau} \concat \nu}_{1 \wedge 2}$, $\policyEqv{\policy}$, $\cstrs{\nspecProject{\tau}}$, and $\neg \sameSymPc{\sexpr}$.
Whenever the formula is satisfiable, there are two $\policy$-indistinguishable configurations that produce the same non-speculative traces, but whose speculative traces differ on program counter observations, i.e. a violation of SNI.

\begin{example}
	Consider  the  trace $\tau_2 := \symPcObs{{\mathtt{y}} \geq {\mathtt{size}}} \concat  \startObs{0} \concat \pcObs{3} \concat \loadObs{{\mathtt{A}} + {\mathtt{y}}}  \concat \loadObs{{\mathtt{B}} + \symRead{sm}{  ({\mathtt{A}} + {\mathtt{y}})  }*512 } \concat \rollbackObs{0} \concat \pcObs{2} \concat \pcObs{10}$ from Example~\ref{example:symbolic-semantics}. \memcheck{} detects a leak caused by the observation $\loadObs{{\mathtt{B}} + \symRead{sm}{  ({\mathtt{A}} + {\mathtt{y}})  }*512 }$.
	Specifically, it detects that there are distinct symbolic valuations that agree on the non-speculative observations but disagree on the value of $\loadObs{{\mathtt{B}} + \symRead{sm}{  ({\mathtt{A}} + {\mathtt{y}})  }*512 }$. That is, the observation  depends on sensitive information that is not disclosed by $\tau_2$'s non-speculative projection.	\end{example}

\para{Soundness and completeness}
Theorem~\ref{theorem:soundness-and-completeness} states that \tool{} deems secure only speculatively non-interferent programs, and all detected leaks are actual violations of SNI.\looseness=-1

\begin{restatable}{thm}{spectectorSoundnessAndCompleteness}
\label{theorem:soundness-and-completeness}
If $\tool{}(p,\policy,w)$ terminates, then  $\tool{}(p,\policy,w) = \textsc{Secure}$ iff  the program $p$ satisfies speculative non-interference w.r.t. the  policy~$\policy$ and all prediction oracles~$\bp$ with speculative window at most~$w$.
\end{restatable}

The theorem follows from the soundness and completeness of the always-mispredict semantics w.r.t. prediction oracles (Theorem~\ref{theorem:always-mispredict}) and of the symbolic semantics w.r.t. to the always-mispredict semantics (Proposition~\ref{proposition:symbolic-execution}).
\onlyTechReport{The proof of Theorem~\ref{theorem:soundness-and-completeness} is given in Appendix~\ref{appendix:proofs}.}

\section{Tool Implementation}\label{sect:tool}

We implement our approach in our tool \tool{},
which is available at \texttt{\url{https://spectector.github.io}}.
The tool, which is implemented on top of the \ciao{} logic
programming system~\cite{hermenegildo11:ciao-design-tplp-short}, consists of three components: a front end that translates x86 assembly programs into \lang{}, a core engine implementing Algorithm~\ref{algorithm:tool}, and a back end  handling SMT queries.

\para{x86 front end}
The front end translates AT\&T/GAS and Intel-style assembly files into \lang{}. 
It currently supports over 120 instructions: data movement instructions
($\kywd{mov}$, etc.), logical, arithmetic, and comparison instructions
($\kywd{xor}$, $\kywd{add}$, $\kywd{cmp}$, etc.),
branching and jumping instructions
($\kywd{jae}$, $\kywd{jmp}$, etc.), conditional moves
($\kywd{cmovae}$, etc.), stack manipulation ($\kywd{push}$,
$\kywd{pop}$, etc.), and function calls\footnote{We model the so-called ``near calls'', where the callee is in the same code segment as the caller.}
($\kywd{call}$, $\kywd{ret}$).

It currently does not support privileged x86 instructions, e.g., for handling model specific registers and virtual memory.
Further it does not support sub-registers (like $\mathtt{eax}$,
$\mathtt{ah}$, and $\mathtt{al}$) and unaligned memory accesses, i.e.,
we assume that only 64-bit words are read/written at each address
without overlaps.
Finally, the translation currently maps symbolic address names
to \lang{} instruction addresses, limiting arithmetic on code addresses.

\para{Core engine}
The core engine implements Algorithm~\ref{algorithm:tool}.
It relies on a concolic approach to implement symbolic execution that performs
a depth-first exploration of the symbolic runs. 
Starting from a concrete initial configuration, the engine executes the program under the always-mispredict speculative semantics while keeping track of the symbolic configuration and path condition.
It discovers new runs by iteratively negating the last (not previously negated) conjunct in the path condition until it finds a new initial configuration, which is then used to re-execute the program concolically.
In our current implementation, indirect jumps are \emph{not} included in the path conditions, and thus new symbolic runs and corresponding inputs are only discovered based on negated branch conditions.\footnote{We plan to remove this limitation in a future release of our tool.}
This process is interleaved with the \memcheck{} and \pccheck{} checks and iterates until a leak is found or all paths have been explored.

\para{SMT back end}
The Z3 SMT solver~\cite{z3} acts as a back end for checking satisfiability and finding models of symbolic expressions using the \textsc{Bitvector} and \textsc{Array} theories, which are used to model registers and memory. 
The implementation currently does not rely on incremental solving, since it was less efficient than one-shot solving for the selected theories.

%% file: casestudies.tex
\section{Case study: Compiler countermeasures}\label{sec:casestudies}

\newcommand{\N}{\circ}
\renewcommand{\P}{\fcirc}
\newcommand{\Q}{?}
\newcommand{\clm}{0.2cm}

\begin{figure*}
\centering
\begin{tabular}
{c  c c c c c c  c c c c  c c c c c c  }
\toprule
\multirow{3}{*}{Ex.}	&	\multicolumn{6}{c}{\textsc{Vcc}}  &  \multicolumn{4}{c }{\icc} & \multicolumn{6}{c }{\clang{}}	\\ 
 \cmidrule(lr){2-7}
 \cmidrule(lr){8-11}
 \cmidrule(lr){12-17}
	& \multicolumn{2}{c}{\unp}	& \multicolumn{2}{c }{\fen{} 19.15}		& \multicolumn{2}{c }{\fen{} 19.20}	&  \multicolumn{2}{c }{\unp}	& \multicolumn{2}{c }{\fen} &  \multicolumn{2}{c }{\unp}	& \multicolumn{2}{c }{\fen} & \multicolumn{2}{c }{\slh} \\
 \cmidrule(lr){2-3}   \cmidrule(lr){4-5}  \cmidrule(lr){6-7}
 \cmidrule(lr){8-9}  \cmidrule(lr){10-11}
 \cmidrule(lr){12-13}  \cmidrule(lr){14-15}
 \cmidrule(lr){16-17} 
	&	\unopt 	&	 	\opt & \unopt 	&	 	\opt	 	&	\unopt		& 		\opt		&	\unopt		& 	\opt			& 	\unopt 	&	 	\opt	 	&	\unopt 	&	 	\opt	  & 	\unopt 	&	 	\opt	 	& 	\unopt 	&	 	\opt	\\ 
 \midrule
01	&	$\N$	&	$\N$	&	$\P$	&	$\P$	& 	$\P$	&	$\P$ &	$\N$	&	$\N$	&	$\P$	&	$\P$	&	$\N$	&	$\N$	&	$\P$	&	$\P$	&	$\P$	&	$\P$ \\ 
02	&	$\N$	&	$\N$	&	$\P$	&	$\P$	& 	$\P$	&	$\P$	 &	$\N$	&	$\N$	&	$\P$	&	$\P$	&	$\N$	&	$\N$	&	$\P$	&	$\P$	&	$\P$	&	$\P$ \\ 
03	&	$\N$ 	&	$\N$	&	$\P$	&	$\N$	& 	$\P$	&	$\P$	 &	$\N$	&	$\N$	&	$\P$	&	$\P$	&	$\N$	&	$\N$	&	$\P$	&	$\P$	&	$\P$	&	$\P$ \\ 
04	&	$\N$	&	$\N$	&	$\N$	&	$\N$	& 	$\P$	&	$\P$ &	$\N$	&	$\N$	&	$\P$	&	$\P$	&	$\N$	&	$\N$	&	$\P$	&	$\P$	&	$\P$	&	$\P$ \\ 
05	&	$\N$	&	$\N$	&	$\P$	&	$\N$	& 	$\P$	&	$\N$	 &	$\N$	&	$\N$	&	$\P$	&	$\P$	&	$\N$	&	$\N$	&	$\P$	&	$\P$	&	$\P$	&	$\P$ \\ 
06	&	$\N$	&	$\N$	&	$\N$	&	$\N$	& 	$\N$	&	$\N$ &	$\N$	&	$\N$	&	$\P$	&	$\P$	&	$\N$	&	$\N$	&	$\P$	&	$\P$	&	$\P$	&	$\P$ \\ 
07	&	$\N$	&	$\N$	&	$\N$	&	$\N$	& 	$\N$	&	$\N$ &	$\N$	&	$\N$	&	$\P$	&	$\P$	&	$\N$	&	$\N$	&	$\P$	&	$\P$	&	$\P$	&	$\P$ \\ 
08	&	$\N$	&	$\P$	&	$\N$	&	$\P$	& 	$\N$	&	$\P$ &	$\N$	&	$\P$	&	$\P$	&	$\P$	&	$\N$	&	$\P$	&	$\P$	&	$\P$	&	$\P$	&	$\P$ \\ 
09	&	$\N$	&	$\N$	&	$\N$	&	$\N$	& 	$\N$	&	$\N$ &	$\N$	&	$\N$	&	$\P$	&	$\P$	&	$\N$	&	$\N$	&	$\P$	&	$\P$	&	$\P$	&	$\P$ \\ 
10	&	$\N$	&	$\N$	&	$\N$	&	$\N$	&	$\N$	&	$\N$	 &	$\N$	&	$\N$	&	$\P$	&	$\P$	&	$\N$	&	$\N$	&	$\P$	&	$\P$	&	$\P$	&	$\N$ \\ 
11	&	$\N$	&	$\N$	&	$\N$	&	$\N$	& 	$\N$	&	$\N$ &	$\N$	&	$\N$	&	$\P$	&	$\P$	&	$\N$	&	$\N$	&	$\P$	&	$\P$	&	$\P$	&	$\P$ \\ 
12	&	$\N$	&	$\N$	&	$\N$	&	$\N$	& 	$\P$	&	$\P$	 &	$\N$	&	$\N$	&	$\P$	&	$\P$	&	$\N$	&	$\N$	&	$\P$	&	$\P$	&	$\P$	&	$\P$ \\ 
13	&	$\N$	&	$\N$	&	$\N$	&	$\N$	& 	$\N$	&	$\N$	 &	$\N$	&	$\N$	&	$\P$	&	$\P$	&	$\N$	&	$\N$	&	$\P$	&	$\P$	&	$\P$	&	$\P$ \\ 
14	&	$\N$	&	$\N$	&	$\N$	&	$\N$	& 	$\P$	&	$\P$	 &	$\N$	&	$\N$	&	$\P$	&	$\P$	&	$\N$	&	$\N$	&	$\P$	&	$\P$	&	$\P$	&	$\P$ \\ 
15	&	$\N$	&	$\N$	&	$\N$	&	$\N$	& 	$\N$	&	$\N$	 &	$\N$	&	$\N$	&	$\P$	&	$\P$	&	$\N$	&	$\N$	&	$\P$	&	$\P$	&	$\N$	&	$\P$ \\ 
\bottomrule
\end{tabular}
\caption{Analysis of Kocher's examples~\cite{Kocher2018examples},
  compiled with different compilers and options.
For each of the 15 examples, we analyzed the unpatched version (denoted by \unp), the version patched with speculation barriers (denoted by \fen), and the version patched using speculative load hardening (denoted by \slh).
Programs have been compiled without optimizations (\unopt{}) or with compiler optimizations (\opt) using the compilers \vcc{} (two versions), \icc{}, and \clang{}. 
$\N$ denotes that \tool{} detects a speculative leak, whereas 
$\P$ indicates that \tool{} proves the program secure.
}\label{figure:case-studies:results}
\end{figure*}

This section reports on a case study in which we apply \tool{} to analyze the security of compiler-level countermeasures against \spectre{}. 
We analyze a corpus of 240 assembly programs derived from the variants of the \spectre{} v1 vulnerability by Kocher~\cite{Kocher2018examples} using different compilers and compiler options.
This case study's goals  are: (1) to determine whether speculative non-interference realistically captures speculative leaks, and (2) to assess \tool{}'s precision.\looseness=-1

\subsection{Experimental Setup}\label{sec:case-studies:setup}

For our analysis, we rely on three state-of-the-art compilers:
Microsoft \vcc{} versions v19.15.26732.1 and v19.20.27317.96, 
 Intel \icc{} v19.0.0.117, and \clang{} v7.0.0. %, and \gcc{} v??.

We compile the programs using two different \textit{optimization levels} (\unopt{} and
\opt{}) and three \textit{mitigation levels}: 
\begin{inparaenum}[(a)]
\item \unp{}: we compile  without any \spectre{} mitigations. 
\item \fen{}: we compile with automated injection of speculation
  barriers.\footnote{Fences are supported by \clang{} with
    the flag \texttt{-x86-speculative\-load-hardening-lfence}, by \icc{}
    with \texttt{-mconditional\-branch=all-fix}, and by \vcc{} with
    \texttt{/Qspectre}.\looseness=-1}
\item \slh{}: we compile using speculative load hardening.\footnote{Speculative load hardening is supported by \clang{} with the flag \texttt{-x86-speculative-load-hardening}.}
\end{inparaenum}

Compiling each of the 15 examples from \cite{Kocher2018examples} with each of
the 3 compilers, each of the 2 optimization levels, and each of
the 2-3 mitigation levels, yields a corpus of 240 x64
assembly programs.\footnote{The resulting assembly files are available at~\texttt{\url{https://spectector.github.io}}.}
For each program, we specify a security policy that
flags as ``low'' all registers and memory locations that can either be
controlled by the adversary or can be assumed to be public. This
includes  variables \inlineASMcode{y} and \inlineASMcode{size}, and the
base addresses of the arrays $\inlineASMcode{A}$ and $\inlineASMcode{B}$
as well as the stack pointer.

\subsection{Experimental Results}

Figure~\ref{figure:case-studies:results} depicts  the results of applying \tool{} to the 240 examples.  We highlight the
following findings:

\begin{asparaitem}
\item \tool{} detects the speculative leaks in almost all unprotected
  programs, for all compilers (see the \unp{} columns). The exception
  is Example \#8, which uses a conditional expression instead of the if statement of
  Figure~\ref{figure:spectrev1:c-code}:
\begin{lstlisting}[style=Cstyle]
temp &= B[A[y<size?(y+1):0]*512];
\end{lstlisting}
At optimization level \unopt{}, this is translated to a (vulnerable) branch instruction by all compilers, and at level \opt{} to a (safe) conditional move, thus closing the leak. See Appendix~\ref{secs:example8} for the 
corresponding \clang{} assembly.
\item The \clang{} and Intel \icc{} compilers defensively insert fences after
  each branch instruction, and \tool{} can prove security for
  all cases (see the \fen{} columns for \clang{} and \icc{}).  In
  Example \#8 with options \opt{} and \fen{}, \icc{} inserts an \textbf{lfence} instruction, even
  though the baseline relies on a conditional move, see line 10 below. This \textbf{lfence} is unnecessary according to our semantics, but may close leaks on processors that speculate over conditional moves. 
  \begin{lstlisting}[style=ASMstyle]
        mov 	y, %rdi
        lea 	1(%rdi), %rdx
        mov 	size, %rax
        xor 	%rcx, %rcx
        cmp 	%rax, %rdi
        cmovb 	%rdx, %rcx
        mov 	temp, %r8b
        mov 	A(%rcx), %rsi
        shl 	$9, %rsi
        lfence
        and 	B(%rsi), %r8b
        mov 	%r8b, temp
\end{lstlisting}
\item For the \vcc{} compiler, \tool{} automatically detects all
  leaks pointed out in~\cite{Kocher2018examples} (see the \fen{} 19.15 \opt{} column for \textsc{Vcc}). Our analysis differs
  from Kocher's only on Example \#8, where  the compiler v19.15.26732.1 introduces a safe conditional move, as explained above. Moreover,
  without compiler optimizations (which is not considered
  in~\cite{Kocher2018examples}), \tool{} establishes the security of Examples \#3 and \#5 (see the \fen{} 19.15 \unopt{} column).
The latest \textsc{Vcc} compiler additionally mitigates the leaks in Examples \#4, \#12, and \#14 (see the \fen{} 19.20 column).\looseness=-1

\item \tool{} can prove the security of speculative load hardening in
  Clang (see the \slh{} column for \clang{}), except for Example \#10 with \opt{}
  and Example \#15 with \unopt{}.
\end{asparaitem}

\subsubsection*{Example 10 with Speculative Load Hardening}
Example \#10 differs from Figure~\ref{figure:spectrev1:c-code} in that it leaks sensitive
information into the microarchitectural state by conditionally reading
the content of \inlineCcode{B[0]}, depending on the value of
\inlineCcode{A[y]}.

\begin{lstlisting}[style=Cstyle]
 if (y < size) 
   if (A[y] == k)
     temp &= B[0];
\end{lstlisting}
\tool{} proves the security of the program produced with \clang{} \unopt{}, and speculative load hardening.

However, at optimization level \opt{}, \clang{} outputs the following
code that \tool{} reports as insecure. 

\begin{lstlisting}[style = ASMstyle]
 mov     size, %rdx
 mov     y, %rbx
 mov     $0, %rax
 cmp     %rbx, %rdx
 jbe     END
 cmovbe  $-1, %rax
 or      %rax, %rbx
 mov     k, %rcx
 cmp     %rcx, A(%rbx)
 jne     END
 cmovne  $-1, %rax
 mov     B, %rcx
 and     %rcx, temp
 jmp     END
\end{lstlisting}

The reason for this is that \clang{} masks only the register
\inlineASMcode{\%rbx} that contains the index of the memory access
\inlineCcode{A[y]}, cf.~lines 6--7. However, it does {\em not} mask
the value that is read from \inlineCcode{A[y]}.  As a result, the
comparison at line 9 speculatively leaks (via the jump target) whether
the content of \inlineCcode{A[}\texttt{0xFF...FF}\inlineCcode{]} is \inlineCcode{k}.
\tool{} detects this subtle leak and flags a violation of speculative
non-interference.

While this example nicely illustrates the scope of \tool{}, it is likely not a problem in practice: 
First, the leak may be mitigated by how data dependencies are handled in modern out-of-order CPUs. Specifically, the conditional move in line 6 relies on the comparison in Line 4. If executing the conditional leak effectively terminates speculation, the reported leak is spurious.  
Second, the leak can be mitigated at the OS-level by ensuring that \texttt{0xFF...FF} is not mapped in the page tables, or that the value of \inlineCcode{A[}\texttt{0xFF...FF}\inlineCcode{]} does not contain any secret~\cite{marinas}. Such contextual information can be expressed with policies (see~\ref{sec:secpolicy}) to improve the precision of the analysis.

\subsection{Performance}\label{sec:case-studies:performance}

We run all experiments on a Linux machine (kernel 4.9.0-8-amd64) with Debian 9.0,
 a Xeon Gold 6154 CPU, and 64 GB of RAM. 
We use  \ciao{}  version 1.18 and the Z3  version 4.8.4.\looseness=-1

 \tool{}  terminates within less than 30 seconds on all examples, with several examples being analyzed in about 0.1 seconds, except for Example \#5 in mode \slh{} \opt{}.  
In this exceptional case, \tool{} needs 2 minutes for proving security.
This is due to  Example \#5's complex control-flow, which leads to loops involving several branch instructions.

%% file: xen.tex
\section{Case study: Xen Project Hypervisor}\label{sec:xen}

\newcommand{\figureheight}{6cm}
\begin{figure*}
	\centering
	\subfigure[Checking non-interference with \memcheck{}]{
		\label{figure:xen:sni:memcheck}
		\adjustbox{height=\figureheight, trim={{0.04\width} {0.025\height} {.1\width} {0.1\height}},clip}		
		{\ORIGINALincludegraphics[page=1]{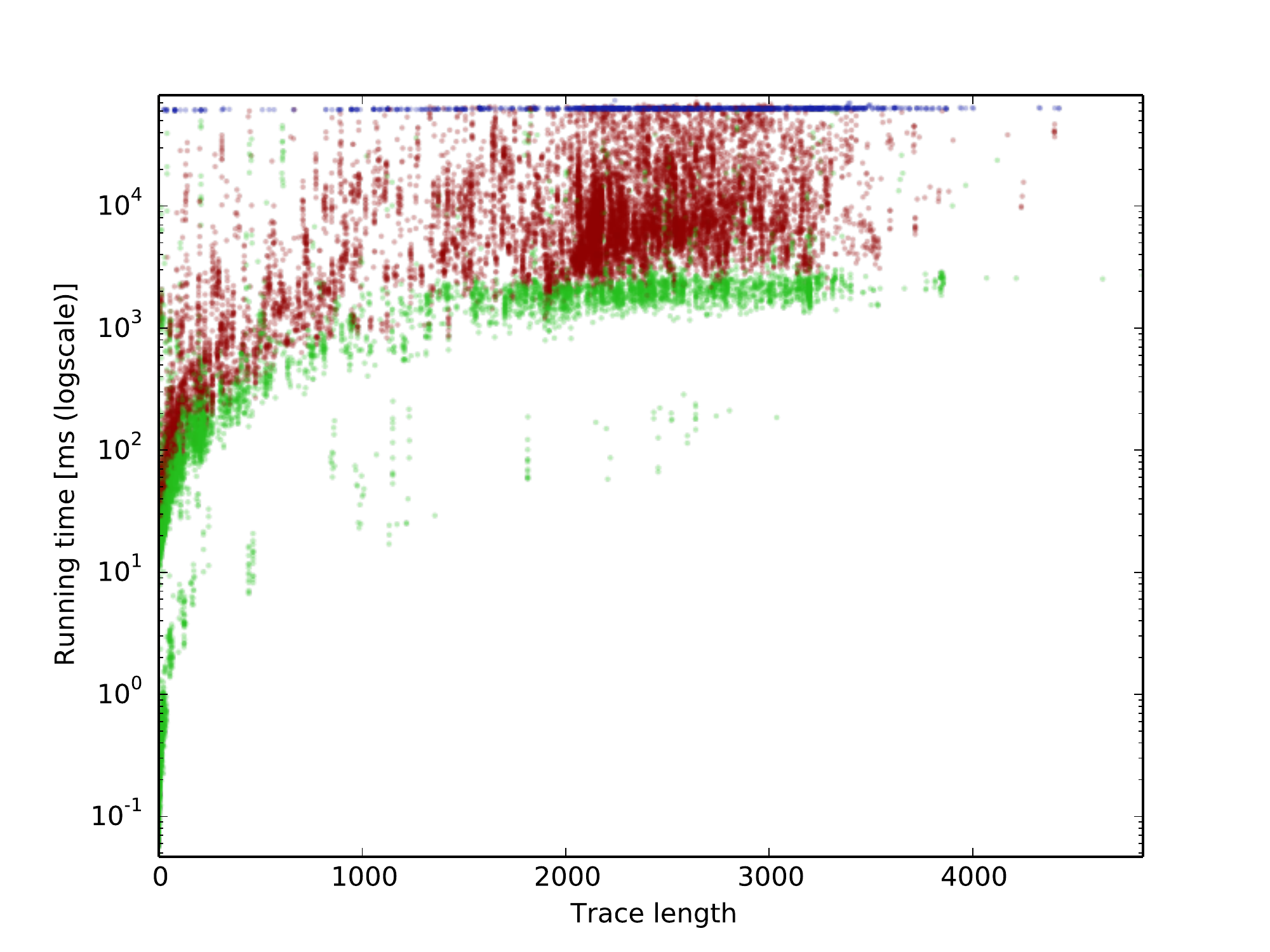}}
	}
	\qquad
	\subfigure[Checking non-interference with \pccheck{}]{
		\label{figure:xen:sni:pccheck}
		\adjustbox{height=\figureheight, trim={{0.04\width} {0.025\height} {.1\width} {0.1\height}},clip}
		{\ORIGINALincludegraphics[page=3]{plots}}
	}

	\subfigure[Discovering symbolic paths]{
		\label{figure:xen:symbolic-execution}
		 \adjustbox{height=\figureheight, trim={{0.04\width} {0.025\height} {.1\width} {0.1\height}}, clip}
		 {\ORIGINALincludegraphics[page=5]{plots}}
	}
	\qquad
	\subfigure[Symbolic execution versus SNI check]{
		\label{figure:xen:sni-vs-symb}
		 \adjustbox{height=\figureheight, trim={{0.04\width} {0.0125\height} {.1\width} {0.1\height}},clip}{
			\ORIGINALincludegraphics[page=6]{plots}
		}
	}

\caption{Scalability analysis for the Xen Project hypervisor. 
In  (a) and (b), green denotes secure traces, red denotes insecure traces, and blue denotes traces producing timeouts.
In  (c) and (d), yellow denotes the first trace discovered for each function, while blue denotes all discovered further traces.
The vertical lines in (d) represent traces where either \memcheck{} times out and \pccheck{} succeed or both time out.
\looseness=-1
}
\vspace{-12pt}
\end{figure*}

This section reports on a case study in which we apply \tool{} on the Xen Project hypervisor~\cite{xen}.
This case study's goal is to understand the challenges in scaling the tool to a significant real-world code base.
It forms a snapshot of our ongoing effort towards the comprehensive side-channel analysis of the Xen hypervisor. 

\subsection{Challenges for scaling-up}\label{sec:scale-challenge}
There are three main challenges for scaling \tool{} to a large code base such as the Xen hypervisor: 

\subsubsection*{ISA support}
Our front end currently supports only a fraction of the x64 ISA (cf.~Section~\ref{sect:tool}).
Supporting the full x64 ISA is conceptually straightforward but out of the scope of this paper. 
For this case study, we treat unsupported instructions  as $\pskip{}$, sacrificing the analysis's correctness. 

\subsubsection*{Policies} \tool{} uses policies specifying the public and secret parts of configurations. The manual specification of precise policies (as in Section~\ref{sec:casestudies}) is infeasible for large code bases, and their automatic inference from the calling context is not yet supported by \tool{}. 
For this case study, we use a  policy that treats registers as ``low'' and  memory locations as ``high'', which may introduce false alarms. 
For instance, the policy treats as ``high'' all function parameters that are retrieved from memory (e.g., popped from the stack), which is why \tool{} flags their speculative uses in memory or branching instructions as leaks.\looseness=-1

\subsubsection*{Path explosion and nontermination} \tool{} is based on symbolic execution, which suffers from path explosion and nontermination when run on programs with loops and indirect jumps. 
In the future, we plan to address this challenge by employing approximate but sound static analysis techniques, such as abstract interpretation.
Such techniques can be employed both to efficiently infer loop invariants, and jump targets, but also to directly address the question whether a given program satisfies SNI or not.
A systematic study of techniques to soundly approximate SNI is out of scope of this paper.
For this case study, as discussed in the following section, we bound the number and the lengths of symbolic paths that are explored, thereby sacrificing the soundness of our analysis.

\subsection{Evaluating scalability}
\subsubsection*{Approach}

To perform a meaningful evaluation of \tool{}'s scalability despite the incomplete path coverage, we compare the time spent on discovering new symbolic paths with the time spent on checking SNI. Analyzing paths of different lengths enables us to evaluate the scalability of checking SNI {\em relative to that} of symbolic execution, which factors out the path explosion problem from the analysis.

We stress that we sacrifice soundness and completeness of the analysis for running \tool{} on the full Xen codebase (see Section~\ref{sec:scale-challenge}). This is why in this section we do {\em not} attempt to make statements about the security of the hypervisor.

\newcommand{\numberOfFunctions}{\numprint{3959}}
\newcommand{\numberOfTraces}{\numprint{24701}}

\subsubsection*{Setup}

We analyze the Xen Project hypervisor version 4.10, 
which we compile using \clang{} v7.0.0.
We identify \numberOfFunctions{}  functions in the generated assembly. 
For each function, we explore at most 25 symbolic paths of at most \numprint{10000} instructions each, with a global timeout of 10 minutes.\footnote{The sources and scripts needed for reproducing our results are available at~\texttt{\url{https://spectector.github.io}}.}

We record execution times as a function of the trace length, i.e., the number of $\loadObs{se}$, $\storeObs{se}$, and $\symPcObs{se}$ observations, rather than path length, since the former is more relevant for the size of the resulting SMT formulas.
We execute our experiments on the machine described in Section~\ref{sec:case-studies:performance}.

\subsection{Experimental results}

\subsubsection*{Cost of symbolic execution}
We measure the time taken for discovering symbolic paths (cf.~Section~\ref{sect:tool}).
In total, \tool{} discovers \numberOfTraces{} symbolic paths.
Figure~\ref{figure:xen:symbolic-execution} depicts the time for discovering paths. 
We highlight the following findings:
\begin{asparaitem}
\item As we apply concolic execution, discovering the first symbolic path does not require any SMT queries and is hence cheap. These cases are depicted by yellow dots in Fig.~\ref{figure:xen:symbolic-execution}.

\item Discovering further paths requires SMT queries. This increases execution time by approximately two orders of magnitude. These cases correspond to the blue dots in  Fig.~\ref{figure:xen:symbolic-execution}. 
\item For 48.3\% of the functions we do not reach the limit of 25 paths, for 35.4\% we do not reach the limit of \numprint{10000} instructions per path, and for 18.7\% we do not encounter unsupported instructions. 13 functions satisfy all three conditions. 
\end{asparaitem}

\subsubsection*{Cost of checking SNI} 
We apply \memcheck{} and \pccheck{} to the \numberOfTraces{} traces (derived from the discovered paths), with a timeout of 1 minute each.
Figure~\ref{figure:xen:sni:memcheck} and~\ref{figure:xen:sni:pccheck} depict the respective analysis runtimes; 
Figure~\ref{figure:xen:sni-vs-symb} relates the time required for discovering a new trace with the time for checking SNI, i.e., for executing lines 3--4 in Algorithm~\ref{algorithm:tool}.

We highlight the following findings:
\begin{asparaitem}

\item \memcheck{} and \pccheck{} can analyze 93.8\% and 94.7\%, respectively, of the \numberOfTraces{} traces in less than 1 minute. 
The remaining traces result in timeouts.

\item For 41.9\% of the traces, checking SNI is at most 10x faster than discovering the trace, and for 20.2\% of the traces it is between 10x and 100x faster.
On the other hand, for 26.9\% of the traces, discovering the trace is at most 10x faster than checking SNI, and for 7.9\% of the traces, discovering the trace is between 10x and 100x faster than checking SNI.
\end{asparaitem}

\subsubsection*{Summary}
Overall, our data indicates that the cost of checking SNI is comparable to that of discovering symbolic paths.
This may be surprising since SNI is a relational property, which requires {\em comparing} executions and is know to scale poorly. However, note that \tool only compares executions that follow the {\em same symbolic path}. This is sufficient because the program counter is observable, i.e., speculative non-interference never requires to consider two executions that disagree on path conditions. We hence conclude that our approach does not exhibit fundamental bottlenecks beyond those it inherits from symbolic execution.

%% file: discussion.tex
\section{Discussion}\label{sec:discussion}

\subsection{Exploitability}

Exploiting speculative execution attacks requires an adversary to (1) prepare the microarchitectural state, (2) run victim code---partially speculatively---to encode information into the microarchitectural state, and (3) extract the leaked information from the microarchitectural state. 
\tool analyzes the victim code to determine whether it may speculatively leak information into the microarchitectural state in any possible attack context. 
Following the terminology of~\cite{microsoft-mitigation,spectreSoK}, speculative non-interference is a semantic characterization of {\em disclosure gadgets} enabled by speculative execution.

\subsection{Scope of the model}

The results obtained by \tool are only valid to the extent that the speculative semantics and the observer model accurately capture the target system.

In particular, \tool may incorrectly classify a program as secure if the speculative semantics does not implicitly\footnote{\emph{Implicitly}, because we take the memory accesses performed by the program and the flow of control as a proxy for the observations an adversary might make, e.g., through the cache.} capture all  additional observations an adversary may make due to speculative execution on an actual microarchitecture. 
For example, microarchitectures could potentially speculate on the condition of a conditional update, which our speculative semantics currently does not permit.

Similarly, \tool may incorrectly classify a program as insecure if the speculative semantics admits speculative executions that are not actually possible on an actual microarchitecture. This might be the case for speculative load hardening on Kocher's Example \#10, as discussed in Section~\ref{sec:casestudies}.

The speculative semantics, however, can always be adapted to more accurately reflect reality, once better documentation of processor behavior becomes available.
In particular, it would be relatively straightforward to extend the speculative semantics with models of indirect jump predictors~\cite{Kocher2018spectre}, return stack buffers~\cite{Maisuradze:2018:RSE:3243734.3243761}, and memory disambiguation predictors~\cite{spectreV4}.
The notion of SNI itself is robust to such changes, as it is defined relative to the speculative semantics. 

We capture ``leakage into the microarchitectural state'' using the relatively powerful  observer of the program execution that sees the location of memory accesses and the jump targets.
This observer could be replaced by a weaker one, which accounts for more detailed models of a CPU's memory hierarchy, and \tool could be adapted accordingly, e.g. by adopting the cache models from CacheAudit~\cite{doychev2015cacheaudit}.
We believe, however, that highly detailed models are not actually desirable for several reasons: 
\begin{inparaenum}[(a)]
	\item they encourage brittle designs that break under small changes to the model,
	\item they have to be adapted frequently, and
	\item they are hard to understand and reason about for compiler developers and hardware engineers.
\end{inparaenum}
The ``constant-time'' observer model adopted in this paper has  proven to offer a good tradeoff between precision and robustness~\cite{MolnarPSW05,AlmeidaBBDE16}.

%% file: relatedwork.tex
\section{Related work}

\para{Speculative execution attacks}
These attacks exploit  speculatively executed instructions to leak information.
After \spectre{}~\cite{Kocher2018spectre, kiriansky2018speculative}, many speculative execution attacks have been discovered that differ in the exploited speculation sources~\cite{Maisuradze:2018:RSE:3243734.3243761, 220586, spectreV4}, the covert channels~\cite{trippel2018meltdownprime,schwarz2018netspectre, stecklina2018lazyfp} used, or the target platforms~\cite{chen2018sgxpectre}.
We refer the reader to~\cite{spectreSoK} for a  survey of  speculative execution attacks and their countermeasures.

 Here, we overview only \spectre{} v1 software-level countermeasures.
 AMD and Intel suggest inserting \textbf{lfence} instructions after branches~\cite{Intel,amd}.
 These instructions effectively act as speculation barriers, and prevent speculative leaks.
 The Intel C++ compiler~\cite{Intel-compiler}, the Microsoft Visual C++ compiler~\cite{microsoft}, and the \textsc{Clang}~\cite{spec-hard} compiler can automatically inject this countermeasure at compile time.
 Taram et al.~\cite{Taram2019} propose \emph{context-sensitive fencing}, a  defense mechanism that dynamically injects fences at the microoperation level where necessary, as determined by a  dynamic information-flow tracker. 
 An alternative technique to injecting fences is to introduce artificial data dependencies~\cite{DBLP:journals/corr/abs-1805-08506, spec-hard}.
 Speculative Load Hardening (SLH)~\cite{spec-hard}, implemented in the \textsc{Clang} compiler, employs carefully injected data dependencies and masking operations to prevent the leak of sensitive information into the microarchitectural state.
 A third software-level countermeasure consists in replacing branching instructions by other computations, like bit masking, that do not trigger speculative execution~\cite{Webkit}.

\para{Detecting speculative leaks}
oo7~\cite{oo7} is a binary analysis tool for detecting speculative leaks.
The tool looks for specific syntactic code patterns and it can analyze large code bases. 
However, it misses some speculative leaks, like Example \#4 from Section~\ref{sec:casestudies}.
oo7 would also incorrectly classify all the programs patched by SLH in our case studies as insecure, since they still match oo7's vulnerable patterns.
In contrast, \tool{} builds on a semantic notion of security and is thus not limited to particular syntactic code patterns.

Disselkoen et al.~\cite{disselkoen2018code} and Mcilroy et al.~\cite{McIlroy19} develop models for capturing speculative execution, which they use to illustrate several known Spectre variants. Neither approach provides a security notion or a detection technique.
Compared with our speculative semantics, the model of~\cite{McIlroy19} more closely resembles microarchitectural implementations by explicitly modeling the reorder buffer, caches, and branch predictors, which we intentionally abstract away.

In work concurrent to ours, Cheang et al.~\cite{Cheang19} introduce the notion of trace property-dependent observational determinism (TPOD), which they instantiate to formally capture the new leaks introduced by the interaction of microarchitectural side channels with speculative execution.
As TPOD is a 4-safety property it can be checked using 4-way self composition.
In contrast, SNI can be checked by 2-way self composition (cf. Proposition~\ref{proposition:speculative-and-non-speculative}), which is likely to be more efficient.

\para{Formal architecture models}
Armstrong et al.~\cite{conf/popl19/armstrong} present formal models for the ARMv8-A, RISC-V,  MIPS, and CHERI-MIPS instruction-set architectures.
Degenbaev~\cite{degenbaev2012x86} and Goel et al.~\cite{Goel2017} develop formal models for parts of the x86 architecture.
Such models enable, for instance, the formal verification of compilers, operating systems, and hypervisors.
However, ISA models naturally abstract from microarchitectural aspects such as speculative execution or caches, which are required to reason about side-channel vulnerabilities.

Zhang et al.~\cite{zhang2018coppelia} present Coppelia, a tool to automatically generate software exploits for hardware designs. 
However, the processor designs they consider, OR1200, PULPino, and Mor1kx, do not feature speculative execution.

\para{Static detection of side-channel vulnerabilities}
Several approaches have been proposed for statically detecting side-channel vulnerabilities in programs~\cite{doychev2015cacheaudit,dkpldi17,brotzman2019casym,AlmeidaBBDE16}. 
These differ from our work in that (1) they do not consider speculative execution, and (2) we exclusively target speculation leaks, i.e., we ignore leaks from the standard semantics. 
However, we note that our tool could easily be adapted to also detect leaks from the standard semantics.

%% file: conclusions.tex
\section{Conclusions}

We introduce speculative non-interference, the first semantic notion
of security against speculative execution attacks.  Based on this
notion we develop \tool{}, a tool for automatically detecting
speculative leaks or proving their absence, and we show how it can be
used to detect subtle leaks---and optimization opportunities---in the way state-of-the-art compilers apply
\spectre{} mitigations.

\medskip
\subsubsection*{Acknowledgments}
We thank Roberto Giacobazzi, Catalin Marinas, Matt Miller, Matthew Parkinson, Niki Vazou, and the anonymous reviewers for helpful discussions and comments. 
This work was supported by a grant from Intel Corporation,
Ram{\'o}n y Cajal grant RYC-2014-16766, Atracci\'on de Talento Investigador grant 2018-T2/TIC-11732A, Spanish projects
TIN2015-70713-R DEDETIS, TIN2015-67522-C3-1-R TRACES, and RTI2018-102043-B-I00 SCUM, and Madrid
regional projects S2013/ICE-2731 N-GREENS and S2018/TCS-4339 BLOQUES.

%% file: non-speculative-semantics.tex
\section{Non-speculative semantics}\label{appendix:non-speculative-semantics}

Given a program $p$, we formalize its non-speculative semantics using the relation $\eval{p}{} \subseteq \Conf \times \Obs \times \Conf$  in Figure~\ref{figure:language:semantics:full}.

\begin{figure*}

\begin{tabular}{l c l c l c l c}
\multicolumn{8}{l}{\bf Expression evaluation}\\
	$\exprEval{n}{a} = n$ & & $\exprEval{x}{a} = a(x)$ &
	 $\exprEval{\unaryOp{e}}{a} = \unaryOp{\exprEval{e}{a}}$ & & $\exprEval{\binaryOp{e_1}{e_2}}{a} = \binaryOp{\exprEval{e_1}{a}}{\exprEval{e_2}{a}}$\\\\
\multicolumn{8}{l}{\bf Instruction evaluation}\\
\multicolumn{8}{l}{
{
\small
\begin{mathpar}
\inferrule[Skip]
{
\select{p}{a(\pc)} = \pskip
}
{
\tup{m, a} \eval{p}{} \tup{m, a[\pc \mapsto a(\pc)+1]}
}

\inferrule[Barrier]
{
	p(a(\pc)) = \pbarrier
}
{
	\tup{m,a} \eval{p}{} \tup{m,a[\pc \mapsto a(\pc) + 1]}
}

\inferrule[Assign]
{
\select{p}{a(\pc)} = \passign{x}{e}\\
x \neq \pc
}
{
\tup{m, a} \eval{p}{} \tup{m, a[\pc \mapsto a(\pc)+1,x \mapsto \exprEval{e}{a}]}
}

\inferrule[ConditionalUpdate-Sat]
{
	p(a(\pc)) = \pcondassign{x}{e}{e'}\\
	\exprEval{e'}{a} = 0\\
	x \neq \pc
}
{
	\tup{m,a} \eval{p}{} \tup{m,a[\pc \mapsto a(\pc) + 1, x \mapsto \exprEval{e}{a}]}
}

\inferrule[ConditionalUpdate-Unsat]
{
	p(a(\pc)) = \pcondassign{x}{e}{e'}\\
	\exprEval{e'}{a} \neq 0\\
	x \neq \pc
}
{
	\tup{m,a} \eval{p}{} \tup{m,a[\pc \mapsto a(\pc) + 1]}
}

\inferrule[Terminate]
{
	\select{p}{a(\pc)} = \bot
}
{
\tup{m, a} \eval{p}{} \tup{m, a[\pc \mapsto \bot]}
}

\inferrule[Load]
{
\select{p}{a(\pc)} = \pload{x}{e} \\ 
x \neq \pc \\
n = \exprEval{e}{a}
}
{
\tup{m, a} \eval{p}{\loadObs{n}} \tup{m, a[\pc \mapsto a(\pc)+1, x \mapsto m(n)]}
}

\inferrule[Store]
{
\select{p}{a(\pc)} = \pstore{x}{e} \\ 	
n = \exprEval{e}{a}
}
{
\tup{m, a} \eval{p}{\storeObs{n}} \tup{ m[n \mapsto a(x)], a[\pc \mapsto a(\pc)+1]}
}

\inferrule[Beqz-Sat]
{
\select{p}{a(\pc)} = \pjz{x}{\lbl} \\
a(x) = 0
}
{
\tup{m, a} \eval{p}{\pcObs{\lbl}} \tup{ m, a[\pc \mapsto \lbl]}
}

\inferrule[Beqz-Unsat]
{
\select{p}{a(\pc)} = \pjz{x}{\lbl} \\
a(x) \neq 0
}
{
\tup{ m, a} \eval{p}{\pcObs{a(\pc)+1}} \tup{ m, a[\pc \mapsto a(\pc) +1]}
}

\inferrule[Jmp]
{
\select{p}{a(\pc)} = \pjmp{e} \\
\lbl = \exprEval{e}{a}
}
{
\tup{m, a} \eval{p}{\pcObs{\lbl}} \tup{ m, a[\pc \mapsto \lbl]}
}
\end{mathpar}
}
}
\end{tabular}
\caption{\lang{} semantics for a program $p$}\label{figure:language:semantics:full}
\end{figure*}

\section{Trace projections}\label{appendix:projections}
Here, we formalize the speculative projection $\specProject{\tau}$ and the non-speculative projection $\nspecProject{\tau}$.

\para{Non-speculative projection}
Given a trace $\tau$, its non-speculative projection contains only the observations that are produced by committed transactions; in other words, rolled-back transactions are removed in the projection.
Formally,  $\nspecProject{\tau}$ is defined as follows:
$\nspecProject{\emptysequence} = \emptysequence$, 
$\nspecProject{(o \concat \tau)} = o \concat \nspecProject{\tau}$ if $o$ is $\loadObs{\sexpr}$, $\storeObs{\sexpr}$, $\pcObs{n}$, or $\symPcObs{\sexpr}$, 
$\nspecProject{(\startObs{i} \concat \tau)} = \nspecProject{\tau}$ if $\rollbackObs{i}$ is not in $\tau$,
$\nspecProject{(\commitObs{i} \concat \tau)} = \nspecProject{\tau}$,
$\nspecProject{(\startObs{i} \concat \tau \concat \rollbackObs{i} \concat \tau') } = \nspecProject{\tau'}$, and
$\nspecProject{\tau} = \emptysequence$ otherwise.

\para{Speculative projection}
Given a speculative trace $\tau$, its speculative projection contains only the observations produced by rolled-back transactions.
Formally, $\specProject{\tau}$ is defined as:
$\specProject{\emptysequence} = \emptysequence$,
$\specProject{(o \concat \tau)} = \specProject{\tau}$ if $o$ is $\loadObs{\sexpr}$, $\storeObs{\sexpr}$, $\pcObs{n}$, or $\symPcObs{\sexpr}$, 
$\specProject{(\startObs{i} \concat \tau)} = \specProject{\tau}$ if $\rollbackObs{i}$ is not in $\tau$,
$\specProject{(\commitObs{i} \concat \tau)} = \specProject{\tau}$,
$\specProject{(\startObs{i} \concat \tau \concat \rollbackObs{i} \concat \tau') } = \mathit{filter}(\tau) \concat \nspecProject{\tau'}$, and
$\nspecProject{\tau} = \emptysequence$ otherwise, where $\mathit{filter}(\tau)$ denotes the trace obtained by dropping all extended observations $\startObs{\id}$, $\commitObs{\id}$, and $\rollbackObs{\id}$ from~$\tau$.

%% file: always-mispredict-semantics.tex
\section{Always-mispredict semantics}\label{appendix:always-mispredict}
We describe the execution of \lang{} programs under the always-mispredict oracle with speculative window $w$ as a ternary evaluation relation $\tup{\ctr,\sigma,s} \ameval{\tau} \tup{\ctr',\sigma',s'}$ mapping a configuration $\tup{\ctr,\sigma,s}$ to a configuration $\tup{\ctr',\sigma',s'}$ while producing the observations~$\tau$.
Differently from the speculative semantics, the always-mispredict semantics does not require a branching history $h$, since its prediction only depends on the branch outcome.
The rules formalizing the always-mispredict semantics are given in Figure~\ref{figure:always-mispredict-semantics}. % 

\textsc{Am-NoBranch} captures the behavior of non-branching instructions. 
Similar to its counterpart \textsc{Se-NoBranch}, the rule acts as a wrapper for the standard semantics. 
The difference lies in the  the auxiliary predicate $\exhaustedTop{s}$ and the auxiliary functions $\decrementTop{s}$, and $\zeroesTop{s}$, which apply their non-primed counterpart only to the {\em last} transaction in the speculative state.
E.g., $\exhaustedTop{s \concat \tup{\id,w,\lbl,\sigma}} = \exhausted{\tup{\id,w,\lbl,\sigma}}$.
This ensures that upon rolling back a nested transaction, its enclosing transaction can explore the other alternative branch to the full depth of the speculative window (corresponding to the case of a correct prediction). 

\textsc{Am-Branch} models the behavior of branching instructions  $\pjz{x}{\lbl'}$.
The rule mispredicts the outcome of the branch instruction by setting the program counter to $\lbl'$ only when the condition is {\em not} satisfied.
The length of the new transaction is set to the minimum of the oracle's speculative window $w$ and $\window{s} -1$, where $\window{s}$ is the remaining length of the last speculative transaction in $s$.
This ensures that nested transactions are not explored for longer than permitted by their enclosing transactions, whose remaining lengths are not decremented during the execution of the nested transaction.

\textsc{Am-Rollback} models the rollback of speculative transactions.
Different from \textsc{SE-Rollback}, and by design of \textsc{Am-NoBranch}, the rule applies only to the last transaction in $s$.
Since the semantics always-mispredicts the outcome of branch instructions, \textsc{SE-Rollback} is always applied, i.e there is no need for a rule that handles committed transactions.

\begin{figure*}
{\small
\begin{mathpar}
\inferrule[Am-NoBranch]
{
p(\sigma(\pc)) \neq \pjz{x}{\lbl}\\
\sigma \eval{p}{\tau} \sigma'\\
 \exhaustedTop{s}\\\\
  s' = {
 	\begin{cases}
 		\decrementTop{s} & \text{if}\ p(\sigma(\pc)) \neq \pbarrier\\
		\zeroesTop{s}  & 	\text{otherwise}
 	\end{cases}
	}
}
{
\tup{\ctr, \sigma, s} \ameval{\tau} \tup{\ctr, \sigma', s'}
}	

\inferrule[Am-Branch]
{
p(\sigma(\pc)) = \pjz{x}{\lbl'}\\
\exhaustedTop{s}\\\\
\lbl = {
\begin{cases}
\sigma(\pc) + 1 & \text{if}\ \sigma(x) = 0  \\
\lbl' 	 	& \text{if}\ \sigma(x) \neq 0 \\
\end{cases}
}\\
\id = \ctr\\\\
s' = \decrementTop{s} \concat \tup{\sigma,\ctr, \mathit{min}(w, \window{s}-1), \lbl}
}
{
\tup{\ctr, \sigma, s} \ameval{ \startObs{\id} \concat \pcObs{\lbl}} \tup{ \ctr +1 , \sigma[\pc \mapsto \lbl], s'}
}	

\inferrule[Am-Rollback]
{
\sigma' \eval{p}{\tau} \sigma''\\
}
{
\tup{\ctr, \sigma, s \concat \tup{\sigma',\id, 0, \lbl}} \ameval{\rollbackObs{\id} \concat \pcObs{\sigma''(\pc)}} \tup{\ctr, \sigma'', s}
}
\end{mathpar}
}
\caption{Always-mispredict speculative semantics for a program $p$ and speculative window $w$}\label{figure:always-mispredict-semantics}
\end{figure*}

Similarly to Proposition~\ref{proposition:speculative-and-non-speculative},  a program's non-speculative behavior can be recovered from  the always-mispredict semantics.

\begin{restatable}{prop}{alwaysMispredSpecNonSpec}
\label{proposition:always-mispredict:speculative-and-non-speculative}
	Let $p$ be a program and $w$ be a speculative window.
Then, $\nspecTraces{p} = \nspecProject{\amTraces{p}{w}}$.
\end{restatable}

\onlyTechReport{The proof of Proposition~\ref{proposition:always-mispredict:speculative-and-non-speculative} is given in Appendix~\ref{appendix:always-mispredict-and-non-speculative}.}

Proposition~\ref{proposition:always-mispredict-worst-case} states that the always-mispredict semantics yields the worst-case leakage.
\begin{restatable}{prop}{alwaysMispredWorstCase}
\label{proposition:always-mispredict-worst-case}
Let $p$ be a program, $w \in \Nat$ be a speculative window, and  $\sigma,\sigma' \in \Init$ be initial configurations.
$\amEval{p}{w}{\sigma} = \amEval{p}{w}{\sigma'}$ iff  $\specEval{p}{\bp}{\sigma} = \specEval{p}{\bp}{\sigma'}$ for all prediction oracles $\bp$ with speculative window at most~$w$.
\end{restatable}

\onlyTechReport{The proof of Proposition~\ref{proposition:always-mispredict-worst-case} is given in Appendix~\ref{appendix:always-mispredict-worst-case}.}

%% file: symbolic-semantics-appendix.tex
\section{Symbolic semantics}\label{appendix:symbolic-semantics}

\begin{figure*}
\begin{tabular}{l l}
	\multicolumn{2}{l}{\bf Expression evaluation}\\
$\exprEval{n}{a} = n$ & if $n \in \Val$  \\
$\exprEval{se}{a} = se$ & if $sa \in \SymExpr \setminus \Val$\\
$\exprEval{x}{a} = a(x)$ & if $x \in \Var$ \\
$\exprEval{\unaryOp{e}}{a} = \mathit{apply}(\unaryOp{}, \exprEval{e}{a}) $ & if $\exprEval{e}{a} \in \Val$\\
$\exprEval{\unaryOp{e}}{a} = \unaryOp{\exprEval{e}{a}} $ & if $\exprEval{e}{a} \in \SymExpr \setminus \Val$\\
$\exprEval{\binaryOp{e_1}{e_2}}{a} = \mathit{apply}(\binaryOp{}{},\exprEval{e_1}{a},\exprEval{e_2}{a})$ & if $\exprEval{e_1}{a}, \exprEval{e_2}{a} \in \Val$\\
$\exprEval{\binaryOp{e_1}{e_2}}{a} = \binaryOp{\exprEval{e_1}{a}}{\exprEval{e_2}{a}}$ & if $\exprEval{e_1}{a} \in \SymExpr \setminus \Val$\\
$\exprEval{\binaryOp{e_1}{e_2}}{a} = \binaryOp{\exprEval{e_1}{a}}{\exprEval{e_2}{a}}$ & if $\exprEval{e_2}{a} \in \SymExpr \setminus \Val$\\
\\
\multicolumn{2}{l}{\bf Instruction evaluation}\\
\end{tabular}
{\small
\begin{mathpar}
\inferrule[Skip]
{
\select{p}{sa(\pc)} = \pskip}
{
\tup{sm, sa} \symEval{p}{} \tup{sm, sa[\pc \mapsto sa(\pc)+1]}
}

\inferrule[Barrier]
{
\select{p}{sa(\pc)} = \pbarrier}
{
\tup{sm, sa} \symEval{p}{} \tup{sm, sa[\pc \mapsto sa(\pc)+1]}
}

\inferrule[Assign]
{
\select{p}{sa(\pc)} = \passign{x}{e}\\
x \neq \pc
}
{
\tup{sm, sa} \symEval{p}{} \tup{sm, sa[\pc \mapsto sa(\pc)+1,x \mapsto \exprEval{e}{sa}]}
}

\inferrule[ConditionalUpdate-Concr-Sat]
{
	\select{p}{sa(\pc)} = \pcondassign{x}{e}{e'}\\
	\exprEval{e'}{sa} = 0\\
	x \neq \pc
}
{
	\tup{sm,sa} \symEval{p}{} \tup{sm,sa[\pc \mapsto sa(\pc) + 1, x \mapsto \exprEval{e}{sa}]}
}

\inferrule[ConditionalUpdate-Concr-Unsat]
{
	\select{p}{sa(\pc)} = \pcondassign{x}{e}{e'}\\
	\exprEval{e'}{sa} =  n\\
	n \in \Val\\\\
	n \neq 0 \\
	x \neq \pc
}
{
	\tup{sm,sa} \symEval{p}{} \tup{sm,sa[\pc \mapsto sa(\pc) + 1]}
}

\inferrule[ConditionalUpdate-Symb]
{
	\select{p}{sa(\pc)} = \pcondassign{x}{e}{e'}\\
	\exprEval{e'}{sa} =  \sexpr\\
	\sexpr \not\in \Val\\
	x \neq \pc
}
{
	\tup{sm,sa} \symEval{p}{} \tup{sm,sa[\pc \mapsto sa(\pc) + 1, x \mapsto \ite{\sexpr = 0}{\exprEval{e}{sa}}{sa(x)}]}
}

\inferrule[Load-Symb]
{
\select{p}{sa(\pc)} = \pload{x}{e} \\
x \neq \pc \\ 
\sexpr = \exprEval{e}{sa}\\\\
\sexpr' = \symRead{sm}{se}
}
{
\tup{sm, sa} \symEval{p}{\loadObs{se}} \tup{sm, sa[\pc \mapsto sa(\pc)+1, x \mapsto \sexpr']}
}

\inferrule[Store-Symb]
{
\select{p}{sa(\pc)} = \pstore{x}{e} \\ 
\sexpr = \exprEval{e}{sa}\\\\
sm' = \symWrite{sm}{se}{sa(x)}
}
{
\tup{sm, sa} \symEval{p}{\storeObs{\sexpr}} \tup{ sm', sa[\pc \mapsto sa(\pc)+1]}
}

\inferrule[Beqz-Concr-Sat]
{
\select{p}{sa(\pc)} = \pjz{x}{\lbl} \\
sa(x) = 0\\
sa(x) \in \Val
}
{
\tup{sm, sa} \symEval{p}{\symPcObs{\top} \concat \pcObs{\lbl}} \tup{ sm, sa[\pc \mapsto \lbl]}
}

\inferrule[Beqz-Symb-Sat]
{
\select{p}{sa(\pc)} = \pjz{x}{\lbl} \\
sa(x) \not\in \Val 
}
{
\tup{sm, sa} \symEval{p}{\symPcObs{sa(x) = 0} \concat \pcObs{\lbl}} \tup{ sm, sa[\pc \mapsto \lbl]}
}

\inferrule[Beqz-Concr-Unsat]
{
\select{p}{sa(\pc)} = \pjz{x}{\lbl} \\
sa(x) \neq 0 \\
sa(x) \in \Val
}
{
\tup{ sm, sa} \symEval{p}{\symPcObs{\top} \concat \pcObs{sa(\pc)+1}} \tup{ sm, sa[\pc \mapsto sa(\pc) +1]}
}

\inferrule[Beqz-Symb-Unsat]
{
\select{p}{sa(\pc)} = \pjz{x}{\lbl} \\
sa(x) \not\in \Val 
}
{
\tup{sm, sa} \symEval{p}{\symPcObs{sa(x) \neq 0} \concat \pcObs{sa(\pc)+1}} \tup{ sm, sa[\pc \mapsto sa(\pc) +1]}
}

\inferrule[Jmp-Concr]
{
\select{p}{sa(\pc)} = \pjmp{e} \\
\lbl = \exprEval{e}{sa}\\
\lbl \in \Val
}
{
\tup{sm, sa} \symEval{p}{\symPcObs{\top} \concat \pcObs{\lbl}} \tup{ sm, sa[\pc \mapsto \lbl]}
}

\inferrule[Jmp-Symb]
{
\select{p}{sa(\pc)} = \pjmp{e} \\
\exprEval{e}{sa} \not\in \Val\\
\lbl \in \Val 
}
{
\tup{sm, sa} \symEval{p}{\symPcObs{\exprEval{e}{sa} = \lbl} \concat \pcObs{\lbl}} \tup{ sm, sa[\pc \mapsto \lbl]}
}

\inferrule[Terminate]
{
	\select{p}{sa(\pc)} = \bot
}
{
\tup{sm, sa} \symEval{p}{} \tup{sm, sa[\pc \mapsto \bot]}
}
\end{mathpar}
}
 \caption{\lang{} symbolic non-speculative semantics for a program $p$}\label{figure:symbolic:non-speculative-semantics}
\end{figure*}

Here, we formalize the symbolic  semantics.

\para{Symbolic expressions}
Symbolic expressions represent computations over symbolic values.
A \textit{symbolic expression} $\sexpr$ is a concrete value $n \in \Val$, a symbolic value $s \in \SymVal$, an if-then-else expression $\ite{\sexpr}{\sexpr'}{\sexpr''}$, or the application of a unary   $\unaryOp{}$  or a binary operator $\binaryOp{}{}$.
\begin{align*}
 \sexpr := n \mid s \mid \ite{\sexpr}{\sexpr'}{\sexpr''} \mid \unaryOp{\sexpr} \mid \binaryOp{\sexpr}{\sexpr'}
\end{align*}

\begin{figure*}[h!]
{\small
\begin{mathpar}
\inferrule[Am-NoBranch]
{
p(\sigma(\pc)) \neq \pjz{x}{\lbl}\\
\sigma \symEval{p}{\tau} \sigma'\\
 \exhaustedTop{s}\\\\
  s' = {
 	\begin{cases}
 		\decrementTop{s} & \text{if}\ p(\sigma(\pc)) \neq \pbarrier\\
		\zeroesTop{s}  & 	\text{otherwise}
 	\end{cases}
	}
}
{
\tup{\ctr, \sigma, s} \symSpeval{p}{\mbp_n}{\tau} \tup{\ctr, \sigma', s'}
}	

\inferrule[Am-Branch-Symb]
{
p(\sigma(\pc)) = \pjz{x}{\lbl''}\\
\exhaustedTop{s}\\\\
\sigma \symEval{p}{\symPcObs{se} \concat \pcObs{\lbl'}} \sigma'\\
\lbl = {
\begin{cases}
\sigma(\pc) + 1 & \text{if}\ \lbl' \neq \sigma(\pc) + 1 \\
\lbl'' 	 	& \text{if}\ \lbl' = \sigma(\pc) + 1 \\
\end{cases}
}\\\\
s' = \decrementTop{s} \concat \tup{\sigma, \ctr, \mathit{min}(w, \window{s}-1), \lbl}\\
\id = \ctr
}
{
\tup{\ctr, \sigma, s} \symSpeval{p}{\mbp_n}{\symPcObs{se} \concat \startObs{\id} \concat \pcObs{\lbl}} \tup{ \ctr +1 , \sigma[\pc \mapsto \lbl], s'}
}	

\inferrule[Am-Rollback]
{
\sigma' \symEval{p}{\tau} \sigma''\\
\sigma''(\pc) \neq \lbl
}
{
\tup{\ctr, \sigma, s \concat \tup{\sigma',\id, 0, \lbl}} \symSpeval{p}{\mbp_n}{\rollbackObs{\id} \concat \pcObs{\sigma''(\pc)}} \tup{\ctr, \sigma'', s}
}
\end{mathpar}
}
\caption{Symbolic always-mispredict speculative semantics for a program $p$ and speculative window $w$}\label{figure:symbolic:symbolic-always-mispredict-semantics}
\end{figure*}

\para{Symbolic memories}
We model symbolic memories as symbolic arrays using the standard theory of arrays~\cite{bradley2007calculus}.
That is, we model memory updates as triples of the form $\symWrite{sm}{\sexpr}{\sexpr'}$, which updates the symbolic memory $sm$ by assigning the symbolic value $\sexpr'$ to the symbolic location $\sexpr$, and memory reads as $\symRead{sm}{\sexpr}$, which denote retrieving the value assigned to the symbolic expression $\sexpr$.

A \textit{symbolic memory} $sm$ is either a function $mem : \Nat \to \SymVal$ mapping memory addresses to symbolic values or a term $\symWrite{sm}{\sexpr}{\sexpr'}$, where $sm$ is a symbolic memory and $\sexpr, \sexpr'$ are symbolic expressions.
To account for symbolic memories, we extend symbolic expressions with terms of the form $\symRead{sm}{\sexpr}$, where $sm$ is a symbolic memory and $\sexpr$ is a symbolic expression, representing memory reads.
\begin{align*}
	sm 		&:= mem \mid \symWrite{sm}{\sexpr}{\sexpr'} \\
	\sexpr 	&:= \ldots \mid \symRead{sm}{\sexpr}
\end{align*}

\para{Evaluating symbolic expressions}
The value of a symbolic expression $\sexpr$ depends on a \textit{valuation} $\mu : \SymVal \to \Val$  mapping symbolic values to concrete ones:
\begingroup
\allowdisplaybreaks
\begin{align*}
 	\symExprEval{n}{\mu} &= n\ \text{if}\ n \in \Val \\
	\symExprEval{s}{\mu} &= \mu(s)\ \text{if}\ s \in \SymVal\\
	\symExprEval{\ite{\sexpr}{\sexpr'}{\sexpr''}}{\mu} &= \symExprEval{\sexpr'}{\mu}\ \text{if}\ \symExprEval{\sexpr}{\mu} \neq 0 \\
	\symExprEval{\ite{\sexpr}{\sexpr'}{\sexpr''}}{\mu} &= \symExprEval{\sexpr''}{\mu}\ \text{if}\ \symExprEval{\sexpr}{\mu} = 0 \\
	\symExprEval{\ominus \sexpr}{\mu} &= \unaryOp{\symExprEval{\sexpr}{\mu}}\\
	\symExprEval{\sexpr \otimes \sexpr'}{\mu} &=  \binaryOp{\symExprEval{\sexpr}{\mu}}{\symExprEval{\sexpr'}{\mu}}\\
	\symExprEval{mem}{\mu} &= \mu \circ mem \\
	\symExprEval{ \symWrite{sm}{\sexpr}{\sexpr'} }{\mu} &= \mu(sm)[\mu(\sexpr) \mapsto \mu(\sexpr')] \\
	\symExprEval{\symRead{sm}{\sexpr}}{\mu} &= \mu(sm)( \mu(\sexpr))
\end{align*}
\endgroup
An expression $\sexpr$ is \textit{satisfiable} if there is a valuation $\mu$ satisfying it, i.e., $\symExprEval{\sexpr}{\mu} \neq 0$.

\para{Symbolic assignments}
A \textit{symbolic assignment} $sa$ is a function mapping registers to symbolic expressions $sa : \Var \to \SymExpr$.
Given a symbolic assignment $sa$ and a valuation $\mu$,  $\mu(sa)$ denotes the  assignment $\mu \circ sa$.
We assume the program counter $\pc$ to always be concrete, i.e., $sa(\pc) \in \Val$.

\para{Symbolic configurations}
A \textit{symbolic configuration} is a pair $\tup{sm,sa}$ consisting of a symbolic memory $sm$  and  a symbolic assignment $sa$.
We lift speculative states to symbolic configurations.
A \textit{symbolic extended configuration} is a triple $\tup{\ctr,\sigma,s}$ where $\ctr \in \Nat$ is a counter, $\sigma \in \Conf$ is a symbolic configuration, and $s$ is a symbolic speculative state.

\para{Symbolic observations}
When symbolically executing a program, we may produce observations whose value is symbolic.
To account for this, we introduce symbolic observations of the form $\loadObs{\sexpr}$ and $\storeObs{\sexpr}$, which are produced when symbolically executing $\loadKywd$ and $\storeKywd$ commands, and $\symPcObs{\sexpr}$, produced when symbolically evaluating branching instructions, where $\sexpr$ is a symbolic expression.
In our symbolic semantics, we use the observations $\symPcObs{\sexpr}$ to represent the symbolic path condition indicating when a path is feasible.
Given a sequence of symbolic observations $\tau$ and a valuation $\mu$,  $\mu(\tau)$ denotes the  trace obtained by (1) dropping all observations $\symPcObs{\sexpr}$, and (2) evaluating all symbolic observations different from $\symPcObs{\sexpr}$ under $\mu$.

\para{Symbolic semantics}
The \textit{non-speculative} semantics is captured by the relation $\symEval{p}{}$  in Fig.~\ref{figure:symbolic:non-speculative-semantics}, while  the  \textit{speculative} semantics is captured by the relation $\symSpeval{p}{\mbp}{}$ in Fig.~\ref{figure:symbolic:symbolic-always-mispredict-semantics}.

\para{Computing symbolic runs and traces}
We now fix the symbolic values. 
The set $\SymVal$ consists of a symbolic value $\symb{x}$ for each register identifier $x$ and of a symbolic value $\symb{mem^n}$ for each memory address $n$. 
We also fix the initial symbolic memory $sm_0 = \lambda n \in \Nat.\ \symb{m^n}$ and the symbolic assignment $sa_0$ such that $sa_0(\pc) = 0$ and $sa_0(x) = \symb{x}$.

The set $\symbTraces{p}{w}$ contains all runs that can be derived using the symbolic semantics (with speculative window $w$) starting from the initial configuration $\tup{sm_0,sa_0}$.
That is, $\symbTraces{p}{w}$ contains all triples $\tup{\tup{sm_0,sa_0}, \tau, \sigma'}$, where $\tau$ is a symbolic trace and $\sigma'$ is a  final symbolic configuration, corresponding to symbolic computations $\tup{0, \tup{sm_0,sa_0}, \emptysequence} \symSpeval{p}{\bp}{\tau}^* \tup{\ctr, \sigma',\emptysequence}$ where the path condition $\bigwedge_{\symPcObs{\sexpr} \in \tau} \sexpr$ is satisfiable.

We compute $\symbTraces{p}{w}$ in the standard way.
We keep track of a path constraint $PC$ and we update it whenever the semantics produces an observation $\symPcObs{\sexpr}$.
We start the computation from $\tup{0, \tup{sm_0,sa_0}, \emptysequence}$ and $PC = \top$.
When executing branch and jump instructions, we explore all branches consistent with the current $PC$, and, for each of them, we update $PC$.

%% file: examples.tex
\section{Code from Case Studies}\label{sec:code}

\subsection{Example \#8}\label{secs:example8}

In Example \#8, the bounds check of Figure~\ref{figure:spectrev1:c-code}
is implemented using a conditional operator:
\begin{lstlisting}[style=Cstyle]
temp &= B[A[y<size?(y+1):0]*512];
\end{lstlisting}
When compiling the example without countermeasures or optimizations,
the conditional operator is translated to a branch instruction (cf.~line 4), which is a source of speculation. Hence, the resulting program
contains a speculative leak, which \tool{} correctly detects.

\begin{lstlisting}[style=ASMstyle]
	 mov    size, %rcx
	 mov    y, %rax
	 cmp    %rcx, %rax
	 jae    .L1
	 add    $1, %rax
	 jmp    .L2
.L1:
	 xor    %rax, %rax
	 jmp    .L2
.L2:
	 mov    A(%rax), %rax
	 shl    $9, %rax
	 mov    B(%rax), %rax
	 mov    temp, %rcx
	 and    %rax, %rcx
	 mov    %rcx, temp
\end{lstlisting}

In the \unp{} \opt{} mode, the
conditional operator is translated as a conditional move
(cf. line 6), for which \tool{} can prove security.

\begin{lstlisting}[style=ASMstyle]
 mov    size, %rax
 mov    y, %rdx
 xor    %rcx, %rcx
 cmp    %rdx, %rax
 lea    1(%rdx), %rax
 cmova  %rax, %rcx
 mov    A(%rcx), %rax
 shl    $9, %rax
 mov    B(%rax), %rax
 and    %rax, temp
\end{lstlisting}

\subsection{Example \#15 in \slh{} mode}\label{ssec:example15}
Here, the adversary provides the input via the pointer \inlineCcode{*y}:
\begin{lstlisting}[style=Cstyle]
 if (*y < size)
    temp &= B[A[*y] * 512];
\end{lstlisting}
In the \unopt{} \slh{} mode, 
\clang{} hardens the address used for performing the memory access \inlineCcode{A[*y]} in lines 8--12, but not the resulting value, which is stored in the register \inlineASMcode{\%cx}.
However, the value stored in  \inlineASMcode{\%cx} is used to perform a second memory access at line 14.
An adversary can exploit the second memory access to speculatively leak the content of \inlineCcode{A[0xFF...FF]}.
In our experiments, \tool{} correctly detected such leak.
\begin{lstlisting}[style=ASMstyle]
 mov    $0, %rax
 mov    y, %rdx
 mov    (%rdx), %rsi
 mov    size, %rdx
 cmp    %rdx, %rsi
 jae     END
 cmovae $-1, %rax
 mov    y, %rcx
 mov    (%rcx), %rcx
 mov    %rax, %rdx
 or     %rcx, %rdx
 mov    A(%rdx), %rcx
 shl    $9, %rcx
 mov    B(%rcx), %rcx
 mov    temp, %rdx
 and    %rcx, %rdx
 mov    %rdx, temp
\end{lstlisting}
In contrast, when Example~\#15 is compiled with the \opt{} flag, \clang{} correctly hardens  \inlineCcode{A[*y]}'s result (cf.~line 10).
This prevents information from flowing into the microarchitectural state during speculative execution.
Indeed, \tool{} proves that the program satisfies speculative non-interference.
\begin{lstlisting}[style=ASMstyle]
 mov    $0, %rax
 mov    y, %rdx
 mov    (%rdx), %rdx
 mov    size, %rsi
 cmp    %rsi, %rdx
 jae    END
 cmovae $-1, %rax
 mov    A(%rdx), %rcx
 shl    $9, %rcx
 or     %rax, %rcx
 mov    B(%rcx), %rcx
 or     %rax, %rcx
 and    %rcx, temp
\end{lstlisting}

%% file: speculative-and-non-speculative-semantics.tex
\section{Relating the speculative and non-speculative semantics (Proposition~\ref{proposition:speculative-and-non-speculative})}
\label{appendix:speculative-and-non-speculative}

We start by spelling out the definition of $\specTraces{p}{\bp}$.
Given a program $p$ and a prediction oracle~$\bp$, we denote by $\specTraces{p}{\bp}$ the set of all triples $\tup{\sigma, \tau, \sigma'} \in \Init \times \ExtObs^* \times \Final$ corresponding to executions $\tup{0, \sigma, \emptysequence,\emptysequence} \speval{p}{\bp}{\tau}^* \tup{\ctr, \sigma',\emptysequence,h}$.

We are now ready to prove Proposition~\ref{proposition:speculative-and-non-speculative}, which we restate here for simplicity:

\speculativeAndNonSpeculative*

\begin{proof}
We prove the two directions separately.

\para{$\mathbf{(\Leftarrow)}$}	
Assume that $\tup{\sigma, \tau', \sigma'} \in \specTraces{p}{\bp}$.
From Proposition~\ref{proposition:speculative-and-non-speculative:soundness} (proved in Section~\ref{appendix:speculative-and-non-speculative:soundness}), it follows that  $\tup{\sigma, \nspecProject{\tau}, \sigma'} \in \nspecTraces{p}$.

\para{$\mathbf{(\Rightarrow)}$}
Assume that $\tup{\sigma, \tau, \sigma'} \in \nspecTraces{p}$.
From Proposition~\ref{proposition:speculative-and-non-speculative:completeness} (proved in Section~\ref{appendix:speculative-and-non-speculative:completeness}), it  follows that there exists a trace $\tau' \in \ExtObs^*$ such that $\tup{\sigma, \tau', \sigma'} \in \specTraces{p}{\bp}$ and $\nspecProject{\tau'} = \tau$.
\end{proof}

In the following, we rely on the concept of \textit{rolled-back transactions} to denote portions of program executions associated with rolled-back speculative transactions.
Concretely,  an execution $\tup{\ctr_0,\sigma_0,s_0, h_0} \speval{p}{\bp}{\tau_0} \tup{\ctr_1,\sigma_1,s_1, h_1} \speval{p}{\bp}{\tau_1} \ldots \speval{p}{\bp}{\tau_{n-1}}  \tup{\ctr_n,\sigma_n,s_n, h_n} \speval{p}{\bp}{\tau_n} \tup{\ctr_{n+1},\sigma_{n+1},s_{n+1}, h_{n+1}}$ is  a \textit{rolled back speculative transaction} iff there is an identifier $\id \in \Nat$ and two labels $\lbl,\lbl' \in \Val$ such that $\tau_0 = \startObs{i} \concat \pcObs{\lbl}$ and $\tau_n = \rollbackObs{i} \concat \pcObs{\lbl'}$.

\subsection{Soundness of the speculative semantics}\label{appendix:speculative-and-non-speculative:soundness}

Proposition~\ref{proposition:speculative-and-non-speculative:soundness} states that the speculative semantics does not introduce spurious non-speculative behaviors.

\begin{restatable}{prop}{}\label{proposition:speculative-and-non-speculative:soundness}
Let $p$ be a program and $\bp$ be a prediction oracle.
Whenever $\tup{0, \sigma, \emptysequence, \emptysequence} \speval{p}{\bp}{\tau}^* \tup{\ctr, \sigma', \emptysequence, h}$, $\sigma \in \Init$, and $\sigma' \in \Final$, then $\sigma \eval{p}{\nspecProject{\tau}}^* \sigma'$.	
\end{restatable}

\begin{proof}
Let $p$ be a program and $\bp$ be a prediction oracle.
Furthermore, let $r$ be an execution  $\tup{0, \sigma, \emptysequence,\emptysequence} \speval{p}{\bp}{\tau}^* \tup{\ctr, \sigma', \emptysequence,h}$ such that  $\sigma \in \Init$ and $\sigma' \in \Final$.
We denote by $R$ the set of rolled back transactions and by $C$ the set containing all identifiers of transactions committed in $r$, i.e.,  $R = \{ \id \in \Nat \mid \exists i \in \Nat.\ \elt{\tau}{i} = \rollbackObs{\id}\}$ and $C = \{ \id \in \Nat \mid \exists i \in \Nat.\ \elt{\tau}{i} = \commitObs{\id}\}$.
We can incrementally construct the execution $\sigma \eval{p}{\nspecProject{\tau}}^* \sigma'$ by applying Lemmas~\ref{lemma:transparency:no-jump}--\ref{lemma:transparency:rollback}.
In particular, for instructions that are not inside a speculative transaction we  apply Lemma~\ref{lemma:transparency:no-jump}.
For branch instructions, there are two cases.
If the associated speculative transaction is eventually rolled back (i.e., its identifier is in $R$), then we apply Lemma~\ref{lemma:transparency:rollback} to the whole rolled back transaction to generate the corresponding step in the non-speculative semantics.
In contrast, if the speculative transaction is eventually committed (i.e., its identifier is in $C$), then we can apply Lemma~\ref{lemma:transparency:committed-jump} (in conjunction with Lemmas~\ref{lemma:transparency:correctness-of-pc-committed-transactions}) to generate the corresponding step in the non-speculative semantics as well as Lemma~\ref{lemma:transparency:no-jump} for the other instructions in the transaction.
Finally, we can simply ignore the steps associated with the \textsc{Se-Commit} rule (see Lemma~\ref{lemma:transparency:commit-no-changes}).
Observe that the resulting traces can be combined to obtain $\nspecProject{\tau}$ since $\tup{0, \sigma, \emptysequence, \emptysequence} \speval{p}{\bp}{\tau}^* \tup{\ctr, \sigma', \emptysequence, h}$ is a terminated execution.
\end{proof}

The above proof relies on the auxiliary lemmas below to handle non-branch instructions (Lemma~\ref{lemma:transparency:no-jump}), committed transactions (Lemmas~\ref{lemma:transparency:committed-jump}, \ref{lemma:transparency:commit-no-changes}, and~\ref{lemma:transparency:correctness-of-pc-committed-transactions}), and rolled back transactions (Lemma~\ref{lemma:transparency:rollback}).

\begin{restatable}{lem}{}\label{lemma:transparency:no-jump}
Let $p$ be a program and $\bp$ be a prediction oracle.
Whenever $\tup{\ctr, \sigma, s,h} \speval{p}{\bp}{\tau} \tup{\ctr', \sigma', s',h'}$, $p(\sigma(\pc)) \neq \pjz{x}{\lbl}$,  and $\exhausted{s}$, then $\sigma \eval{p}{\tau} \sigma'$. 
\end{restatable}

\begin{proof}
Let $p$ be a program and $\bp$ be a prediction oracle.
Assume that $\tup{\ctr, \sigma, s,h} \speval{p}{\bp}{\tau} \tup{\ctr', \sigma', s',h'}$, $p(\sigma(\pc)) \neq \pjz{x}{\lbl}$,  and $\exhausted{s}$.
From this and the \textsc{Se-NoBranch} rule, it follows $\sigma \eval{p}{\tau} \sigma'$.
\end{proof}

\begin{restatable}{lem}{}\label{lemma:transparency:committed-jump}
Let $p$ be a program and $\bp$ be a prediction oracle.
Whenever 
$\tup{\ctr, \sigma, s,h} \speval{p}{\bp}{\tau} \tup{\ctr', \sigma', s',h'}$, 
$p(\sigma(\pc)) = \pjz{x}{\lbl}$,
$\exhausted{s}$,  
$\bp(p,h,\sigma(\pc)) = \tup{\lbl',w}$,
$\sigma \eval{p}{\pcObs{\lbl'}} \sigma''$, 
then $\sigma \eval{p}{\nspecProject{\tau}} \sigma'$. 
\end{restatable}

\begin{proof}
Let $p$ be a program and $\bp$ be a prediction oracle.
Assume that $\tup{\ctr, \sigma, s,h} \speval{p}{\bp}{\tau} \tup{\ctr', \sigma', s',h'}$, 
$p(\sigma(\pc)) = \pjz{x}{\lbl}$,
$\exhausted{s}$,  
$\bp(p,h,\sigma(\pc)) = \tup{\lbl',w}$,
$\sigma \eval{p}{\pcObs{\lbl'}} \sigma''$.
From this, the run was produced by applying the \textsc{Se-Branch} rule.
Hence, $\sigma'= \sigma[\pc \mapsto \lbl']$ and $\tau = \startObs{\ctr} \concat \pcObs{\lbl'}$.
From this, $\sigma \eval{p}{\pcObs{\lbl'}} \sigma''$,  $\nspecProject{\tau} = \pcObs{\lbl'}$, and $\sigma'' = \sigma[\pc \mapsto \lbl']$, it follows that $\nspecProject{\tau} = \pcObs{\lbl'}$ and $\sigma' = \sigma''$. 	
Hence, $\sigma \eval{p}{\tau} \sigma'$. 
\end{proof}

\begin{restatable}{lem}{}\label{lemma:transparency:commit-no-changes}
Let $p$ be a program and $\bp$ be a prediction oracle.
Whenever $\tup{\ctr, \sigma, s, h } \speval{p}{\bp}{\commitObs{\id}} \tup{\ctr', \sigma', s', h'}$, then $\sigma = \sigma'$.
\end{restatable}

\begin{proof}
Let $p$ be a program and $\bp$ be a prediction oracle.
Assume that $\tup{\ctr, \sigma, s, h } \speval{p}{\bp}{\commitObs{\id}} \tup{\ctr', \sigma', s', h'}$.
The only applicable rule is \textsc{Se-Commit} and $\sigma = \sigma'$ directly follows from the rule's definition.
\end{proof}

\begin{restatable}{lem}{}\label{lemma:transparency:rollback}
Let $p$ be a \lang{} program and $\bp$ be a prediction oracle.
Whenever $r:= \tup{\ctr, \sigma, s,h} \speval{p}{\bp}{\tau}^* \tup{\ctr', \sigma', s',h'}$, $\exhausted{s}$, $p(\sigma(\pc)) = \pjz{x}{\lbl}$, and the run $r$ is a rolled-back transaction, then $\sigma \eval{p}{\nspecProject{\tau}} \sigma'$. 
\end{restatable}

\begin{proof}
Let $p$ be a \lang{} program and $\bp$ be a prediction oracle.
Assume that $r:= \tup{\ctr, \sigma, s,h} \speval{p}{\bp}{\tau}^* \tup{\ctr', \sigma', s',h'}$, $\exhausted{s}$, $p(\sigma(\pc)) = \pjz{x}{\lbl}$, and the run $r$ is a rolled-back transaction.
Since $r$ is a rolled back transaction, there is an identifier $\id \in \Nat$ and two labels $\lbl,\lbl' \in \Val$ such that $\tup{\ctr,\sigma,s,h} \speval{p}{\bp}{\startObs{\id} \concat \pcObs{\lbl}} \tup{\ctr_1,\sigma_1,s_1,h_1} \speval{p}{\bp}{\tau_1} \ldots \speval{p}{\bp}{\tau_{n-1}}  \tup{\ctr_n,\sigma_n,s_n,h_n} \speval{p}{\bp}{\rollbackObs{\id} \concat \pcObs{\lbl'}} \tup{\ctr_{n+1},\sigma_{n+1},s_{n+1},h_{n+1}}$.

First, observe that $\tau$ is $\startObs{i} \concat \pcObs{\lbl} \concat \tau' \concat \rollbackObs{i} \concat \pcObs{\lbl'}$.
From this, $\nspecProject{\tau} = \nspecProject{\pcObs{\lbl'}} = \pcObs{\lbl'}$.

Next, the first step in $r$ is obtained by applying the \textsc{Se-Branch} rule (since it is the only rule producing $\startObs{\id}$ observations).
Hence, the last transaction in $s'$ is $\tup{\sigma,\ctr,w,\lbl}$, where $\tup{\lbl,w}$ is the prediction produced by $\bp$, and $\id = \ctr$.
We never modify the identifier, label, and rollback state inside the speculative transaction during the execution up to the last step.
Therefore, it follows (from the rules) that $\tup{\sigma,\ctr,0,\lbl}$ is a transaction in $s_n$.

Finally, the last step is obtained by applying the \textsc{Se-Rollback} rule (since it is the only rule producing $\rollbackObs{\id}$ observations).
Hence, the rolled back transaction in $s_n$ is $\tup{\sigma,\ctr,0,\lbl}$.
From this, we have $\sigma \eval{p}{\pcObs{\lbl'}} \sigma_{n+1}$.
From this, $\nspecProject{\tau} = \pcObs{\lbl'}$, and $\sigma' = \sigma_{n+1}$ (by construction), we have $\sigma \eval{p}{\nspecProject{\tau}} \sigma'$.
\end{proof}

\begin{restatable}{lem}{}\label{lemma:transparency:correctness-of-pc-committed-transactions}
Let $p$ be a program, $\bp$ be a prediction oracle, and $\tup{\ctr_0,\sigma_0,s_0,h_0} \speval{p}{\bp}{\tau_0} \tup{\ctr_1,\sigma_1,s_1,h_1} \speval{p}{\bp}{\tau_1} \ldots \speval{p}{\bp}{\tau_{n-1}}  \tup{\ctr_n,\sigma_n,s_n,h_n} \speval{p}{\bp}{\tau_n} \tup{\ctr_{n+1},\sigma_{n+1},s_{n+1},h_{n+1}}$ be an execution.
For all identifiers $\id \in \Nat$, 
if there is a index $0 \leq j \leq n$ such that $\tau_j = \commitObs{\id}$,
then there exists an $0 \leq i <j$ and $w \in \Nat$ such that 
$\tup{m,a} \eval{p}{\pcObs{\lbl}} \tup{m,a[\pc \mapsto \lbl]}$, 
$\tup{m,a} = \sigma_i$,  
$\tup{\lbl,w} = \bp(p,h_i,\sigma_i(\pc))$, and 
$\tau_i = \startObs{\ctr_i} \concat \pcObs{\sigma_{i+1}(\pc)}$.
\end{restatable}

\begin{proof}
Let $p$ be a program, $\bp$ be a prediction oracle, and $\tup{\ctr_0,\sigma_0,s_0,h_0} \speval{p}{\bp}{\tau_0} \tup{\ctr_1,\sigma_1,s_1,h_1} \speval{p}{\bp}{\tau_1} \ldots \speval{p}{\bp}{\tau_{n-1}}  \tup{\ctr_n,\sigma_n,s_n,h_n} \speval{p}{\bp}{\tau_n} \tup{\ctr_{n+1},\sigma_{n+1},s_{n+1},h_{n+1}}$ be an execution.
Let $\id \in \Nat$ be an arbitrary identifier such that  $\tau_j = \commitObs{\id}$ for some $0 \leq j \leq n$.
The $j$-th step in the run has been produced by applying the \textsc{Se-Commit} rule, which is the only rule producing $\commitObs{\id}$ observations, while committing the the speculative transaction $\tup{\tup{m,a}, \id, 0, \lbl}$ in $s_j$.
Hence, $\tup{m,a} \eval{p}{\pcObs{\lbl}} \tup{m,a[\pc \mapsto \lbl]}$ holds.
Therefore, there is an $0 \leq i < j$ when we executed the \textsc{Se-Branch} rule for the identifier $\id$.
In particular, the rule created the speculative transaction $\tup{\sigma_i, \ctr_i, w, \lbl'}$ where $\tup{\lbl',w} = \bp(p,h_i,\sigma_i(\pc))$, the observations are $\tau_i = \startObs{\ctr_i} \concat \pcObs{\lbl'}$, and the next program counter is $\lbl'$.
Since (1) we only modify the speculative windows inside the speculative transactions states, and (2) there is always only one speculative transaction per identifier, it follows that $\ctr_i = \id$, $\sigma_i = \tup{m,a}$, and $\lbl' = \lbl$.
From this, $\tup{m,a} \eval{p}{\pcObs{\lbl}} \tup{m,a[\pc \mapsto \lbl]}$, 
$\tup{m,a} = \sigma_i$,  
$\tup{\lbl,w} = \bp(p,h_i,\sigma_i(\pc))$, and 
$\tau_i = \startObs{\ctr_i} \concat \pcObs{\sigma_{i+1}(\pc)}$.
\end{proof}

\subsection{Completeness of the speculative semantics}
\label{appendix:speculative-and-non-speculative:completeness}

Proposition~\ref{proposition:speculative-and-non-speculative:completeness} states that the speculative semantics models all behaviors of the non-speculative semantics (independently on the prediction oracle).

\begin{restatable}{prop}{}\label{proposition:speculative-and-non-speculative:completeness}
Let $p$ be a program and $\bp$ be a prediction oracle.
Whenever $\sigma \eval{p}{\tau}^* \sigma'$, $\sigma \in \Init$, and $\sigma' \in \Final$, there are $\tau' \in \ExtObs^*$, $\ctr \in \Nat$, and $h \in \histories$ such that $\tup{0, \sigma, \emptysequence,\emptysequence} \speval{p}{\bp}{\tau'}^* \tup{\ctr, \sigma', \emptysequence,h}$ and $\tau = \nspecProject{\tau'}$.
\end{restatable}

\begin{proof}
Let $p$ be a program and $\bp$ be a prediction oracle.
Moreover, let $r$ be an execution $\sigma_0 \eval{p}{\tau_0} \sigma_1 \eval{p}{\tau_1} \ldots \sigma_n \eval{p}{\tau_n} \sigma_{n+1}$ such that $\sigma_0 \in \Init$ and $\sigma_{n+1} \in \Final$.

We now define a procedure $\expand{\tup{\ctr,\sigma,s,h}}{\sigma \eval{p}{\tau} \sigma'}$ for constructing the execution under the speculative semantics as follows:
\begin{compactitem}
\item If  $p(\sigma(\pc)) \neq \pjz{x}{\lbl'}$ and $\exhausted{s}$, then we apply the \textsc{Se-NoBranch} rule once.
This produces the step $r':= \tup{\ctr,\sigma,s,h} \speval{p}{\bp}{\tau} \tup{\ctr,\sigma', \decrement{s},h'}$.
We return $r'$.

\item If $p(\sigma(\pc)) = \pjz{x}{\lbl'}$ and $\exhausted{s}$, there are two cases:
\begin{compactitem}
\item If  $\bp(p,h,\sigma(\pc)) = \tup{\lbl, w}$ and $\lbl = \sigma'(\pc)$, then we correctly predict the outcome of the branch instruction.
We produce the execution $r'$ by applying  the \textsc{Se-Branch} rule:  $r':= \tup{\ctr,\sigma,s,h} \speval{p}{\bp}{\startObs{\ctr} \concat \pcObs{\lbl} }  \tup{\ctr+1,\sigma', \decrement{s} \concat \tup{\sigma,\ctr,w,\lbl}, h \concat \tup{\sigma(\pc),\ctr,\lbl} } $.
Observe that this step starts a speculative transaction that will be committed (since $\lbl = \sigma'(\pc)$).
We return $r'$.

\item If $\bp(p,h,\sigma(\pc)) = \tup{\lbl, w}$ and $\lbl \neq \sigma'(\pc)$, then we mispredict the outcome of the branch instruction.%
We produce the execution $r'$ by (1) applying \textsc{Se-Branch} once to start the transaction $\ctr$, (2) we repeatedly apply the speculative semantics rules until we roll back the transaction $\ctr$.
This is always possible since speculative transactions always terminate.
Hence, $r'$ is 
$
\tup{\ctr,\sigma,s,h} 
	\speval{p}{\bp}{\startObs{\ctr} \concat \pcObs{\lbl}} 
\tup{\ctr + 1,\sigma[\pc \mapsto \lbl], \decrement{s} \concat \tup{\sigma,\ctr, w, \lbl }, h \concat \tup{\sigma(\pc),\ctr,\lbl} } 
	\speval{p}{\bp}{\tau'}^*
\tup{\ctr'', \sigma'', s'' \concat \tup{\sigma, \ctr, 0, \lbl } \concat s''', h' }
	\speval{p}{\bp}{\rollbackObs{\ctr} \concat \pcObs{ \sigma'(\pc)}}
\tup{\ctr'', \sigma', s'', h' \concat \tup{\sigma(\pc),\ctr,\sigma'(\pc) } }$.
We return $r'$.

\end{compactitem}

\item If $\neg \exhausted{s}$, then there is a pending transaction (with identifier $\id$) that has to be committed.
In this case, we simply perform one step by applying the \textsc{Se-Commit} rule.
\end{compactitem}
Finally, the execution $r'$ is constructed by concatenating the partial executions $r_0, \ldots, r_n$, which are constructed as follows:
\begin{align*}
& 	\mathit{cfg}_0 := \tup{\ctr_0, \sigma_0, \emptysequence, \emptysequence} \\
&	r_i = \expand{\mathit{cfg}_i}{\sigma_i \eval{p}{\tau_i} \sigma_{i+1}}\\
&	\mathit{cfg}_{i} := \mathit{last}(r_{i-1}) \qquad \text{if}\ i>0
\end{align*}
From the construction, it immediately follows that the initial extended configuration in $r'$ is $\tup{\ctr_0, \sigma_0, \emptysequence, \emptysequence}$ and the final configuration is $\tup{\ctr, \sigma_0, \emptysequence,h}$ (for some $\ctr \in \Nat$ and $h \in \histories$) since (1) $r$ is a terminating run, (2) we always execute all transactions until completion, (3) speculative transactions always terminate given that speculative windows are decremented during execution, (4) rolled-back speculative transactions do not affect the configuration $\sigma$.
Similarly, it is easy to see that $\nspecProject{\tau} = \tau_0 \concat \ldots \concat \tau_n$ since (1) the observations associated with committed steps are equivalent to those in the non-speculative semantics, (2) observations associated with rolled back transactions are removed by the projection (and we execute all transactions until termination, and (3) other extended observations are removed by the non-speculative projection. 
\end{proof}

%% file: always-mispredict-and-non-speculative-semantics.tex
\section{Relating the always-mispredict and non-speculative semantics (Proposition~\ref{proposition:always-mispredict:speculative-and-non-speculative})}
\label{appendix:always-mispredict-and-non-speculative}

We start by spelling out the definition of $\amTraces{p}{w}$.
Given a program $p$ and a speculative window~$w$, we denote by $\amTraces{p}{w}$ the set of all triples $\tup{\sigma, \tau, \sigma'} \in \Init \times \ExtObs^* \times \Final$ corresponding to executions $\tup{0, \sigma, \emptysequence} \ameval{\tau}^* \tup{\ctr, \sigma',\emptysequence}$.

We are now ready to prove Proposition~\ref{proposition:always-mispredict:speculative-and-non-speculative}, which we restate here for simplicity:

\alwaysMispredSpecNonSpec*

\begin{proof}
The proposition can be proved in a similar way to Proposition~\ref{proposition:speculative-and-non-speculative}.
Specifically, the proof of the direction $\bm{(\Leftarrow)}$ is a simpler version of Proposition~\ref{proposition:speculative-and-non-speculative:soundness}'s proof since there are only rolled-back transactions. The proof of the direction  $\bm{(\Rightarrow)}$, instead, is identical  Proposition~\ref{proposition:speculative-and-non-speculative:completeness}'s proof.	
\end{proof}

%% file: always-mispredict-worst-case.tex
\newcommand{\amProgEval}[3]{\lbbar #1 \rbbar_{#2}(#3)}
\newcommand{\spProgEval}[3]{\llbracket #1 \rrbracket_{#2}(#3)}
\newcommand{\commProject}[1]{#1\mathord{\upharpoonright_{com}}}
\newcommand{\projectNested}[1]{#1\mathord{\upharpoonright_{nst}}}
\renewcommand{\ameval}[3]{\xLongrightarrow{#3}} 
\newcommand{\minWindow}[1]{\mathit{minWndw}(#1)}
\newcommand{\idx}{\bm{idx}}

\section{Relating speculative and always mispredict semantics (Theorem~\ref{theorem:always-mispredict} and Proposition~\ref{proposition:always-mispredict-worst-case})}\label{appendix:always-mispredict-worst-case}

Below, we provide the proof of Theorem~\ref{theorem:always-mispredict}, which we restate here for simplicity:

\alwaysMispredict*

\begin{proof}
Let $p$ be a program, $P$ be a policy, and $w \in \Nat$ be a speculative window.
We prove the two directions separately.

\para{$\bm{(\Rightarrow)}$}
Assume that $p$ satisfies SNI for $P$ and all prediction oracles~$\bp$ with speculative window at most $w$.
Then, for all $\bp$ with speculative window at most $w$, for all initial configurations $\sigma,\sigma' \in \Init$,
if $\sigma \indist{\policy} \sigma'$ and $\nspecEval{p}{\sigma} = \nspecEval{p}{\sigma'}$, then $\specEval{p}{\bp}{\sigma} = \specEval{p}{\bp}{\sigma'}$.
From this and Proposition~\ref{proposition:always-mispredict-worst-case} (proved below), we have that for all initial configurations $\sigma,\sigma' \in \Init$,
if $\sigma \indist{\policy} \sigma'$ and $\nspecEval{p}{\sigma} = \nspecEval{p}{\sigma'}$, then $\amEval{p}{w}{\sigma} = \amEval{p}{w}{\sigma'}$.

\para{$\bm{(\Leftarrow)}$}
Assume that $p$ is such that for all initial configurations $\sigma,\sigma' \in \Init$,
if $\sigma \indist{\policy} \sigma'$ and $\nspecEval{p}{\sigma} = \nspecEval{p}{\sigma'}$, then $\amEval{p}{w}{\sigma} = \amEval{p}{w}{\sigma'}$.
Let $\bp$ be an arbitrary prediction oracle with speculative window at most $w$.
From Proposition~\ref{proposition:always-mispredict-worst-case} (proved below), we have that for all initial configurations $\sigma,\sigma' \in \Init$,
if $\sigma \indist{\policy} \sigma'$ and $\nspecEval{p}{\sigma} = \nspecEval{p}{\sigma'}$, then $\specEval{p}{\bp}{\sigma} = \specEval{p}{\bp}{\sigma'}$.
Therefore, $p$ satisfies SNI w.r.t. $P$ and $\bp$.
Since $\bp$ is an arbitrary predictor with speculation window at most $w$, then $p$ satisfies SNI for $\policy$ and all prediction oracles~$\bp$ with speculative window at most $w$.
\end{proof}

The proof of Theorem~\ref{theorem:always-mispredict} depends on Proposition~\ref{proposition:always-mispredict-worst-case}, which we restate here for simplicity:

\alwaysMispredWorstCase*

\begin{proof}
The proposition follows immediately from Propositions~\ref{proposition:always-mispred-worst-case:soundness} (proved in Section~\ref{apx:always-mispred-worst-case:soundness}) and~\ref{proposition:always-mispred-worst-case:completeness} (proved in Section~\ref{apx:always-mispred-worst-case:completeness}).
\end{proof}

\subsection{Auxiliary definitions}

Before proving Propositions~\ref{proposition:always-mispred-worst-case:soundness} and~\ref{proposition:always-mispred-worst-case:completeness}, we introduce some auxiliary definitions.

\smallskip
\para{$\pc$-similar configurations}
Two configurations $\sigma,\sigma' \in \Conf$ are \textit{$\pc$-similar}, written $\sigma \pcSim \sigma'$, iff $\sigma(\pc) = \sigma'(\pc)$.

\smallskip
\para{Next-step agreeing configurations}
Two configurations $\sigma, \sigma' \in \Conf$ are \textit{next-step agreeing}, written $\sigma \nextSim \sigma'$, iff there are $\sigma_1, \sigma_1' \in \Conf$ such that $\sigma \eval{p}{\tau} \sigma_1$, $\sigma' \eval{p}{\tau} \sigma_1'$, and $\sigma_1 \pcSim \sigma_1'$.

\smallskip
\para{Similar extended configurations}
Two speculative states $\tup{\sigma, \id, w, \lbl}$ and $\tup{\sigma',\id', w', \lbl'}$ in $\specStates$ are \textit{similar}, written $\tup{\sigma, \id, w, \lbl} \cong \tup{\sigma',\id', w', \lbl'}$, iff $\id = \id'$, $w = w'$, $\lbl = \lbl'$,  $\sigma \pcSim \sigma'$, and $\sigma \nextSim \sigma'$.

Two sequences of speculative states $s$ and $s'$ in $\specStates^*$ are \textit{similar}, written $s \cong s'$, iff all their speculative states are similar.
Formally, $\emptysequence \cong \emptysequence$ and $s \concat \tup{\sigma,\id, w, \lbl} \cong s' \concat \tup{\sigma',\id', w', \lbl'}$ iff $s \cong s'$ and $\tup{\sigma,\id, w, \lbl} \cong \tup{\sigma',\id', w', \lbl'}$.

Two extended configurations $\tup{\ctr, \sigma, s},\tup{\ctr', \sigma', s'}$ are \textit{similar}, written $\tup{\ctr, \sigma, s} \cong \tup{\ctr', \sigma', s'}$, iff $\ctr = \ctr'$, $\sigma \pcSim \sigma'$, and $s \cong s'$.
Similarly, $\tup{\ctr, \sigma, s, h} \cong \tup{\ctr', \sigma', s', h'}$  iff $\ctr = \ctr'$, $\sigma \pcSim \sigma'$, $s \cong s'$, and $h = h'$.

\smallskip
\para{Commit-free projection}
The commit-free projection of a sequence of speculative state $s \in \specStates^*$, written $\commProject{s}$, is as follows:
$\commProject{\emptysequence} = \emptysequence$,
$\commProject{s \concat \tup{\sigma, \id, w, \lbl}} = \commProject{s} \concat \tup{\sigma, \id, w, \lbl}$ if $\sigma \eval{p}{\tau} \sigma'$ and $\sigma'(\pc) \neq \lbl$, and 
$\commProject{s \concat \tup{ \sigma, \id, w, \lbl}} = \commProject{s}$ if $\sigma \eval{p}{\tau} \sigma'$ and $\sigma'(\pc) = \lbl$.

\smallskip
\para{Mirrored configurations}
Two speculative states $\tup{\sigma,\id, w, \lbl}$ and $\tup{\sigma,\id', w', \lbl'}$ in $\specStates$ are \textit{mirrors} if $\sigma = \sigma'$ and $\lbl = \lbl'$.
Two sequences of speculative states $s$ and $s'$ in $\specStates^*$ are \textit{mirrors}, written $s || s'$, if $|s| = |s'|$ and their speculative states are pairwise mirrors.
Finally, $\tup{\ctr, \sigma, s} || \tup{\ctr',\sigma', s', h'}$ iff $\sigma = \sigma'$ and $s || \commProject{s'}$.

\smallskip
\para{Auxiliary functions for handling transactions' lenghts}
The $\window{\cdot} : \specStates^* \to \Nat \cup \{\infty\}$ function is defined as follows:
$\window{\emptysequence} = \infty$ and $\window{s \concat \tup{\ctr, w, n, \sigma}} = w$.
The $\minWindow{\cdot}:\specStates^* \to \Nat \cup \{\infty\}$ function is defined as follows:
$\minWindow{\emptysequence} = \infty$ and $\minWindow{s \concat \tup{\ctr, w, n, \sigma}} = \mathit{min}(w, \minWindow{s})$.

\smallskip
\para{Invariants}
We denote by $INV(s,s')$, where $s, s' \in \specStates^*$, the fact that $|s| = |s'|$ and for all $1 \leq i \leq |s|$, $\minWindow{\prefix{s'}{i}} \leq \window{\prefix{s}{i}}$, where $\prefix{s}{i}$ denotes the prefix of $s$ of length $i$.
Additionally, $INV(\tup{\ctr, \sigma, s}, \tup{\ctr', \sigma', s', h'})$, where $\tup{\ctr, \sigma, s}$ is an extended configuration for the always mispredict semantics whereas $\tup{\ctr', \sigma', s', h'}$ is an extended configuration for the speculative semantics,  holds iff $INV(s,\commProject{s'})$ does.

We denote by $INV2(s,s')$, where $s, s' \in \specStates^*$, the fact that $|s| = |s'|$ and for all $1 \leq i \leq |s|$, $\minWindow{\prefix{s'}{i}} = \window{\prefix{s}{i}}$, where $\prefix{s}{i}$ denotes the prefix of $s$ of length $i$.
Moreover, $INV2(\tup{\ctr, \sigma, s}, \tup{\ctr', \sigma', s', h'})$, where $\tup{\ctr, \sigma, s}$ is an extended configuration for the always mispredict semantics whereas $\tup{\ctr', \sigma', s', h'}$ is an extended configuration for the speculative semantics,  holds iff $INV2(s,\commProject{s'})$ does.

\smallskip
\para{Notation for runs and traces}
For simplicity, we restate here the definitions of $\amProgEval{p}{w}{\sigma}$ and $\spProgEval{p}{\bp}{\sigma}$.
Let $p$ be a program,  $w \in \Nat$ be a speculative window, $\bp$ be a prediction oracle, and  $\sigma \in \Init$ be an initial configuration.
Then, $\amProgEval{p}{w}{\sigma}$ denotes the trace $\tau \in \ExtObs^*$ such that there is a final configuration $\sigma' \in \Final$ such that $\tup{\sigma,\tau,\sigma'} \in \amTraces{p}{w}$, whereas $\spProgEval{p}{w}{\sigma}$ denotes the trace $\tau \in \ExtObs^*$ such that there is a final configuration $\sigma' \in \Final$ such that $\tup{\sigma,\tau,\sigma'} \in \specTraces{p}{\bp}$. 
If there is no such $\tau$, then we write $\amProgEval{p}{w}{\sigma} = \bot$ (respectively $\spProgEval{p}{w}{\sigma}=\bot$).

\subsection{Soundness}\label{apx:always-mispred-worst-case:soundness}

In Proposition~\ref{proposition:always-mispred-worst-case:soundness}	, we prove the soundness of the always mispredict semantics w.r.t. the speculative semantics.

\begin{restatable}{prop}{}
\label{proposition:always-mispred-worst-case:soundness}		
Let $p$ be a program, $w \in \Nat$ be a speculative window, and $\sigma, \sigma' \in \Init$ be two initial configurations.
If $\amProgEval{p}{w}{\sigma} = \amProgEval{p}{w}{\sigma'}$, then $\spProgEval{p}{\bp}{\sigma} = \spProgEval{p}{\bp}{\sigma'}$ for all prediction oracle $\bp$ with speculation window at most $w$.
\end{restatable}

\begin{proof}
Let $w \in \Nat$ be a speculative window and $\sigma, \sigma' \in \Init$ be two initial configurations such that $\amProgEval{p}{w}{\sigma} = \amProgEval{p}{w}{\sigma'}$.
Moreover, let $\bp$ be an arbitrary prediction oracle with speculative window at most $w$.
If $\amProgEval{p}{w}{\sigma} = \bot$, then the computation does not terminate in $\sigma$ and $\sigma'$.
Since speculation does not introduce non-termination, the computation does not terminate according to $\speval{p}{\bp}{}$ as well.
Hence, $\spProgEval{p}{\bp}{\sigma} = \spProgEval{p}{\bp}{\sigma'} = \bot$.
If $\amProgEval{p}{w}{\sigma} \neq \bot$, then we can obtain the runs for $\speval{p}{\bp}{}$ by repeatedly applying Lemma~\ref{lemma:always-mispredict-worstcase:am-to-sp} (proved below) to the runs corresponding to $\amProgEval{p}{w}{\sigma}$ and $ \amProgEval{p}{w}{\sigma'}$.
Observe that we can apply the Lemma~\ref{lemma:always-mispredict-worstcase:am-to-sp} because (I) $\amProgEval{p}{w}{\sigma} = \amProgEval{p}{w}{\sigma'}$, (II) $\amProgEval{p}{w}{\sigma} \neq \bot$, (III) the initial configurations $\tup{0,\sigma, \emptysequence}$, $\tup{0,\sigma', \emptysequence}$, $\tup{0,\sigma, \emptysequence, \emptysequence}$, $\tup{0,\sigma', \emptysequence, \emptysequence}$ trivially satisfy conditions (1)--(4), and (IV) the application of Lemma~\ref{lemma:always-mispredict-worstcase:am-to-sp} preserves (1)--(4).
From the point (g) in Lemma~\ref{lemma:always-mispredict-worstcase:am-to-sp}, we immediately have that the runs have the same traces.
Hence, $\spProgEval{p}{\bp}{\sigma} = \spProgEval{p}{\bp}{\sigma'}$.
\end{proof}

We now prove Lemma~\ref{lemma:always-mispredict-worstcase:am-to-sp}.

\begin{restatable}{lem}{}\label{lemma:always-mispredict-worstcase:am-to-sp}
Let $p$ be a program, $w \in \Nat$ be a speculative window, $\bp$ be a prediction oracle with speculative window at most $w$, 
$am_0, am_0'$ be two extended speculative configurations for the always mispredict semantics, and 
$sp_0, sp_0'$ be two extended speculative configurations for the  speculative semantics.

If the following conditions hold:
\begin{compactenum}
\item $am_0 \cong am_0'$,
\item $sp_0 \cong sp_0'$,
\item $am_0 || sp_0 \wedge am_0' || sp_0'$,
\item $INV(am_0,sp_0) \wedge INV(am_0',sp_0')$,
\item $\amProgEval{p}{w}{am_0} \neq \bot$,
\item $\amProgEval{p}{w}{am_0} = \amProgEval{p}{w}{am_0'}$,
\end{compactenum}
then there are configurations $am_1, am_1'$ (for the always mispredict semantics), $sp_1, sp_1'$ (for the speculative semantics), and  $n \in \Nat$
 such that:
\begin{compactenum}[(a)]
\item $am_0 \ameval{p}{w}{\tau_0}^n am_1 \wedge am_0' \ameval{p}{w}{\tau_0}^n am_1'$,
\item $sp_0 \speval{p}{\bp}{\tau} sp_1 \wedge sp_0' \speval{p}{\bp}{\tau'} sp_1'$,
\item $am_1 \cong am_1'$,
\item $sp_1 \cong sp_1'$,
\item $am_1 || sp_1 \wedge am_1' || sp_1'$,
\item $INV(am_1,sp_1) \wedge INV(am_1',sp_1')$, and
\item $\tau = \tau'$.
\end{compactenum}
\end{restatable}

\begin{proof}
Let $p$ be a program, $w \in \Nat$ be a speculative window, $\bp$ be a prediction oracle with speculative window at most $w$, 
$am_0, am_0'$ be two extended speculative configurations for the ``always mispredict'' semantics, and 
$sp_0, sp_0'$ be two extended speculative configurations for the  speculative semantics.

In the following, we use a dot notation to refer to the components of the extended configurations.
For instance, we write $am_0.\sigma$ to denote the configuration $\sigma$ in $am_0$.
We also implicitly lift functions like $\minWindow{\cdot}$ and $\window{\cdot}$ to extended configurations.

Assume that (1)--(6) hold.
Observe that, from (2), it follows $\minWindow{sp_0} = \minWindow{sp_0'}$.
We proceed by case distinction on $\minWindow{sp_0}$:	
\begin{compactitem}
\item[$\bm{\minWindow{sp_0} > 0}$:] 
From ${\minWindow{sp_0} > 0}$, (1), (2), and (4), we also get that $\window{am_0} >0$ and $\window{am_0'} >0$ (this follows from  $sp_0.s = \prefix{sp_0.s}{|sp_0.s|}$ and (4)).
From $\minWindow{sp_0} > 0$ and (2), it follows that $\exhausted{sp_0}$ and $\exhausted{sp_0'}$ hold.
Additionally, from $\window{am_0} >0$ and (1), we also get $\exhaustedTop{am_0}$ and $\exhaustedTop{am_1}$.
That is, we can apply only the rules \textsc{Se-NoBranch} and \textsc{Se-Branch} according to both semantics.
Observe that $sp_0.\sigma(\pc) = sp_0'.\sigma(\pc)$ (from (2)) and $am_0.\sigma(\pc) = am_0'.\sigma(\pc)$ (from (1)).
Moreover, $am_0.\sigma(\pc) = sp_0.\sigma(\pc)$ and $am_0'.\sigma(\pc) = sp_0'.\sigma(\pc)$ (from (3)).
Hence, the program counters in the four states point to the same instructions.
There are two cases:
\begin{compactitem}
\item[$\bm{p(sp_0.\sigma(\pc))}$ \textbf{is a branch instruction} ${\pjz{x}{n}}$:]
Observe that $\bp(sp_0.\sigma,p,sp_0.h) = \bp(sp_0'.\sigma, p, sp_0'.h)$ since $am_0.\sigma(\pc) = sp_0.\sigma(\pc)$, $sp_0.h = sp_0'.h$ (from (2)), and $\bp$ is a prediction oracle.
Let $\tup{m,w'}$ be the corresponding prediction.

There are two cases:
\begin{compactitem}
\item[\textbf{${\tup{\lbl',w'}}$ is a correct prediction w.r.t. $sp_0$}:]
We first show that $\tup{\lbl',w'}$ is a correct prediction w.r.t. $sp_0'$ as well.
From (6), we know that the trace produced starting from $am_0$ and $am_0'$ are the same.
From this, $am_0.\sigma(\pc) = sp_0.\sigma(\pc)$, and $p(sp_0.\sigma(\pc))$ is a branch instruction, we have that the always mispredict semantics modifies the program counter $\pc$ in the same way in $am_0$ and $am_0'$ when applying the \textsc{Se-Branch} in $\ameval{p}{w}{}$.
Hence, $am_0.\sigma(x) = am_0'.\sigma(x)$.
From this and (3), we also have $sp_0.\sigma(x) = sp_0'.\sigma(x)$.
Thus, $\tup{\lbl',w'}$ is a correct prediction w.r.t. $sp_0'$ as well.

According to the always mispredict semantics, there are $k, k'$ such that:
\begin{align*} 
	r:= am_0 \ameval{p}{w}{\startObs{am_0.\ctr} \concat \pcObs{\lbl_0}} am_2 \ameval{p}{w}{\nu}^{k} am_3 \ameval{p}{w}{\rollbackObs{am_0.\ctr} \concat \pcObs{\lbl_1}} am_1\\
	r':= am_0' \ameval{p}{w}{\startObs{am_0'.\ctr} \concat \pcObs{\lbl_0'}} am_2' \ameval{p}{w}{\nu'}^{k'} am_3' \ameval{p}{w}{\rollbackObs{am_0'.\ctr} \concat \pcObs{\lbl_1'}} am_1'
\end{align*}
From (1) and (6), we immediately have that $am_0.\ctr = am_0'.\ctr$, $\lbl_0 = \lbl_0'$, $k = k'$, $\nu = \nu'$, and $\lbl_1 = \lbl_1'$.
Moreover, since ${\tup{\lbl',w'}}$ is a correct prediction we also have that $\lbl_1 = \lbl'$ and $\lbl_1' = \lbl'$.
Hence, we pick $n$ to be $k+2$.

Additionally, by applying the rule \textsc{Se-Branch} once in $\speval{p}{\bp}{}$ we also get $sp_0 \speval{p}{\bp}{\tau} sp_1$ and $sp_0' \speval{p}{\bp}{\tau'} sp_1'$.

We have already shown that (a) and (b) hold. We now show that (c)--(g) hold as well.
\begin{compactenum}[(a)]
\item[\textbf{(c)}:]
To show that $am_1 \cong am_1'$, we need to show $am_1.\ctr = am_1'.\ctr$, $am_1.\sigma \pcSim am_1'.\sigma$, and $am_1.s \cong am_1'.s$.
First, $am_1.\ctr = am_1'.\ctr$ immediately follows from $\nu = \nu'$ and (1) (that is, $am_0.\ctr = am_0'.\ctr$ and during the mispredicted branch we created the same transactions in both runs).
Second, $am_1.\sigma \pcSim am_1'.\sigma$ follows from $\lbl_1 = \lbl_1'$ (since $am_1.\sigma = am_0.\sigma[\pc \mapsto \lbl_1]$ and $am_1'.\sigma = am_0'.\sigma[\pc \mapsto \lbl_1']$).
Finally, $am_1.s \cong am_1'.s$ immediately follows from $am_1.s = am_0.s$, $am_1'.s = am_0.s$, and (1).

\item[\textbf{(d)}:]
To show that $sp_1 \cong sp_1'$, we need to show $sp_1.\ctr = sp_1'.\ctr$, $sp_1.\sigma \pcSim sp_1'.\sigma$,  $sp_1.s \cong sp_1'.s$, and $sp_1.h = sp_1'.h$.
First,  $sp_1.\ctr = sp_1'.\ctr$ immediately follows from $sp_1.\ctr = sp_0.\ctr + 1$, $sp_1'.\ctr = sp_0'.\ctr+1$, and (2).
Second, $sp_1.\sigma \pcSim sp_1'.\sigma$ follows from (2) and $\bp(sp_0.\sigma,p,sp_0.h) = \bp(sp_0'.\sigma, p, sp_0'.h)$ (so the program counter is updated in the same way in both configurations).
Third, $sp_1.s \cong sp_1'.s$ follows from (2), $sp_1.s$ being $\decrement{sp_0.s} \concat \tup{sp_0.\sigma, sp_0.\ctr, w', \lbl'}$, $sp_1'.s$ being $\decrement{sp_0'.s} \concat \tup{sp_0'.\sigma, sp_0'.\ctr, w', \lbl'}$, $sp_0.\ctr = sp_0'.\ctr$ (from (2)), $sp_0.\sigma \pcSim sp_0'.\sigma$ (from (2)), and $sp_0.\sigma \nextSim sp_0'.\sigma$ (we proved above that $sp_0.\sigma(x) = sp_0'.\sigma(x)$; from this, it follows that the outcome of the branch instruction w.r.t. the non-speculative semantics is the same).
Finally, $sp_1.h = sp_1'.h$ immediately follows from (2), $sp_1.h = sp_0.h \concat \tup{sp_0.\sigma(\pc), sp_0.\ctr, \lbl'}$,  $sp_1'.,h = sp_0'.h \concat \tup{sp_0'.\sigma(\pc), sp_0'.\ctr, \lbl'}$,  $sp_0.\ctr = sp_0'.\ctr$ (from (2)), and $sp_0.\sigma(\pc) = sp_0'.\sigma(\pc)$ (from (2)).

\item[\textbf{(e)}:]
We show only  $am_1||sp_1$ (the proof for $am_1'||sp_1'$ is similar).
To show $am_1 || sp_1$, we need to show $am_1.\sigma = sp_1.\sigma$ and $am_1.s || \commProject{sp_1.s}$.
From the \textsc{Se-Branch} rule in $\speval{p}{\bp}{}$, we have that $sp_1.\sigma = sp_0.\sigma[\pc \mapsto \lbl']$.
The last step in $r$ is obtained by applying the rule \textsc{Se-Rollback} in $\ameval{p}{w}{}$.
From this rule, $am_1.\sigma$ is the configuration obtained by executing one step of the non-speculative semantics starting from $am_0.\sigma$.
From (3), we have that $am_0.\sigma = sp_0.\sigma$.
Moreover, since $\tup{\lbl',w'}$ is a correct prediction for $sp_0$, we have that $am_1.\sigma = am_0.\sigma[\pc \mapsto \lbl']$.
Hence, $am_1.\sigma = sp_1.\sigma$.

We now show that  $am_1.s || \commProject{sp_1.s}$.
From  the \textsc{Se-Branch} rule in $\speval{p}{\bp}{}$, we have that $sp_1.s = \decrement{sp_0.s} \concat \tup{sp_0.\sigma, sp_0.\ctr, w', \lbl'}$.
Since $\tup{\lbl',w'}$ is a correct prediction for $sp_0$, however, $\commProject{sp_1.s} = \commProject{\decrement{sp_0.s}} = \decrement{\commProject{sp_0.s}}$.
Moreover, from $r$, it follows that $am_1.s = am_0.s$.
Hence, $am_1.s || \commProject{sp_1.s}$ follows from (3), $am_1.s = am_0.s$, and $\commProject{sp_1.s} = \commProject{\decrement{sp_0.s}}$.

\item[\textbf{(f)}:]
Here we prove that  $INV(am_1,sp_1)$ holds. $INV(am_1',sp_1')$ can be proved in a similar way.
To show $INV(am_1,sp_1)$, we need to show that $|am_1.s| = |\commProject{sp_1.s}|$ and for all $i \leq i \leq |am_1.s|$, $\minWindow{\prefix{\commProject{sp_1.s}}{i}} \leq \window{\prefix{am_1.s}{i}}$.

We first prove $|am_1.s| = |\commProject{sp_1.s}|$.
From (4), we have $|am_0.s| = |\commProject{sp_0.s}|$.
From the run $r$, we have that $am_1.s = am_0.s$ and, therefore, $|am_1.s| = |am_0.s|$.
From the \textsc{Se-Branch} rule in $\speval{p}{w}{}$, we have that $sp_1.s = \decrement{sp_0.s} \concat \tup{sp_0.\sigma, sp_0.\ctr, w', \lbl'}$.
Since $\tup{\lbl',w'}$ is a correct prediction for $sp_0$, we have that  $\commProject{sp_1.s} = \commProject{\decrement{sp_0.s}}$.
Hence, $|\commProject{sp_1.s}| = |\commProject{sp_0.s}|$.
From  $|am_0.s| = |\commProject{sp_0.s}|$,  $|\commProject{sp_1.s}| = |\commProject{sp_0.s}|$, and $|am_1.s| = |am_0.s|$, we therefore have that $|am_1.s| = |\commProject{sp_1.s}|$.

We now show that for all $1 \leq i \leq |am_1.s|$, $\minWindow{\prefix{\commProject{sp_1.s}}{i}} \leq \window{\prefix{am_1.s}{i}}$.
Observe that $am_1.s = am_0.s$ and $sp_1.s = \decrement{sp_0.s} \concat \tup{sp_0.\sigma, sp_0.\ctr, w', \lbl'}$.
Since $\tup{\lbl',w'}$ is a correct prediction for $sp_0$, we have that  $\commProject{sp_1.s} = \commProject{\decrement{sp_0.s}}$.
Let $i$ be an arbitrary value such that $1 \leq i \leq |am_1.s|$.
Then, $\window{\prefix{am_1.s}{i}} = \window{\prefix{am_0.s}{i}}$ and $\minWindow{\prefix{\commProject{sp_1.s}}{i}} = \minWindow{\prefix{\commProject{\decrement{sp_0.s}}}{i}}$.
From (4), we have that $\minWindow{\prefix{\commProject{sp_0.s}}{i}} \leq \window{\prefix{am_0.s}{i}}$.
From this and  $\window{\prefix{am_1.s}{i}} = \window{\prefix{am_0.s}{i}}$, we get $\minWindow{\prefix{\commProject{sp_0.s}}{i}} \leq \window{\prefix{am_1.s}{i}}$.
From the definition of $\decrement{\cdot}$, we also have $\minWindow{\prefix{\commProject{\decrement{sp_0.s}}}{i}} \leq \minWindow{\prefix{\commProject{sp_0.s}}{i}}$.
From this and $\minWindow{\prefix{\commProject{sp_0.s}}{i}} \leq \window{\prefix{am_1.s}{i}}$, we have $\minWindow{\prefix{\commProject{\decrement{sp_0.s}}}{i}} \leq  \window{\prefix{am_1.s}{i}}$.
Hence, $\minWindow{\prefix{\commProject{sp_1.s}}{i}} \leq \window{\prefix{am_1.s}{i}}$ for all $1 \leq i \leq |am_1.s|$ (since we have proved it for an arbitrary $i$).

\item[\textbf{(g)}:]
The observations are $\tau = \startObs{sp_0.\ctr} \concat \pcObs{\lbl'}$ and $\tau = \startObs{sp_0'.\ctr} \concat \pcObs{\lbl'}$.
From this, $\bp(sp_0.\sigma,p,sp_0.h) = \bp(sp_0'.\sigma, p, sp_0'.h)$, and $sp_0.\ctr = sp_0'.\ctr$ (from (2)), we get $\tau = \tau'$.
\end{compactenum}
This completes the proof of this case.

\item[\textbf{$\tup{\lbl',w'}$ is a misprediction w.r.t. $sp_0$}:]
We first show that $\tup{\lbl',w'}$ is a misprediction w.r.t. $sp_0'$ as well.
From (6), we know that the trace produced starting from $am_0$ and $am_0'$ are the same.
From this, $am_0.\sigma(\pc) = sp_0.\sigma(\pc)$, and $p(sp_0.\sigma(\pc))$ is a branch instruction, we have that the always mispredict semantics modifies the program counter $\pc$ in the same way in $am_0$ and $am_0'$ when applying the \textsc{Se-Branch} in $\ameval{p}{w}{}$.
Hence, $am_0.\sigma(x) = am_0'.\sigma(x)$.
From this and (3), we also have $sp_0.\sigma(x) = sp_0'.\sigma(x)$.
Thus, $\tup{\lbl',w'}$ is a misprediction w.r.t. $sp_0'$ as well.

Let $n = 1$. 
We obtain $am_0 \ameval{p}{w}{\tau_0}^n am_1$ and $am_0' \ameval{p}{w}{\tau_0'}^n am_1'$ by applying once the rule \textsc{Se-Branch} in $\ameval{p}{w}{}$.
From (5) and (6), we have that $\tau_0 = \tau_0'$.
Moreover, by applying the rule \textsc{Se-Branch} once in $\speval{p}{\bp}{}$ we also get $sp_0 \speval{p}{\bp}{\tau} sp_1$ and $sp_0' \speval{p}{\bp}{\tau'} sp_1'$.

We have already shown that (a) and (b) hold. We now show that (c)--(g) hold as well.
\begin{compactenum}[(a)]
\item[\textbf{(c)}:]
To show that $am_1 \cong am_1'$, we need to show $am_1.\ctr = am_1'.\ctr$, $am_1.\sigma \pcSim am_1'.\sigma$, and $am_1.s \cong am_1'.s$.
First,  $am_1.\ctr = am_1'.\ctr$ immediately follows from $am_0.\ctr+1 = am_1.\ctr$, $am_0'.\ctr+1 = am_1'.\ctr$, and (1).
Second, $am_1.\sigma \pcSim am_1'.\sigma$ follows from $\tau_0 = \tau_0'$ (i.e., the program counter is modified in the same ways in both runs).
Finally, to show that $am_1.s \cong am_1'.s$, we need to show that $am_0.\ctr = am_0'.\ctr$, $min(w, \window{am_0.s}-1) = min(w, \window{am_0'.s}-1)$, $n_0 = n_0'$, $am_0.\sigma \pcSim am_0'.\sigma$, and $am_0.\sigma \nextSim am_0'.\sigma$  (since $am_1.s = am_0.s \concat \tup{am_0.\sigma, am_0.\ctr, min(w, \window{am_0.s}-1), n_0}$ and $am_1'.s = am_0'.s \concat \tup{am_0'.\sigma, am_0'.\ctr, min(w, \window{am_0'.s}-1), n_0'}$).
Observe that  $am_0.\ctr = am_0'.\ctr$ follows from (1), $min(w, \window{am_0.s}-1) = min(w, \window{am_0'.s}-1)$ also follow from (1), $n_0 = n_0'$ follows from $\tau_0 = \tau_0'$ (which indicates that $\pc$ is modified in a similar way), $am_0.\sigma \pcSim am_0'.\sigma$ follows from (1), and $am_0.\sigma \nextSim am_0'.\sigma$ follows from $\tau_0 = \tau_0'$ (i.e., since we mispredict in the same direct, then the outcome of the branch instruction must be the same according to the non-speculative semantics).

\item[\textbf{(d)}:]
To show that $sp_1 \cong sp_1'$, we need to show $sp_1.\ctr = sp_1'.\ctr$, $sp_1.\sigma \pcSim sp_1'.\sigma$,  $sp_1.s \cong sp_1'.s$, and $sp_1.h = sp_1'.h$.
First,  $sp_1.\ctr = sp_1'.\ctr$ immediately follows from $sp_1.\ctr = sp_0.\ctr + 1$, $sp_1'.\ctr = sp_0'.\ctr+1$, and (2).
Second, $sp_1.\sigma \pcSim sp_1'.\sigma$ follows from (2) and $\bp(sp_0.\sigma,p,sp_0.h) = \bp(sp_0'.\sigma, p, sp_0'.h)$ (so the program counter is updated in the same way in both configurations).
Third, $sp_1.s \cong sp_1'.s$ follows from (2), $sp_1.s$ being $\decrement{sp_0.s} \concat \tup{sp_0.\sigma, sp_0.\ctr, w', \lbl'}$, $sp_1'.s$ being $\decrement{sp_0'.s} \concat \tup{sp_0'.\sigma, sp_0'.\ctr, w', \lbl'}$, $sp_0.\ctr = sp_0'.\ctr$ (from (2)), $sp_0.\sigma \pcSim sp_0'.\sigma$ (from (2)), and $sp_0.\sigma \nextSim sp_0'.\sigma$ (we proved above that $sp_0.\sigma(x) = sp_0'.\sigma(x)$; from this, it follows that the outcome of the branch instruction w.r.t. the non-speculative semantics is the same).
Finally, $sp_1.h = sp_1'.h$ immediately follows from (2), $sp_1.h = sp_0.h \concat \tup{sp_0.\sigma(\pc), sp_0.\ctr, \lbl'}$,  $sp_1'.,h = sp_0'.h \concat \tup{sp_0'.\sigma(\pc), sp_0'.\ctr, \lbl'}$,  $sp_0.\ctr = sp_0'.\ctr$ (from (2)), and $sp_0.\sigma(\pc) = sp_0'.\sigma(\pc)$ (from (2)).

\item[\textbf{(e)}:]
To show $am_1||sp_1 \wedge am_1'||sp_1'$, we need to show $am_1.\sigma = sp_1.\sigma$, $am_1'.\sigma = sp_1'.\sigma$, $am_1.s || \commProject{sp_1.s}$, and $am_1'.s || \commProject{sp_1'.s}$.
In the following, we show only  $am_1.\sigma = sp_1.\sigma$ and $am_1.s || \commProject{sp_1.s}$ ($am_1'.\sigma = sp_1'.\sigma$ and $\commProject{sp_1'.s}$ can be derived in the same way).
The \textsc{Se-Branch} rule in both semantics only modify the program counter.
From this and (3), for showing $am_1.\sigma = sp_1.\sigma$, it is enough to show $am_1.\sigma(\pc) = sp_1.\sigma(\pc)$.
From the \textsc{Se-Branch} rule in the speculative semantics $\speval{p}{\bp}{}$, we have that $sp_1.\sigma(\pc) = \lbl'$.
There are two cases:
\begin{compactitem}
\item[$\bm{sp_0.\sigma(x) = 0}$:]
From this and $\tup{\lbl',w'}$ being a misprediction for $sp_0$, it follows that $\lbl' = sp_0.\sigma(\pc) + 1$.
Moreover, from $sp_0.\sigma(x) = 0$ and (3), we also have that  $am_0.\sigma(x) = 0$.
From this and the \textsc{Se-Branch} rule in $\ameval{p}{w}{}$, we have  $am_1.\sigma(\pc) = am_0.\sigma(\pc)+1$.
Therefore,  $am_1.\sigma(\pc) = sp_0.\sigma(\pc)+1$ (because $am_0.\sigma(\pc) = sp_0.\sigma(\pc)$ follows from (3)).
Hence, $am_1.\sigma(\pc) = sp_1.\sigma(\pc)$.
Thus, $am_1.\sigma = sp_1.\sigma$.

\item[$\bm{sp_0.\sigma(x) \neq 0}$:]
From this and $\tup{\lbl',w'}$ being a misprediction for $sp_0$, it follows that $\lbl' = \lbl$.
Moreover, from $sp_0.\sigma(x) \neq 0$ and (3), we also have that  $am_0.\sigma(x) \neq 0$.
From this and the \textsc{Se-Branch} rule in $\ameval{p}{w}{}$, we have  $am_1.\sigma(\pc) = \lbl$.
Hence, $am_1.\sigma(\pc) = sp_1.\sigma(\pc)$.
Thus, $am_1.\sigma = sp_1.\sigma$.
\end{compactitem}

We still have to show that $am_1.s || \commProject{sp_1.s}$.
From the \textsc{Se-Branch} rule in $\ameval{p}{w}{}$, we have that $am_1.s = am_0.s \concat \tup{am_0.\sigma, am_0.\ctr, \mathit{min}(w, \window{am_0.s} -1), \lbl'}$ (note that the prediction is $\lbl'$ as we showed above).
In contrast, from the \textsc{Se-Branch} rule in $\speval{p}{w}{}$, we have that $sp_1.s = \decrement{sp_0.s} \concat \tup{sp_0.\sigma, sp_0.\ctr, w', \lbl'}$.
Since $\tup{\lbl',w'}$ is a misprediction for $sp_0$, we have that  $\commProject{sp_1.s} = \commProject{\decrement{sp_0.s}} \concat \tup{sp_0.\sigma, sp_0.\ctr, w', \lbl'}$.
Observe that $am_0.s || \commProject{\decrement{sp_0.s}}$ immediately follows from (3).
Hence, we just have to show that $\tup{am_0.\sigma, am_0.\ctr, \mathit{min}(w, \window{am_0.s} -1), \lbl'} || \tup{sp_0.\sigma, sp_0.\ctr, w', \lbl'}$.
Since the prediction is the same in both speculative states, this requires only to show that $am_0.\sigma = sp_0.\sigma$, which follows from (3).

\item[\textbf{(f)}:]
Here we prove that  $INV(am_1,sp_1)$ holds. $INV(am_1',sp_1')$ can be proved in a similar way.
To show $INV(am_1,sp_1)$, we need to show that $|am_1.s| = |\commProject{sp_1.s}|$ and for all $i \leq i \leq |am_1.s|$, $\minWindow{\prefix{\commProject{sp_1.s}}{i}} \leq \window{\prefix{am_1.s}{i}}$.

We first prove $|am_1.s| = |\commProject{sp_1.s}|$.
From (4), we have $|am_0.s| = |\commProject{sp_0.s}|$.
From the \textsc{Se-Branch} rule in $\ameval{p}{w}{}$, we have that $am_1.s = am_0.s \concat \tup{ am_0.\sigma, am_0.\ctr, \mathit{min}(w, \window{am_0.s} -1), \lbl'}$ (note that the prediction is $\lbl'$ as we showed above).
Hence, $|am_1.s| = |am_0.s| +1$.
From the \textsc{Se-Branch} rule in $\speval{p}{w}{}$, we have that $sp_1.s = \decrement{sp_0.s} \concat \tup{ sp_0.\sigma, sp_0.\ctr, w', \lbl'}$.
Since $\tup{\lbl',w'}$ is a misprediction for $sp_0$, we have that  $\commProject{sp_1.s} = \commProject{\decrement{sp_0.s}} \concat \tup{sp_0.\sigma, sp_0.\ctr, w', \lbl'}$.
Hence, $|\commProject{sp_1.s}| = |\commProject{sp_0.s}|+1$.
From $|am_0.s| = |\commProject{sp_0.s}|$, $|am_1.s| = |am_0.s| +1$, and $|\commProject{sp_1.s}| = |\commProject{sp_0.s}|+1$, we therefore have $|am_1.s| = |\commProject{sp_1.s}|$.

We now show that for all $i \leq i \leq |am_1.s|$, $\minWindow{\prefix{\commProject{sp_1.s}}{i}} \leq \window{\prefix{am_1.s}{i}}$.
Observe that $am_1.s = am_0.s \concat \tup{am_0.\sigma, am_0.\ctr, \mathit{min}(w, \window{am_0.s} -1), \lbl'}$ and $sp_1.s = \decrement{sp_0.s} \concat \tup{sp_0.\sigma, sp_0.\ctr, w', \lbl'}$.
Since $\tup{\lbl',w'}$ is a misprediction for $sp_0$, we have that  $\commProject{sp_1.s} = \commProject{\decrement{sp_0.s}} \concat \tup{sp_0.\sigma, sp_0.\ctr, w', \lbl'}$.
There are two cases:
\begin{compactitem}
\item[$\bm{i = |am_1.s|}$:]
From the \textsc{Se-Branch} rule in $\ameval{p}{w}{}$, we have that $\window{\prefix{am_1.s}{i}} = \mathit{min}(w, \window{am_0.s} -1)$.
From the \textsc{Se-Branch} rule in $\speval{p}{\bp}{}$, we also have that $\minWindow{\prefix{\commProject{sp_1.s}}{i}} = \mathit{min}(\minWindow{\prefix{\commProject{sp_0.s}}{i}} -1,w')$.
Moreover, we have that (1) $w' \leq w$ (since $\bp$ has speculative window at most $w$), and (2) $\minWindow{\prefix{\commProject{sp_0.s}}{i}} -1 \leq \window{am_0.s} -1$ (since we have $\minWindow{\prefix{\commProject{sp_0.s}}{i}}  \leq \window{am_0.s}$ from (4) because $am_0.s = \prefix{am_0.s}{i}$).
As a result, we immediately have that $\mathit{min}(\minWindow{\prefix{\commProject{sp_0.s}}{i}} -1,w') \leq \mathit{min}(w, \window{am_0.s} -1)$.
Hence, $\minWindow{\prefix{\commProject{sp_1.s}}{i}}   \leq \window{\prefix{am_1.s}{i}}$.

\item[$\bm{i < |am_1.s|}$:]
We have  $\window{\prefix{am_1.s}{i}} = \window{\prefix{am_0.s}{i}}$ (because $\prefix{am_1.s}{i} = \prefix{am_0.s}{i}$) and $\minWindow{\prefix{\commProject{sp_1.s}}{i}} = \minWindow{\prefix{\commProject{sp_0.s}}{i}} -1$ (because $\prefix{\commProject{sp_1.s}}{i} = \decrement{\prefix{\commProject{sp_0.s}}{i}}$).
From (4), we also have that $\minWindow{\prefix{\commProject{sp_0.s}}{i}} \leq \window{\prefix{am_0.s}{i}}$.
Therefore, $\minWindow{\prefix{\commProject{sp_0.s}}{i}} - 1  \leq \window{\prefix{am_0.s}{i}}$.
From this, $\window{\prefix{am_1.s}{i}} = \window{\prefix{am_0.s}{i}}$, and $\minWindow{\prefix{\commProject{sp_1.s}}{i}} = \minWindow{\prefix{\commProject{sp_0.s}}{i}} -1$, we have $\minWindow{\prefix{\commProject{sp_1.s}}{i}}   \leq \window{\prefix{am_1.s}{i}}$.

\end{compactitem}

\item[\textbf{(g)}:]
The observations are $\tau = \startObs{sp_0.\ctr} \concat \pcObs{\lbl'}$ and $\tau = \startObs{sp_0'.\ctr} \concat \pcObs{\lbl'}$.
From this, $\bp(sp_0.\sigma,p,sp_0.h) = \bp(sp_0'.\sigma, p, sp_0'.h)$, and $sp_0.\ctr = sp_0'.\ctr$ (from (2)), we get $\tau = \tau'$.
\end{compactenum}
This completes the proof of this case.
\end{compactitem}
This completes the proof of this case.

\item[$\bm{p(sp_0.\sigma(\pc))}$ \textbf{is not a branch instruction}:]
Let $n = 1$. 
We obtain $am_0 \ameval{p}{w}{\tau_0}^n am_1$ and $am_0' \ameval{p}{w}{\tau_0'}^n am_1'$ by applying once the rule \textsc{Se-NoBranch} in $\ameval{p}{w}{}$.
From (5) and (6), we have that $\tau_0 = \tau_0'$.
Moreover, by applying the rule \textsc{Se-NoBranch} once in $\speval{p}{\bp}{}$ we also get $sp_0 \speval{p}{\bp}{\tau} sp_1$ and $sp_0' \speval{p}{\bp}{\tau'} sp_1'$.

We have already shown that (a) and (b) hold. We now show that (c)--(g) hold as well.
\begin{compactenum}[(a)]
\item[\textbf{(c)}:]
To show that $am_1 \cong am_1'$, we need to show $am_1.\ctr = am_1'.\ctr$, $am_1.\sigma \pcSim am_1'.\sigma$, and $am_1.s \cong am_1'.s$.
First,  $am_1.\ctr = am_1'.\ctr$ immediately follows from $am_0.\ctr = am_1.\ctr$, $am_0'.\ctr = am_1'.\ctr$, and (1).
Second, $am_1.\sigma \pcSim am_1'.\sigma$ follows from (1) and the fact that $\eval{p}{}$ modifies the program counters in $am_0.\sigma$ and $am_0'.\sigma$ in the same way (if the current instruction is not a jump, both program counters are incremented by 1 and if the instruction is a jump, the target is the same since $\tau_0 = \tau_0'$).
Finally, $am_1.s \cong am_1'.s$ follows from (1), $am_1.s$ being either $\decrTop{am_0.s}$ or $\zeroesTop{am_0.s}$ (depending on $am_0.\pc$), and $am_1'.s$ being either $\decrTop{am_0'.s}$ or $\zeroesTop{am_0'.s}$ (depending on $am_0'.\pc$).

\item[\textbf{(d)}:]
To show that $sp_1 \cong sp_1'$, we need to show $sp_1.\ctr = sp_1'.\ctr$, $sp_1.\sigma \pcSim sp_1'.\sigma$,  $sp_1.s \cong sp_1'.s$, and $sp_1.h = sp_1'.h$.
First,  $sp_1.\ctr = sp_1'.\ctr$ immediately follows from $sp_0.\ctr = sp_1.\ctr$, $sp_0'.\ctr = sp_1'.\ctr$, and (2).
Second, $sp_1.\sigma \pcSim sp_1'.\sigma$ follows from (2) and the fact that $\eval{p}{}$ modifies the program counters in $sp_0.\sigma$ and $sp_0'.\sigma$ in the same way (if the current instruction is not a jump, both program counters are incremented by 1 and if the instruction is a jump, the target is the same since $\tau = \tau'$, which we prove in (g) below).
Third, $sp_1.s \cong sp_1'.s$ follows from (2), $sp_1.s$ being either $\decrement{sp_0.s}$ or $\zeroes{sp_0.s}$ (depending on $sp_0.\pc$), and $sp_1'.s$ being either $\decrement{sp_0'.s}$ or $\zeroes{sp_0'.s}$ (depending on $sp_0'.\pc$).
Finally, $sp_1.h = sp_1'.h$ immediately follows from (2), $sp_1.h = sp_0.h$, and $sp_1'.h = sp_0'.h$.

\item[\textbf{(e)}:]
To show $am_1||sp_1 \wedge am_1'||sp_1'$, we need to show $am_1.\sigma = sp_1.\sigma$, $am_1'.\sigma = sp_1'.\sigma$, $am_1.s || \commProject{sp_1.s}$, and $am_1'.s || \commProject{sp_1'.s}$.
From the \textsc{Se-NoBranch}, we have that $am_0.\sigma \eval{p}{\tau_0} am_1.\sigma$,  $am_0'.\sigma \eval{p}{\tau_0'} am_1'.\sigma$,  $sp_0.\sigma \eval{p}{\tau} sp_1.\sigma$, $sp_0'.\sigma \eval{p}{\tau'} sp_1'.\sigma$.
From this and (3), we have that $am_1.\sigma = sp_1.\sigma$ and $am_1'.\sigma = sp_1'.\sigma$ (since $am_0.\sigma = sp_0.\sigma$, $am_0'.\sigma = sp_0'.\sigma$).
Moreover, $am_1.s || \commProject{sp_1.s}$, and $am_1'.s || \commProject{sp_1'.s}$ immediately follow from (3),  $sp_1.s$ being either $\decrement{sp_0.s}$ or $\zeroes{sp_0.s}$ (depending on $sp_0.\pc$), and $sp_1'.s$ being either $\decrement{sp_0'.s}$ or $\zeroes{sp_0'.s}$ (depending on $sp_0'.\pc$), which only modify the speculation windows.

\item[\textbf{(f)}:]
For simplicity, we show only $INV(am_1,sp_1)$, the proof for $INV(am_1', sp_1')$ is analogous.
To show $INV(am_1,sp_1)$, we need to show that $|am_1.s| = |\commProject{sp_1.s}|$ (which immediately follows from $INV(am_0,sp_0)$ and the fact that we do not add or remove speculative states) and for all $i \leq i \leq |am_1.s|$, $\minWindow{\prefix{\commProject{sp_1.s}}{i}} \leq \window{\prefix{am_1.s}{i}}$.
There are two cases:
\begin{compactitem}
\item[$\bm{i = |am_1.s|}$:]
If the  instruction is not $\pbarrier$, then  $\window{\prefix{am_1.s}{i}} = \window{\prefix{am_0.s}{i}} - 1$ and $\minWindow{\prefix{\commProject{sp_1.s}}{i}} = \minWindow{\prefix{\commProject{sp_0.s}}{i}} -1$.
From (4), we also have that $\minWindow{\prefix{\commProject{sp_0.s}}{i}} \leq \window{\prefix{am_0.s}{i}}$.
Hence,  $\minWindow{\prefix{\commProject{sp_1.s}}{i}} \leq \window{\prefix{am_1.s}{i}}$.

If the executed instruction is $\pbarrier$, then $\window{\prefix{am_1.s}{i}} = 0$ and $\minWindow{\prefix{\commProject{sp_1.s}}{i}} = 0$.
Therefore,  $\minWindow{\prefix{\commProject{sp_1.s}}{i}} \leq \window{\prefix{am_1.s}{i}}$.

\item[$\bm{i < |am_1.s|}$:]
If the executed instruction is not $\pbarrier$, then  $\window{\prefix{am_1.s}{i}} = \window{\prefix{am_0.s}{i}}$ (because $\prefix{am_1.s}{i} = \prefix{am_0.s}{i}$) and $\minWindow{\prefix{\commProject{sp_1.s}}{i}} = \minWindow{\prefix{\commProject{sp_0.s}}{i}} -1$.
From (4), we also have that $\minWindow{\prefix{\commProject{sp_0.s}}{i}} \leq \window{\prefix{am_0.s}{i}}$.
Therefore, $\minWindow{\prefix{\commProject{sp_0.s}}{i}} - 1  \leq \window{\prefix{am_0.s}{i}}$.
From this, $\window{\prefix{am_1.s}{i}} = \window{\prefix{am_0.s}{i}}$, and $\minWindow{\prefix{\commProject{sp_1.s}}{i}} = \minWindow{\prefix{\commProject{sp_0.s}}{i}} -1$, we have $\minWindow{\prefix{\commProject{sp_1.s}}{i}}   \leq \window{\prefix{am_1.s}{i}}$.

If the executed instruction is $\pbarrier$, then  $\window{\prefix{am_1.s}{i}} = \window{\prefix{am_0.s}{i}}$ (because $\prefix{am_1.s}{i} = \prefix{am_0.s}{i}$) and $\minWindow{\prefix{\commProject{sp_1.s}}{i}} = 0$.
From this, we immediately have  $\minWindow{\prefix{\commProject{sp_1.s}}{i}}   \leq \window{\prefix{am_1.s}{i}}$.
\end{compactitem}

\item[\textbf{(g)}:]  
From (3), we have that $am_0.\sigma = sp_0.\sigma$ and $am_0'.\sigma = sp_0'.\sigma$.
From the \textsc{Se-NoBranch}, we have that $am_0.\sigma \eval{p}{\tau_0} am_1.\sigma$,  $am_0'.\sigma \eval{p}{\tau_0'} am_1'.\sigma$,  $sp_0.\sigma \eval{p}{\tau} sp_1.\sigma$, $sp_0'.\sigma \eval{p}{\tau'} sp_1'.\sigma$.
That is,  $\tau_0 = \tau$ and $\tau_0' = \tau'$.
From this and $\tau_0  = \tau_0'$, we get $\tau = \tau'$.
\end{compactenum}
This completes the proof of this case.

\end{compactitem}
This completes the proof of this case.

\item[$\bm{\minWindow{sp_0} = 0}$:] 
Let $i, i'$ be the largest values such that $\minWindow{\prefix{sp_0.s}{i}} >0$ and $\minWindow{\prefix{sp_0'.s}{i'}} >0$.
We now show that $i = i'$.
From (2), we have that $sp_0.s \cong sp_0'.s$ which means that the speculative states in $sp_0.s$ and $sp_0'.s$ are pairwise similar.
As a result, the remaining speculative windows are pairwise the same.
Hence,  $\minWindow{\prefix{sp_0.s}{i}} = \minWindow{\prefix{sp_0'.s}{i'}}$ for all $1 \leq i \leq |sp_0.s|$.
Observe also that each pair of speculative states in $sp_0.s$ and $sp_0'.s$ results either in two commits or two rollbacks (again from (2)).

There are two cases:
\begin{compactitem}
\item[\textbf{We can apply the \textsc{Se-Commit} rule in $sp_0$}:]
Hence, $sp_0.s = s_0 \concat \tup{\sigma_0, \id_0, 0, \lbl_0} \concat s_1$, $\sigma_0 \eval{p}{\rho} \sigma_1$, $\exhausted{s_1}$, and $\sigma_1(\pc) = \lbl_0$.
From (2), we have that $sp_0'.s = s_0' \concat \tup{\sigma_0', \id_0', 0, \lbl_0'} \concat s_1'$, $\exhausted{s_1'}$, $\sigma_0' \eval{p}{\rho'} \sigma_1'$, $\lbl_0 = \lbl_0'$, $\id_0 = \id_0'$, and $\sigma_1'(\pc) = \lbl_0'$.
As a result, we can apply the \textsc{Se-Commit} rule also to $sp_0'$.
Therefore, we have $sp_0 \speval{p}{\bp}{\rollbackObs{\id_0} \concat \pcObs{\sigma_1(\pc)}} sp_1$ and $sp_0' \speval{p}{\bp}{\rollbackObs{\id_0'} \concat \pcObs{\sigma_1'(\pc)}} sp_1'$.

Let $n = 0$.
Therefore, we have $am_0 \ameval{p}{w}{}^0 am_1$, $am_0' \ameval{p}{w}{}^0 am_1'$, $am_0 = am_1$, and $am_0' = am_1'$.

We have already shown that (a) and (b) hold. We now show that (c)-(g) hold as well.
\begin{compactenum}[(a)]
\item[\textbf{(c)}:]
$am_1 \cong am_1'$ immediately follows from (1), $am_0 = am_1$, and $am_0' = am_1'$.

\item[\textbf{(d)}:]
To show that $sp_1 \cong sp_1'$, we need to show $sp_1.\ctr = sp_1'.\ctr$, $sp_1.\sigma \pcSim sp_1'.\sigma$,  $sp_1.s \cong sp_1'.s$, and $sp_1.h = sp_1'.h$.
First,  $sp_1.\ctr = sp_1'.\ctr$ immediately follows from $sp_0.\ctr = sp_1.\ctr$, $sp_0'.\ctr = sp_1'.\ctr$, and (2).
Second, $sp_1.\sigma \pcSim sp_1'.\sigma$ follows from (2), $sp_1.\sigma = sp_0.\sigma$, and $sp_1'.\sigma = sp_0'.\sigma$.
Third, $sp_1.s \cong sp_1'.s$ follows from (2), $sp_1 = s_0 \concat s_1$, and  $sp_1' = s_0' \concat s_1'$.
Finally, $sp_1.h = sp_1'.h$ immediately follows from (2), $sp_1.h = sp_0.h \concat \tup{\sigma_0(\pc),\id_0,\sigma_1(\pc)}$, and $sp_1'.h = sp_0'.h \concat \tup{\sigma_0'(\pc),\id_0',\sigma_1'(\pc)}$, $\id_0 = \id_0'$ (from (2)), $\sigma_0(\pc)= \sigma_0'(\pc)$ (since all configurations in the speculative states are pairwise $\pc$-agreeing), and $\sigma_1(\pc)= \sigma_1'(\pc)$ (since all configurations in the speculative states are pairwise next-step agreeing).

\item[\textbf{(e)}:]
Here we prove only $am_1||sp_1$ ($am_1'||sp_1'$ can be proved in the same way).
To show $am_1||sp_1$, we need to show $am_1.\sigma = sp_1.\sigma$ and $am_1.s || \commProject{sp_1.s}$.
From the \textsc{Se-Commit} rule in $\speval{p}{\bp}{}$, we have that $sp_1.\sigma =sp_0.\sigma_1$.
Moreover, $am_1 = am_0$ and, therefore, $am_1.\sigma = am_0.\sigma$.
From (3), $sp_1.\sigma =sp_0.\sigma_1$, and $am_1.\sigma = am_0.\sigma$, we have $am_1.\sigma = sp_1.\sigma$.
Observe also that $am_1.s || \commProject{sp_1.s}$ follows from $am_1.s = am_0.s$ (since $am_1 = am_0$) and $\commProject{sp_1.s} = \commProject{sp_0.s}$ (since the removed speculative state is one leading to a commit).

\item[\textbf{(f)}:]
For simplicity, we show only $INV(am_1,sp_1)$, the proof for $INV(am_1', sp_1')$ is analogous.
To show $INV(am_1,sp_1)$, we need to show that $|am_1.s| = |\commProject{sp_1.s}|$ and for all $i \leq i \leq |am_1.s|$, $\minWindow{\prefix{\commProject{sp_1.s}}{i}} \leq \window{\prefix{am_1.s}{i}}$.
Observe that $am_1.s = am_0.s$ (because $am_1 = am_0$) and $\commProject{sp_1.s} = \commProject{sp_0.s}$ (because we remove a committed speculative state).
As a result, both $|am_1.s| = |\commProject{sp_1.s}|$ and $\minWindow{\prefix{\commProject{sp_1.s}}{i}} \leq \window{\prefix{am_1.s}{i}}$ for all $i \leq i \leq |am_1.s|$ immediately follow from (4).

\item[\textbf{(g)}:]
The observations $\tau_0$ and $\tau_0'$ are respectively $\commitObs{\id_0}$ and $\commitObs{\id_0'}$.
From (2), we have that $sp_0.s \cong sp_0'.s$ which implies $\id_0 = \id_0'$.
Hence, $\tau_0 = \tau_0'$.

\end{compactenum}
This completes the proof of this case.

\item[\textbf{We can apply the \textsc{Se-Rollback} rule in $sp_0$}:] 
Hence, $sp_0.s = s_0 \concat \tup{\sigma_0, \id_0, 0, \lbl_0} \concat s_1$, $\sigma_0 \eval{p}{\rho} \sigma_1$, $\exhausted{s_1}$, and $\sigma_1(\pc) \neq \lbl_0$.
From (2), we have that $sp_0'.s = s_0' \concat \tup{\sigma_0', \id_0', 0, \lbl_0'} \concat s_1'$, $\exhausted{s_1'}$, $\sigma_0' \eval{p}{\rho'} \sigma_1'$, $\lbl_0 = \lbl_0'$, $\id_0 = \id_0'$, and $\sigma_1'(\pc) \neq \lbl_0'$.
As a result, we can apply the \textsc{Se-Rollback} rule also to $sp_0'$.
Therefore, we have $sp_0 \speval{p}{\bp}{\rollbackObs{\id_0} \concat \pcObs{\sigma_1(\pc)}} sp_1$ and $sp_0' \speval{p}{\bp}{\rollbackObs{\id_0'} \concat \pcObs{\sigma_1'(\pc)}} sp_1'$.

In the following, we denote by $\idx$ the index of $sp_0.s$'s $i+1$-th speculative state (i.e., the one that we are rolling back) inside its commit-free projection $\commProject{sp_0.s}$, i.e., the $i+1$-th speculative state in $sp_0$ is the $\idx$-th state in $\commProject{sp_0.s}$.
Observe that $\idx$ is also the index of $sp_0'.s$'s $i+1$-th speculative state, which we are rolling back, inside $\commProject{sp_0'.s}$.

Let $\id$ be the transaction identifier in the $\idx$-th speculative state in $am_0.s$.
From (1), it follows that $\id$ is also the identifier in the $\idx$-th speculative state in $am_0'.s$.
According to the always mispredict semantics, there are $k, k'$ such that:\looseness=-1
\begin{align*} 
	r:= am_0  \ameval{p}{w}{\nu}^{k} am_2 \ameval{p}{w}{\rollbackObs{\id} \concat \pcObs{\lbl_1}} am_1\\
	r':= am_0' \ameval{p}{w}{\nu'}^{k'} am_2' \ameval{p}{w}{\rollbackObs{\id} \concat \pcObs{\lbl_1'}} am_1'
\end{align*}
From (1) and (6), we immediately have that  $k = k'$, $\nu = \nu'$, and $\lbl_1 = \lbl_1'$.
Hence, we pick $n$ to be $k+1$.

We have already shown that (a) and (b) hold. We now show that (c)-(g) hold as well.
\begin{compactenum}[(a)]
\item[\textbf{(c)}:]
To show that $am_1 \cong am_1'$, we need to show $am_1.\ctr = am_1'.\ctr$, $am_1.\sigma \pcSim am_1'.\sigma$, and $am_1.s \cong am_1'.s$.
First,  $am_1.\ctr = am_1'.\ctr$ follows from (1) and $\nu = \nu'$.
Second, $am_1.\sigma \pcSim am_1'.\sigma$ follows from (1) (since the configurations in the speculative states in $am_0.s$ and $am_0'.s$ are pairwise next-step agreeing and no intermediate steps in $r,r'$ has modified the state that we are rolling back in the last step).
Finally, $am_1.s \cong am_1'.s$ follows from (1), $am_1.s = \prefix{am_0.s}{\idx-1}$, and  $am_1'.s = \prefix{am_0'.s}{\idx-1}$.

\item[\textbf{(d)}:]
To show that $sp_1 \cong sp_1'$, we need to show $sp_1.\ctr = sp_1'.\ctr$, $sp_1.\sigma \pcSim sp_1'.\sigma$,  $sp_1.s \cong sp_1'.s$, and $sp_1.h = sp_1'.h$.
First,  $sp_1.\ctr = sp_1'.\ctr$ immediately follows from $sp_0.\ctr = sp_1.\ctr$, $sp_0'.\ctr = sp_1'.\ctr$, and (2).
Second, $sp_1.\sigma \pcSim sp_1'.\sigma$ follows from (2) (since the configurations in the speculative states are pairwise next-step agreeing).
Third, $sp_1.s \cong sp_1'.s$ follows from (2), $sp_1.s = \prefix{sp_0.s}{i}$, and  $sp_1'.s = \prefix{sp_0'.s}{i}$.
Finally, $sp_1.h = sp_1'.h$ immediately follows from (2), $sp_1.h = sp_0.h \concat \tup{\sigma_0(\pc),\id_0,\sigma_1(\pc)}$, and $sp_1'.h = sp_0'.h \concat \tup{\sigma_0'(\pc),\id_0',\sigma_1'(\pc)}$, $\id_0 = \id_0'$ (from (2)), $\sigma_0(\pc)= \sigma_0'(\pc)$ (since all configurations in the speculative states are pairwise $\pc$-agreeing), and $\sigma_1(\pc)= \sigma_1'(\pc)$ (since all configurations in the speculative states are pairwise next-step agreeing).

\item[\textbf{(e)}:]
Here we prove only $am_1||sp_1$ ($am_1'||sp_1'$ can be proved in the same way).
To show $am_1||sp_1$, we need to show $am_1.\sigma = sp_1.\sigma$ and $am_1.s || \commProject{sp_1.s}$.
From the \textsc{Se-Rollback} rule in $\speval{p}{\bp}{}$, we have that $sp_1.\sigma =\sigma_1$, which is the configuration obtained by executing one step starting from the configuration $\sigma_0$ in the $i+1$-th speculative state in $sp_0.s$.
From the run $r$, $am_1$ is obtained by applying the \textsc{Se-Rollback} rule in $\ameval{p}{w}{}$.
As a result, $am_1.\sigma$ is the configuration obtained by executing one step starting from the configuration in $\idx$-th speculative state in $am_0.s$.
Observe that, from (3), $\sigma_0$ is exactly the configuration in the $\idx$-th speculative state in $am_0.s$ (because the $i+1$-th state in $sp_0.s$ that we are rolling back is also the $\idx$-th state in $\commProject{sp_0.s}$). 
Hence, $am_1.\sigma = sp_1.\sigma$.
Moreover, observe that $am_1.s || \commProject{sp_1.s}$ immediately follows from (3), $am_1.s$ being $\prefix{am_0.s}{\idx-1}$, $\commProject{sp_1.s} = \commProject{\prefix{sp_0.s}{i}}$.

\item[\textbf{(f)}:]
For simplicity, we show only $INV(am_1,sp_1)$, the proof for $INV(am_1', sp_1')$ is analogous.
To show $INV(am_1,sp_1)$, we need to show that $|am_1.s| = |\commProject{sp_1.s}|$ and for all $i \leq i \leq |am_1.s|$, $\minWindow{\prefix{\commProject{sp_1.s}}{i}} \leq \window{\prefix{am_1.s}{i}}$.

First, $|am_1.s| = |\commProject{sp_1.s}|$ follows from 
(4), $am_1.s = \prefix{am_0.s}{\idx -1}$ (because we are rolling back the $\idx$-th transaction in the speculative state and the semantics only modifies the last speculative state), $\commProject{sp_1.s} = \commProject{\prefix{sp_0.s}{i}}$, and $|\commProject{\prefix{sp_0.s}{i}}| = \idx -1$ (because $sp_0.s$'s $i+1$-th speculative state corresponds to $\commProject{sp_0.s}$'s $\idx$-th state by construction).

We now show that for all $1 \leq j \leq |am_1.s|$, $\minWindow{\prefix{\commProject{sp_1.s}}{j}} \leq \window{\prefix{am_1.s}{j}}$.
Let $j$ be an arbitrary value such that  $1 \leq j \leq |am_1.s|$.
Observe that $\prefix{am_1.s}{j} = \prefix{am_0.s}{j}$  and $\prefix{\commProject{sp_1.s}}{j} = \prefix{\commProject{sp_0.s}}{j}$ (since $j \leq |am_1.s| < \idx$).
From this and (4), we therefore immediately have that $\minWindow{\prefix{\commProject{sp_1.s}}{j}} \leq \window{\prefix{am_1.s}{j}}$.

\item[\textbf{(g)}:]
The observations $\tau_0$ and $\tau_0'$ are respectively $\rollbackObs{\id_0} \concat \pcObs{\sigma_1(\pc)}$ and $\rollbackObs{\id_0'} \concat \pcObs{\sigma_1'(\pc)}$.
From (2), we have that $sp_0.s \cong sp_0'.s$ which implies $\id_0 = \id_0'$.
Moreover, we also have that $\sigma_0 \nextSim \sigma_0'$ and, therefore, $\sigma_1(\pc) = \sigma_1'(\pc)$.
Hence, $\tau_0 = \tau_0'$.

\end{compactenum}
This completes the proof of this case.

\end{compactitem}
This completes the proof of this case.

\end{compactitem}
This completes the proof of our claim.
\end{proof}

\subsection{Completeness}\label{apx:always-mispred-worst-case:completeness}

In Proposition~\ref{proposition:always-mispred-worst-case:completeness}, we prove the completeness of the always mispredict semantics w.r.t. the speculative semantics.

\begin{restatable}{prop}{}
\label{proposition:always-mispred-worst-case:completeness}	
Let $w \in \Nat$ be a speculative window and $\sigma, \sigma' \in \Init$ be two initial configurations.
If $\amProgEval{p}{w}{\sigma} \neq \amProgEval{p}{w}{\sigma'}$, then there exists a  prediction oracle $\bp$ with speculation window at most $w$ such that $\spProgEval{p}{\bp}{\sigma} \neq \spProgEval{p}{\bp}{\sigma'}$.
\end{restatable}

\begin{proof}
Let $w \in \Nat$ be a speculative window and $\sigma, \sigma' \in \Init$ be two initial configurations such that $\amProgEval{p}{w}{\sigma} \neq \amProgEval{p}{w}{\sigma'}$.

We first analyze the case when the program does not terminate.
There are two cases: either $\amProgEval{p}{w}{\sigma} = \bot \wedge \amProgEval{p}{w}{\sigma'} \neq \bot$ or $\amProgEval{p}{w}{\sigma} \neq \bot \wedge \amProgEval{p}{w}{\sigma'} = \bot$.
Consider the first case $\amProgEval{p}{w}{\sigma} = \bot \wedge \amProgEval{p}{w}{\sigma'} \neq \bot$ (the proof for the second case is similar).
Since speculative execution does not introduce or prevent non-termination, we have that $\spProgEval{p}{\bp}{\sigma} = \bot \wedge \spProgEval{p}{\bp}{\sigma'} \neq \bot$ for all prediction oracles $\bp$.
Hence, there is a prediction oracle with speculative window at most $w$ such that $\spProgEval{p}{\bp}{\sigma} \neq \spProgEval{p}{\bp}{\sigma'}$.

We now consider the case where both $\amProgEval{p}{w}{\sigma} \neq \bot$ and $\amProgEval{p}{w}{\sigma'} \neq \bot$.
From this and $\amProgEval{p}{w}{\sigma} \neq \amProgEval{p}{w}{\sigma'}$, it follows that there are $n, k, k' \in \Nat$, extended configurations $am_1,am_2,am_1',am_2'$, final extended configurations $am_3, am_3'$, and sequences of observations $\tau, \tau_{am,1}, \tau_{am,2}, \tau_{am,1}', \tau_{am,2}'$ such that $\tau_{am,1} \neq \tau_{am,1}'$, $am_1 \cong am_1'$ and:
\begin{align*} 
	\tup{0,\sigma,\emptysequence}  \ameval{p}{w}{\tau}^{n} am_1 \ameval{p}{w}{\tau_{am,1}}  am_2 \ameval{p}{w}{\tau_{am,2}}^{k} am_3 \\
	\tup{0,\sigma,\emptysequence}  \ameval{p}{w}{\tau}^{n} am_1' \ameval{p}{w}{\tau_{am,1}'}  am_2' \ameval{p}{w}{\tau_{am,2}'}^{k'} am_3'
\end{align*}
We claim that there is a prediction oracle $\bp$ with speculative window at most $w$ such that 
(1) $\tup{0,\sigma,\emptysequence, \emptysequence}  \speval{p}{\bp}{\nu}^{*} sp_1$ and $sp_1.\sigma = am_1.\sigma$ and $INV2(sp_1,am_1)$, and
(2) $\tup{0,\sigma',\emptysequence, \emptysequence}  \speval{p}{\bp}{\nu}^{*} sp_1'$ and $sp_1'.\sigma = am_1'.\sigma$ and $INV2(sp_1',am_1')$, and (3) $sp_1 \cong sp_1'$.

We now proceed by case distinction on the rule in $\ameval{p}{w}{}$ used to derive $am_1 \ameval{p}{w}{\tau_{am,1}}  am_2$.
Observe that from $am_1 \cong am_1'$ and $\tau_{am,1} \neq \tau_{am,1}'$, it follows that the only possible rules are \textsc{Se-NoBranch} and \textsc{Se-Branch} (since applying the \textsc{Se-Rollback} rule would lead to $\tau_{am,1} = \tau_{am,1}'$):
\begin{compactitem}
\item[\textbf{\textsc{Se-NoBranch}}:]
From the rule, we have that $p(am_1.\sigma(\pc)) \neq \pjz{x}{\lbl}$, $\exhaustedTop{am_1.s}$, and $am_1.\sigma \eval{p}{\tau_{am,1}} am_2.\sigma$.
From  $am_1 \cong am_1'$, we also have $p(am_1'.\sigma(\pc)) \neq \pjz{x}{\lbl}$ (since $am_1.\sigma(\pc) = am_1'.\sigma(\pc)$) and $\exhaustedTop{am_1'.s}$ (since $am_1'.s \cong am_1.s$).
Hence, also  $am_1' \ameval{p}{w}{\tau_{am,1}'}  am_2'$ is obtained by applying the \textsc{Se-NoBranch} rule.
Therefore,  $am_1'.\sigma \eval{p}{\tau_{am,1}} am_2'.\sigma$.

From our claim, we have 
(1) $\tup{0,\sigma,\emptysequence, \emptysequence}  \speval{p}{\bp}{\nu}^{*} sp_1$, $sp_1.\sigma = am_1.\sigma$, and $INV2(sp_1,am_1)$, and
(2) $\tup{0,\sigma',\emptysequence, \emptysequence}  \speval{p}{\bp}{\nu}^{*} sp_1'$ and $sp_1'.\sigma = am_1'.\sigma$ and $INV2(sp_1',am_1')$.
From $sp_1.\sigma = am_1.\sigma$,  $sp_1'.\sigma = am_1'.\sigma$,  $p(am_1'.\sigma(\pc)) \neq \pjz{x}{\lbl}$, and $p(am_1.\sigma(\pc)) \neq \pjz{x}{\lbl}$, we immediately have $p(sp_1'.\sigma(\pc)) \neq \pjz{x}{\lbl}$, and $p(sp_1.\sigma(\pc)) \neq \pjz{x}{\lbl}$.
Moreover, from $\exhaustedTop{am_1.s}$, $\exhaustedTop{am_1'.s}$, $INV2(sp_1,am_1)$, and $INV2(sp_1',am_1')$, we also have $\exhausted{sp_1}$ and $\exhausted{sp_1'}$.
Hence, we can apply to both configurations the \textsc{Se-NoBranch} rule in $\speval{p}{\bp}{}$.
As a result, $sp_1 \speval{p}{\bp}{\tau_{sp,1}} sp_2$ and $sp_1' \speval{p}{\bp}{\tau_{sp,1}'} sp_2'$ where the traces are derived using the non-speculative semantics $sp_1.\sigma \eval{p}{\tau_{sp,1}} sp_2.\sigma$ and $sp_1'.\sigma \eval{p}{\tau_{sp,1}'} sp_2'.\sigma$.
From this, $sp_1.\sigma = am_1.\sigma$, $sp_1'.\sigma = am_1'.\sigma$, and the determinism of the non-speculative semantics, we have that $\tau_{sp,1} = \tau_{am,1}$ and $\tau_{sp,1}' = \tau_{am,1}'$.
This, together with $\tau_{am,1} \neq \tau_{am,1}'$, leads to $\tau_{sp,1} \neq \tau_{sp,1}'$.
Observe also that both $\tau_{sp,1}$ and $\tau_{sp,1}'$ are observations in $\ExtObs$.
Therefore, there are sequences of observations $\rho, \rho'$ such that $\spProgEval{p}{\bp}{\sigma} = \nu \concat \tau_{sp,1} \concat \rho$ and $\spProgEval{p}{\bp}{\sigma'} = \nu \concat \tau_{sp,1}' \concat \rho'$.
As a result, $\spProgEval{p}{\bp}{\sigma} \neq \spProgEval{p}{\bp}{\sigma'}$.

\item[\textbf{\textsc{Se-Branch}}:]
From the rule, we have that $p(am_1.\sigma(\pc)) = \pjz{x}{\lbl''}$ and $\exhaustedTop{am_1.s}$.
From  $am_1 \cong am_1'$, we also have $p(am_1'.\sigma(\pc)) = \pjz{x}{\lbl''}$ (since $am_1.\sigma(\pc) = am_1'.\sigma(\pc)$) and $\exhaustedTop{am_1'.s}$ (since $am_1'.s \cong am_1.s$).
Hence, also  $am_1' \ameval{p}{w}{\tau_{am,1}'}  am_2'$ is obtained by applying the \textsc{Se-Branch} rule.
Therefore, $\tau_{am,1} = \startObs{am_1.\ctr} \concat \pcObs{\lbl}$ and $\tau_{am,1}' = \startObs{am_1'.\ctr} \concat \pcObs{\lbl'}$, where $\lbl = \begin{cases} \lbl'' & \text{if}\ am_1.\sigma(x) \neq 0\\ am_1.\sigma(\pc)+1 & \text{if}\ am_1.\sigma(x) = 0 \end{cases}$ and $\lbl' = \begin{cases} \lbl'' & \text{if}\ am_1'.\sigma(x) \neq 0\\ am_1'.\sigma(\pc)+1 & \text{if}\ am_1'.\sigma(x) = 0 \end{cases}$.
Since $am_1.\ctr = am_1.\ctr$ (from $am_1'.s \cong am_1.s$) and $\tau_{am,1} \neq \tau_{am,1}'$, we have that either $am_1.\sigma(x) = 0 \wedge am_1'.\sigma(x) \neq 0$ or $am_1.\sigma(x) \neq 0 \wedge am_1'.\sigma(x) = 0$.
Without loss of generality, we assume that $am_1.\sigma(x) = 0 \wedge am_1'.\sigma(x) \neq 0$ (the proof for the other case is analogous).

We now executes two steps starting from $sp_1$ and $sp_1'$ with respect to $\speval{p}{\bp}{}$, and we construct $sp_1 \speval{p}{\bp}{\tau_{sp,1}} sp_2 \speval{p}{\bp}{\tau_{sp,2}} sp_3$ and $sp_1' \speval{p}{\bp}{\tau_{sp,1}'} sp_2' \speval{p}{\bp}{\tau_{sp,2}'} sp_3'$.
From our claim, we have 
(1) $\tup{0,\sigma,\emptysequence, \emptysequence}  \speval{p}{\bp}{\nu}^{*} sp_1$, $sp_1.\sigma = am_1.\sigma$, and $INV2(sp_1,am_1)$, and
(2) $\tup{0,\sigma',\emptysequence, \emptysequence}  \speval{p}{\bp}{\nu}^{*} sp_1'$ and $sp_1'.\sigma = am_1'.\sigma$ and $INV2(sp_1',am_1')$, and
(3) $|sp_1.h| = |sp_1'.h|$.
From $sp_1.\sigma = am_1.\sigma$,  $sp_1'.\sigma = am_1'.\sigma$,  $p(am_1'.\sigma(\pc)) = \pjz{x}{\lbl''}$, and $p(am_1.\sigma(\pc)) = \pjz{x}{\lbl''}$, we immediately have $\bp(sp_1.\sigma,p,sp_1.h) \neq \bot$, 
$\bp(sp_1'.\sigma,p,sp_1'.h) \neq \bot$, $p(sp_1'.\sigma(\pc)) = \pjz{x}{\lbl''}$, and $p(sp_1.\sigma(\pc)) = \pjz{x}{\lbl''}$.
Moreover, from $\exhaustedTop{am_1.s}$, $\exhaustedTop{am_1'.s}$, $INV2(sp_1,am_1)$, and $INV2(sp_1',am_1')$, we also have $\exhausted{sp_1}$ and $\exhausted{sp_1'}$.
Hence, we can apply to both configurations the \textsc{Se-Branch} rule in $\speval{p}{\bp}{}$.
From $sp_1.\sigma = am_1.\sigma$,  $sp_1'.\sigma = am_1'.\sigma$, and $am_1 \cong am_1'$, we immediately have that $sp_1.\sigma(\pc) = sp_1'.\sigma(\pc)$.
From this, (3), and the way in which we construct $\bp$ (see below), we have that $\bp(sp_1.\sigma,p, sp_1.h) = \bp(sp_1'.\sigma,p,sp_1'.h) = \tup{\lbl'',0}$.
By applying the \textsc{Se-Branch} rule to $sp_1$ we get $sp_1 \speval{p}{\bp}{\tau_{sp,1}} sp_2$, where $\tau_{sp,1} = \startObs{sp_1.\ctr} \concat \pcObs{\lbl''}$, $sp_2.\ctr = sp_1.\ctr+1$, $sp_2.\sigma = sp_1.\sigma[\pc \mapsto \lbl'']$, $sp_2.s = \decrement{sp_1.s} \concat \tup{sp_1.\sigma, sp_1.\id, 0, \lbl''}$, and $sp_2.h = sp_1.h \concat \tup{sp_1.\sigma(\pc), sp_1.\id, \lbl''}$.
Similarly, by applying the \textsc{Se-Branch} rule to $sp_1'$ we get $sp_1' \speval{p}{\bp}{\tau_{sp,1}'} sp_2'$, where $\tau_{sp,1}' = \startObs{sp_1'.\ctr} \concat \pcObs{\lbl''}$, $sp_2'.\ctr = sp_1'.\ctr+1$, $sp_2'.\sigma = sp_1'.\sigma[\pc \mapsto \lbl'']$, $sp_2'.s = \decrement{sp_1'.s} \concat \tup{sp_1'.\sigma, sp_1'.\id, 0, \lbl''}$, and $sp_2'.h = sp_1'.h \concat \tup{sp_1'.\sigma(\pc), sp_1'.\id, \lbl''}$.
Observe that $\tau_{sp,1} = \tau_{sp,1}'$ immediately follows from (3).

Since $sp_2.s = \decrement{sp_1.s} \concat \tup{sp_1.\id, 0, n, sp_1.\sigma}$ and $sp_2'.s = \decrement{sp_1'.s} \concat \tup{sp_1'.\id, 0, n, sp_1'.\sigma}$, we can commit or rollback the transactions associated with the last speculative states in $sp_2.s$ and $sp_2'.s$ (i.e., the transactions we started in the step  $sp_1 \speval{p}{\bp}{\tau_{sp,1}} sp_2$ and  $sp_1' \speval{p}{\bp}{\tau_{sp,1}'} sp_2'$).
We start by considering $sp_2$.
From $am_1.\sigma(x) = 0$ and $sp_1.\sigma = am_1.\sigma$, we have that $sp_1.\sigma(x) = 0$.
From this, $p(sp_1.\sigma(\pc)) = \pjz{x}{\lbl''}$, and  $\bp(sp_1.\sigma,p, sp_1.h) = \tup{\lbl'',0}$, the configuration $\sigma$ obtained as $sp_1.\sigma \eval{p}{\tau} \sigma$ is such that $\sigma(\pc) = \lbl''$.
Hence, we can apply the \textsc{Se-Commit} rule to $sp_2$ and obtain $sp_2 \speval{p}{\bp}{\tau_{sp,2}} sp_3$, where $\tau_{sp,2} = \commitObs{sp_1.\ctr}$, $sp_3.\ctr = sp_2.\ctr$, $sp_3.\sigma = sp_2.\sigma$, $sp_3.s = \decrement{sp_1.s}$, and $sp_3.h = sp_2.h$.
Consider now  $sp_2'$.
From $am_1'.\sigma(x) \neq 0$ and $sp_1'.\sigma = am_1'.\sigma$, we have that $sp_1'.\sigma(x) \neq 0$.
From this, $p(sp_1'.\sigma(\pc)) = \pjz{x}{\lbl''}$, and  $\bp(sp_1'.\sigma,p, sp_1'.h) = \tup{\lbl'',0}$, the configuration $\sigma'$ obtained as $sp_1'.\sigma \eval{p}{\tau} \sigma'$ is such that $\sigma'(\pc) = sp_1'.\sigma(\pc) + 1 $.
From this and the well-formedness of $p$, we have that $\sigma'(\pc) \neq \lbl''$.
Hence, we can apply the \textsc{Se-Rollback} rule to $sp_2'$ and obtain $sp_2' \speval{p}{\bp}{\tau_{sp,2}'} sp_3'$, where $\tau_{sp,2}' = \rollbackObs{sp_1'.\ctr} \concat \pcObs{\sigma'(\pc)}$, $sp_3'.\ctr = sp_2'.\ctr$, $sp_3'.\sigma = \sigma'$, $sp_3'.s = \decrement{sp_1'.s}$, and $sp_3'.h = sp_2'.h \concat \tup{sp_1'.\sigma(\pc), sp_1'.\ctr, \sigma'(\pc)}$.
Therefore, there are sequences of observations $\rho,\rho'$ such that $\spProgEval{p}{\bp}{\sigma} = \nu \concat \tau_{sp,1} \concat \tau_{sp,2} \concat \rho$ and $\spProgEval{p}{\bp}{\sigma'} = \nu \concat \tau_{sp,1}' \concat \tau_{sp,2}' \concat \rho'$, where $\tau_{sp,1} = \tau_{sp,1}'$ and $\tau_{sp,2} \neq \tau_{sp,2}'$.
As a result, $\spProgEval{p}{\bp}{\sigma} \neq \spProgEval{p}{\bp}{\sigma'}$.
\end{compactitem}
This completes the proof our claim.

\para{Constructing the prediction oracle $\bp$}
Here we prove our claim that there is a prediction oracle $\bp$ with speculative window at most $w$ such that 
(1) $\tup{0,\sigma,\emptysequence, \emptysequence}  \speval{p}{\bp}{\nu}^{*} sp_1$ and $sp_1.\sigma = am_1.\sigma$ and $INV2(sp_1,am_1)$, and
(2) $\tup{0,\sigma',\emptysequence, \emptysequence}  \speval{p}{\bp}{\nu}^{*} sp_1'$ and $sp_1'.\sigma = am_1'.\sigma$ and $INV2(sp_1',am_1')$, and 
(3) $sp_1 \cong sp_1'$.

We denote by $RB$ the set of transaction identifiers that occur $am_1.s$.
Observe that from $\tup{0,\sigma,\emptysequence}  \ameval{p}{w}{\tau}^{n} am_1$ and $\tup{0,\sigma,\emptysequence}  \ameval{p}{w}{\tau}^{n} am_1'$ it follows that $RB$ is also the set of transaction identifiers occurring in $am_1'.s$.
The set $RB$ identifies which branch instructions we must mispredict to reach the configuration $am_1$.

We now construct a list of predictions $\tup{\lbl,w}$.
In the following, we write $am_0$ instead of $\tup{0,\sigma,\emptysequence}$.
We construct the list by analyzing the execution  $am_0  \ameval{p}{w}{\tau}^{n} am_1$ (observe that deriving the list from  $\tup{0,\sigma',\emptysequence}  \ameval{p}{w}{\tau}^{n} am_1'$ would result in the same list).
The list $L$ is iteratively constructed as follows. 
During the execution, we denote by $\idx$ the position of the configuration in the execution from which we should continue the construction.
Initially, $\idx = 0$.
At each iteration, we find the smallest $\idx' \in \mathbb{N}$ such that $\idx \leq \idx' < n$ and the extended configuration $am_{\idx'}$ is such that $am_0 \ameval{p}{w}{\tau_1}^{\idx'} am_{\idx'}$, and $p(am_{\idx'}.\sigma(\pc)) = \pjz{x}{\lbl}$ and $\exhaustedTop{am_{\idx'}.s}$ (i.e., the next step in the execution is obtained using the \textsc{Se-Branch} rule).
Next, we proceed as follows:
\begin{compactitem}
	\item[$\mathbf{am_{\idx'}.\ctr \in RB}$]
	In this case, we mispredict the outcome of the corresponding branch instruction.
	The prediction associated with the branch instruction  is $\tup{\lbl'',w}$, where $\lbl'' = am_{\idx'}.\sigma(\pc) +1$ if $am_{\idx'}.\sigma(x) = 0$ and $\lbl'' = \lbl$ otherwise.
	Concretely, we update $L$ by appending $\tup{\lbl'',w}$ and setting $\idx$ to $\idx' +1$ (i.e., we continue constructing $L$ by analyzing the next step in the execution). 
	
	\item[$\mathbf{am_{\idx'}.\ctr \not\in RB}$]
	In this case, we correctly predict the outcome of the corresponding branch instruction.
	The prediction associated with the branch instruction is $\tup{\lbl'',0}$, where $\lbl'' = \lbl$ if $am_{\idx'}.\sigma(x) = 0$ and $\lbl'' = am_{\idx'}.\sigma(\pc) +1$ otherwise.
	Concretely, we update $L$ by appending $\tup{\lbl'',0}$ and we update the index $\idx$ to  $\idx' + k +2$ where the value $k \in \Nat$ is such that $am_0 \ameval{p}{w}{\tau_1}^{\idx'} am_{\idx'} \ameval{p}{w}{\startObs{am_{\idx'}.\ctr} \concat \pcObs{\lbl'} } am_{\idx'+1} \ameval{p}{w}{\nu}^{k} am_{\idx'+k+1} \ameval{p}{w}{\rollbackObs{am_{\idx'}.\ctr} \concat \pcObs{\lbl''} } am_{\idx'+k+2}$.  
	That is, we ignore the portion of the execution associated with the speculative transaction $am_{\idx'}.\ctr$ and we continue constructing $L$ after the transaction   $am_{\idx'}.\ctr$ has been rolled back.
	Observe also that $\idx' + k +2 \leq n$ (otherwise $am_{\idx'}.\ctr$ would have been included into $RB$).
\end{compactitem}
Finally, if the configuration $am_1$ in $\tup{0,\sigma,\emptysequence}  \ameval{p}{w}{\tau}^{n} am_1$ is such that $p(am_{1}.\sigma(\pc)) = \pjz{x}{\lbl}$ and $\exhaustedTop{am_{1}.s}$, we append to $L$ the prediction $\tup{\lbl,0}$.

The prediction oracle $\bp$ is defined as follows, where $\#(h)$ is the number of unique transaction identifiers in $h$:
\[
	\bp(\sigma, p, h) := \begin{cases}
 		\tup{\lbl',w'} & \text{if}\ \#(h)+1 \leq |L| \wedge \elt{L}{\#(h)+1} = \tup{\lbl',w'} \wedge ( p(\sigma(\pc)) = \pjz{x}{\lbl'} \vee \lbl' = \sigma(\pc) +1) \\
 		\tup{\lbl,0} & \text{otherwise (where}\ p(\sigma(\pc)) = \pjz{x}{\lbl})
 \end{cases}
\]
Observe that $\bp$ is a prediction oracle as it depends only on the program counter, the program itself, and the length of the history.
Moreover, its speculative window is at most $w$ (since the predictions contain either $w$ or $0$ as speculative window).
The key property of $\bp$ is that ``when the execution starts from $\sigma$ (or $\sigma'$), $\bp$ follows the predictions in $L$ for the first $|L|$ branch instructions and afterwards always predicts the branch as taken''.

From $\tup{0,\sigma,\emptysequence}  \ameval{p}{w}{\tau}^{n} am_1$, $\tup{0,\sigma,\emptysequence}  \ameval{p}{w}{\tau}^{n} am_1'$, and the way in which we constructed $\bp$, we have that (1) $\tup{0,\sigma,\emptysequence, \emptysequence}  \speval{p}{\bp}{\nu}^{*} sp_1$ and $sp_1.\sigma = am_1.\sigma$ and $INV2(sp_1,am_1)$, and
(2) $\tup{0,\sigma',\emptysequence, \emptysequence}  \speval{p}{\bp}{\nu}^{*} sp_1'$ and $sp_1'.\sigma = am_1'.\sigma$ and $INV2(sp_1',am_1')$, and (3) $sp_1 \cong sp_1'$.
The proof of this statement is similar to the proof of Lemma~\ref{lemma:always-mispredict-worstcase:am-to-sp}.
\end{proof}

%% file: correctness-symbolic-semantics.tex
\section{Relating concrete and symbolic always mispredict semantics (Proposition~\ref{proposition:symbolic-execution})}\label{appendix:concorete-and-symbolic}

Below, we provide the proof of Proposition~\ref{proposition:symbolic-execution}, which we restate here for simplicity:

\alwaysMispredSymbolic*

\begin{proof}
The proposition follows from Propositions~\ref{propostion:alwasy-mispred-symbolic:soundness} (proved in Section~\ref{apx:alwasy-mispred-symbolic:soundness}) and~\ref{propostion:alwasy-mispred-symbolic:completeness} (proved in Section~\ref{apx:alwasy-mispred-symbolic:completeness}). 

We now prove the two directions:
\begin{compactitem}
\item[$\bm{(\Rightarrow):}$]
Let $\tup{\sigma, \tau, \sigma'}$ be a run in $\amTraces{p}{w}$.
From this, it follows that $\tup{0,\sigma,\emptysequence} \ameval{p}{w}{\tau}^* \tup{\ctr, \sigma', \emptysequence}$.
From Proposition~\ref{propostion:alwasy-mispred-symbolic:completeness} (proved in Section~\ref{apx:alwasy-mispred-symbolic:completeness}), there is a valuation $\mu$, a symbolic trace $\tau'$, and a final symbolic configuration $\sigma''$ such that  $\tup{0,\tup{sm_0,sa_0},\emptysequence} \symSpeval{p}{w}{\tau'}^* \tup{\ctr, \sigma'', \emptysequence}$, $\mu(\tau') = \tau$, $\mu(\tup{sm_0,sa_0}) = \sigma$, $\mu(\sigma'') = \sigma'$, and $\mu \models \pathCond{\tau'}$.
Hence, $\tup{\tup{sm_0,sa_0}, \tau', \sigma''} \in \symTraces{p}{w}$ and $\tup{\mu(\tup{sm_0,sa_0}), \mu(\tau'), \mu(\sigma'')} \in \gamma(\tup{\tup{sm_0,sa_0}, \tau', \sigma''})$ (since $\mu \models \pathCond{\tau'}$).
Therefore, $\tup{\mu(\tup{sm_0,sa_0}), \mu(\tau'), \mu(\sigma'')} \in \gamma(\symTraces{p}{w})$.
From this and $\mu(\tau') = \tau$, $\mu(\tup{sm_0,sa_0}) = \sigma$, $\mu(\sigma'') = \sigma'$, we have $\tup{\sigma, \tau, \sigma'} \in \gamma(\symTraces{p}{w})$.

\item[$\bm{(\Leftarrow):}$]
Let $\tup{\sigma, \tau, \sigma'}$ be a run in $\gamma(\symTraces{p}{w})$.
From this, it follows that there is a symbolic run $\tup{\sigma_s, \tau_s, \sigma_s'} \in \symTraces{p}{w}$ and a valuation $\mu$ such that $\mu \models \pathCond{\tau_s}$ and $\mu(\sigma_s) = \sigma$, $\mu(\sigma_s') = \sigma'$, and $\mu(\tau_s) = \tau$.
From $\tup{\sigma_s, \tau_s, \sigma_s'} \in \symTraces{p}{w}$, we have that $\tup{0, \sigma_s, \emptysequence} \symSpeval{p}{w}{\tau_s}^* \tup{\ctr, \sigma_s', \emptysequence}$.
From Proposition~\ref{propostion:alwasy-mispred-symbolic:soundness} (proved in Section~\ref{apx:alwasy-mispred-symbolic:soundness}) and $\mu \models \pathCond{\tau_s}$, we have that $\tup{0, \mu(\sigma_s), \emptysequence} \ameval{p}{w}{\mu(\tau_s)}^* \tup{\ctr, \mu(\sigma_s'), \emptysequence}$.
Therefore, $\tup{\mu(sigma_s) ,\mu(\tau_s), \mu(\sigma_s')} \in \amTraces{p}{w}$.
From this and $\mu(\sigma_s) = \sigma$, $\mu(\sigma_s') = \sigma'$, and $\mu(\tau_s) = \tau$, we have $\tup{\sigma, \tau, \sigma'} \in \amTraces{p}{w}$.

\end{compactitem}
This completes the proof of our claim.
\end{proof}

\subsection{Auxiliary definitions} 

\newcommand{\wf}[1]{\mathit{wf}(#1)}

Given a symbolic speculative state $\tup{\sigma, \id, w, \lbl} \in \specStates$, its \textit{next-step mispredict path condition} $\pathCond{\tup{\sigma, \id, w, \lbl}}$ is the symbolic expression $\pathCond{\tau}$ such that $\sigma \symEval{p}{\tau} \sigma'$ and $\sigma'(\pc) \neq \lbl$ (if such a $\tau$ does not exists, then $\pathCond{s} = \top$).
Given a sequence of symbolic speculative states $s \in \specStates^*$, its \textit{next-step mispredict path condition} $\pathCond{s} = \bigwedge_{\tup{\sigma, \id, w, \lbl} \in s} \pathCond{\tup{\sigma, \id, w, \lbl}}$.

We say that a concrete speculative state $\tup{\sigma, \id, w, \lbl} \in \specStates$ is  \textit{next-step mispredict well-formed}, written $\wf{\tup{\sigma, \id, w, \lbl}}$, if $\sigma \symEval{p}{\tau} \sigma'$ and $\sigma'(\pc) \neq \lbl$.
We say that a sequence of speculative states $s \in \specStates^*$ is \textit{next-step mispredict well-formed}, written $\wf{s}$, if $\bigwedge_{\tup{\sigma, \id, w, \lbl} \in s} \wf{\tup{\sigma, \id, w, \lbl}}$.

\subsection{Auxiliary lemmas}

Lemma~\ref{lemma:symbolic:useful} provides some useful helping facts, which we will use throughout the soundness proof.

\begin{restatable}{lem}{}\label{lemma:symbolic:useful}
The following facts hold:
\begin{compactenum}
	\item Given a symbolic assignment $sa$ and a valuation $\mu : \SymVal \to \Val$, then $\mu(sa)(x) = \mu(sa(x))$ for all $x \in \Var$.
	\item Given a symbolic memory $sa$ and a valuation $\mu : \SymVal \to \Val$, then $\mu(sm)(n) = \mu(sm(n))$ for all $n \in \Addr$.
	\item Given a symbolic assignment $sa$, an expression $e$, and a valuation $\mu : \SymVal \to \Val$, then $\mu(\exprEval{e}{sa}) = \exprEval{e}{\mu(sa)}$.
\end{compactenum}	
\end{restatable}

\begin{proof}
Below we prove all our claims.

\para{Proof of (1)}
Let $sa$ be a symbolic assignment and $\mu : \SymVal \to \Val$.
Then, $\mu(sa)$ is defined as $\mu \circ sa$.
Hence, $\mu(sa)(x) = \mu(sa(x))$.

\para{Proof of (2)}
Identical to (1).

\para{Proof of (3)}
Let $sa$ be a symbolic assignment, $e$ be an expression, and $\mu : \SymVal \to \Val$.
We show, by structural induction on $e$, that $\mu(\exprEval{e}{sa}) = \exprEval{e}{\mu(sa)}$.
For the base case, there are two possibilities.
If $e$ is a value $n \in \Val$, then $\exprEval{n}{sa} = n$, $\exprEval{n}{\mu(sa)} = n$, and $\mu(n) = n$.
Hence, $\mu(\exprEval{e}{sa}) = \exprEval{e}{\mu(sa)}$.
If $e$ is a register identifier $x \in \Var$, then $\exprEval{x}{\mu(sa)} = \mu(sa)(x)$ and $\mu(\exprEval{e}{sa}) = \mu(sa(x))$.
From this and (1), we get $\mu(\exprEval{e}{sa}) = \exprEval{e}{\mu(sa)}$.
For the induction step, there two cases.
If $e$ is of the form $\ominus e'$, then $\exprEval{\unaryOp{e'}}{\mu(sa)} = \mathit{apply}(\ominus, \exprEval{e'}{\mu(sa)})$, whereas $\mu(\exprEval{\unaryOp{e'}}{sa})$ is either $\mu(\mathit{apply}(\ominus, \exprEval{e'}{sa}))$ if $\exprEval{e'}{sa} \in \Val$ or $\mu( \unaryOp{\exprEval{e'}{sa}})$ otherwise.
In the first case, $\mu(\mathit{apply}(\ominus, \exprEval{e'}{\mu(sa)})) = \mathit{apply}(\ominus, \exprEval{e'}{sa})$ (since we are dealing only with concrete values) and $\exprEval{e'}{\mu(sa)} = \exprEval{e'}{sa}$ since $\exprEval{e'}{sa} \in \Val$.
Therefore, $\mu(\exprEval{e}{sa}) = \exprEval{e}{\mu(sa)}$.
In the second case, $\mu( \unaryOp{\exprEval{e'}{sa}}) = \mathit{apply}(\ominus, \mu(\exprEval{e'}{sa}))$.
By applying the induction hypothesis ($\exprEval{e'}{\mu(sa)} = \mu(\exprEval{e'}{sa}$), we obtain $\mu( \unaryOp{\exprEval{e'}{sa}}) = \mathit{apply}(\ominus, \exprEval{e'}{\mu(sa)})$.
Hence, $\mu(\exprEval{e}{sa}) = \exprEval{e}{\mu(sa)}$.
The proof for the second case (i.e., $e = \binaryOp{e_1}{e_2}$) is similar.
\end{proof}

\subsection{Soundness}\label{apx:alwasy-mispred-symbolic:soundness}

In Lemma~\ref{lemma:symbolic:one-step-soundness-non-speculative}, we prove the soundness of the symbolic non-speculative semantics.

\begin{restatable}{lem}{}\label{lemma:symbolic:one-step-soundness-non-speculative}
Let $p$ be a program.
Whenever $\tup{sm,sa} \symEval{p}{\tau} \tup{sm',sa'}$, then for all mappings $\mu: \SymVal \to \Val$ such that $\mu \models \pathCond{\tau}$, $\tup{\mu(sm),\mu(sa)} \eval{p}{\mu(\tau)} \tup{\mu(sm'),\mu(sa')}$.
\end{restatable}

\begin{proof}
Let $p$ be a program.
Assume that $\tup{sm,sa} \symEval{p}{\tau} \tup{sm',sa'}$ holds.
We proceed by case distinction on the rules defining $\symEval{p}{}$.
\begin{compactitem}
\item[\textbf{Rule \textsc{Skip}.}]
Assume that $\tup{sm,sa} \symEval{p}{\tau} \tup{sm',sa'}$ has been derived using the \textsc{Skip} rule in the symbolic semantics.
Then, $p(a(\pc)) = \pskip$, $sm' = sm$, $\tau = \emptysequence$, and $sa' = sa[\pc \mapsto sa(\pc)+1]$.
Let $\mu : \SymVal \to \Val$ be an arbitrary valuation satisfying $\pathCond{\tau}$ (since $\pathCond{\tau}$ is $\top$ all valuations are models of $\pathCond{\tau}$).
Observe that $\mu(sa)(\pc) = sa(\pc)$ since the program counter is always concrete.
Hence, $p(\mu(sa)(\pc)) = \pskip$.
Thus, we can apply the rule \textsc{Skip} of the concrete semantics starting from $\tup{\mu(sm),\mu(sa)}$.
Hence, $\tup{\mu(sm),\mu(sa)} \eval{p}{} \tup{\mu(sm),\mu(sa)[\pc \mapsto \mu(sa)(\pc)+1]}$.
From this, $\mu(\tau) = \emptysequence$ (since $\tau = \emptysequence$), and $sa' = sa[\pc \mapsto sa(\pc)+1]$, we have that $\tup{\mu(sm),\mu(sa)} \eval{p}{\mu(\tau)} \tup{\mu(sm'),\mu(sa')}$.

\item[\textbf{Rule \textsc{Barrier}.}]
The proof of this case is analogous to that of \textsc{Skip}.

\item[\textbf{Rule \textsc{Assign}.}]
Assume that $\tup{sm,sa} \symEval{p}{\tau} \tup{sm',sa'}$ has been derived using the \textsc{Assign} rule in the symbolic semantics.
Then, $p(a(\pc)) = \passign{x}{e}$, $x \neq \pc$, $sm' = sm$, $\tau = \emptysequence$, and $sa' = sa[\pc \mapsto sa(\pc)+1, x \mapsto \exprEval{e}{sa}]$.
Let $\mu : \SymVal \to \Val$ be an arbitrary valuation (since $\pathCond{\tau}$ is $\top$ all valuations are models of $\pathCond{\tau}$).
Observe that $\mu(sa)(\pc) = sa(\pc)$ since the program counter is always concrete.
Hence, $p(\mu(sa)(\pc)) =  \passign{x}{e}$.
Thus, we can apply the rule \textsc{Assign} of the concrete semantics starting from $\tup{\mu(sm),\mu(sa)}$.
Hence, $\tup{\mu(sm),\mu(sa)} \eval{p}{} \tup{\mu(sm),\mu(sa)[\pc \mapsto \mu(sa)(\pc)+1, x \mapsto \exprEval{e}{\mu(sa)}]}$.
From this, Lemma~\ref{lemma:symbolic:useful}, $\mu(\tau) = \emptysequence$ (since $\tau = \emptysequence$), and $sa' = sa[\pc \mapsto sa(\pc)+1, x \mapsto \exprEval{e}{sa}]$, we have that $\tup{\mu(sm),\mu(sa)} \eval{p}{\mu(\tau)} \tup{\mu(sm'),\mu(sa')}$.

\item[\textbf{Rule \textsc{ConditionalUpdate-Concr-Sat}.}]
Assume that $\tup{sm,sa} \symEval{p}{\tau} \tup{sm',sa'}$ has been derived using the \textsc{ConditionalUpdate-Concr-Sat} rule in the symbolic semantics.
Then, $p(a(\pc)) = \pcondassign{x}{e}{e'}$, $x \neq \pc$, $sm' = sm$, $\tau = \emptysequence$, $\exprEval{e'}{sa} = 0$, and $sa' = sa[\pc \mapsto sa(\pc)+1, x \mapsto \exprEval{e}{sa}]$.
Let $\mu : \SymVal \to \Val$ be an arbitrary valuation (since $\pathCond{\tau}$ is $\top$ all valuations are models of $\pathCond{\tau}$).
Observe that $\mu(sa)(\pc) = sa(\pc)$ since the program counter is always concrete.
Hence, $p(\mu(sa)(\pc)) =  \pcondassign{x}{e}{e'}$.
Thus, we can apply one of the rules  \textsc{ConditionalUpdate-Sat} or  \textsc{ConditionalUpdate-Unsat} of the concrete semantics starting from $\tup{\mu(sm),\mu(sa)}$.
We claim that  $\exprEval{e'}{\mu(sa)} = 0$.
Therefore, we can apply the rule \textsc{ConditionalUpdate-Sat}.
Hence, $\tup{\mu(sm),\mu(sa)} \eval{p}{} \tup{\mu(sm),\mu(sa)[\pc \mapsto \mu(sa)(\pc)+1, x \mapsto \exprEval{e}{\mu(sa)}]}$.
From this, Lemma~\ref{lemma:symbolic:useful}, $\mu(\tau) = \emptysequence$ (since $\tau = \emptysequence$), and $sa' = sa[\pc \mapsto sa(\pc)+1, x \mapsto \exprEval{e}{sa}]$, we have that $\tup{\mu(sm),\mu(sa)} \eval{p}{\mu(\tau)} \tup{\mu(sm'),\mu(sa')}$.

We now show that  $\exprEval{e'}{\mu(sa)}$ is indeed $0$.
Since $\exprEval{e'}{sa} = 0$, the value of $e'$ depends only on registers in $sa$ whose value is concrete.
Hence, for all these registers $x$, $\mu(sa)(x) = sa(x)$.
Thus, $\exprEval{e'}{\mu(sa)} = \exprEval{e'}{sa} = 0$.

\item[\textbf{Rule \textsc{ConditionalUpdate-Concr-Unsat}.}]
The proof of this case is analogous to that of \textsc{ConditionalUpdate-Concr-Unsat}.

\item[\textbf{Rule \textsc{ConditionalUpdate-Symb}.}]
Assume that $\tup{sm,sa} \symEval{p}{\tau} \tup{sm',sa'}$ has been derived using the \textsc{ConditionalUpdate-Symb} rule in the symbolic semantics.
Then, $p(a(\pc)) = \pcondassign{x}{e}{e'}$, $x \neq \pc$, $sm' = sm$, $\tau = \emptysequence$, $\exprEval{e'}{sa} = se$, $se \not\in \Val$, and $sa' = sa[\pc \mapsto sa(\pc)+1, x \mapsto \ite{se = 0}{ \exprEval{e}{sa} }{sa(x)}]$.
Let $\mu : \SymVal \to \Val$ be an arbitrary valuation (since $\pathCond{\tau}$ is $\top$ all valuations are models of $\pathCond{\tau}$).
Observe that $\mu(sa)(\pc) = sa(\pc)$ since the program counter is always concrete.
Hence, $p(\mu(sa)(\pc)) =  \pcondassign{x}{e}{e'}$.

There are two cases:
\begin{compactitem}
\item[$\bm{\mu \models se =0}$:]
Then, $\mu(se) = 0$ and therefore $\exprEval{e'}{\mu(sa)} = 0$.
Hence, we can apply the rule \textsc{ConditionalUpdate-Sat} starting from $\tup{\mu(sm),\mu(sa)}$.
Therefore, we get $\tup{\mu(sm),\mu(sa)} \eval{p}{} \tup{\mu(sm),\mu(sa)[\pc \mapsto \mu(sa)(\pc)+1, x \mapsto \exprEval{e}{\mu(sa)}]}$.
To prove our claim, we have to show that $\mu(sa') = \mu(sa)[\pc \mapsto \mu(sa)(\pc)+1, x \mapsto \exprEval{e}{\mu(sa)}]$.
From  $sa' = sa[\pc \mapsto sa(\pc)+1, x \mapsto \ite{se = 0}{ \exprEval{e}{sa} }{sa(x)}]$, it is enough to show that 
$\mu(\ite{se = 0}{ \exprEval{e}{sa} }{sa(x)}) = \exprEval{e}{\mu(sa)}$.
Since $\mu \models se =0$, $\mu(\ite{se = 0}{ \exprEval{e}{sa} }{sa(x)})$ is equivalent to $\mu(\exprEval{e}{sa})$.
From this and Lemma~\ref{lemma:symbolic:useful}, $\mu(\exprEval{e}{sa}) =\exprEval{e}{\mu(sa)}$.
Hence, we have that $\tup{\mu(sm),\mu(sa)} \eval{p}{\mu(\tau)} \tup{\mu(sm'),\mu(sa')}$.

\item[$\bm{\mu \not\models se =0}$:]
Then, $\mu(se) \neq 0$ and therefore $\exprEval{e'}{\mu(sa)} \neq 0$.
Hence, we can apply the rule \textsc{ConditionalUpdate-Unsat} starting from $\tup{\mu(sm),\mu(sa)}$.
Therefore, we get $\tup{\mu(sm),\mu(sa)} \eval{p}{} \tup{\mu(sm),\mu(sa)[\pc \mapsto \mu(sa)(\pc)+1]}$.
To prove our claim, we have to show that $\mu(sa') = \mu(sa)[\pc \mapsto \mu(sa)(\pc)+1]$.
From  $sa' = sa[\pc \mapsto sa(\pc)+1, x \mapsto \ite{se = 0}{ \exprEval{e}{sa} }{sa(x)}]$, it is enough to show that 
$\mu(\ite{se = 0}{ \exprEval{e}{sa} }{sa(x)}) = \mu(sa)(x)$.
Since $\mu \not\models se =0$, $\mu(\ite{se = 0}{ \exprEval{e}{sa} }{sa(x)})$ is equivalent to $\mu(sa(x))$.
From this and Lemma~\ref{lemma:symbolic:useful}, $\mu(sa(x)) = \mu(sa)(x)$.
Hence, we have that $\tup{\mu(sm),\mu(sa)} \eval{p}{\mu(\tau)} \tup{\mu(sm'),\mu(sa')}$.
\end{compactitem}
This completes the proof of this case.

\item[\textbf{Rule \textsc{Load-Symb}.}]
Assume that $\tup{sm,sa} \symEval{p}{\tau} \tup{sm',sa'}$ has been derived using the \textsc{Load-Symb} rule in the symbolic semantics.
Then, $p(a(\pc)) = \pload{x}{e}$, $x \neq \pc$, $sm' = sm$, $se = \exprEval{e}{sa}$, $\tau = \loadObs{se}$, and $sa' = sa[\pc \mapsto sa(\pc)+1, x \mapsto \symRead{sm}{se} ]$.
Let $\mu : \SymVal \to \Val$ be an arbitrary valuation (since $\pathCond{\tau} = \top$ all valuations $\mu$ are valid models).
Observe that $\mu(sa)(\pc) = sa(\pc)$ since the program counter is always concrete.
Hence, $p(\mu(sa)(\pc)) =  \pload{x}{e}$ and $x \neq \pc$.
Thus, we can apply the rule \textsc{Load} of the concrete semantics starting from $\tup{\mu(sm),\mu(sa)}$.
Hence, $\tup{\mu(sm),\mu(sa)} \eval{p}{\loadObs{ \exprEval{e}{\mu(sa)} } } \tup{\mu(sm),\mu(sa)[\pc \mapsto \mu(sa)(\pc)+1, x \mapsto \mu(sm)(\exprEval{e}{\mu(sa)})]}$.

Observe that $\mu(\exprEval{e}{sa}) = \exprEval{e}{\mu(sa)}$ (from Lemma~\ref{lemma:symbolic:useful}), and therefore $\mu(\loadObs{\exprEval{e}{sa}}) = \loadObs{ \exprEval{e}{\mu(sa)} }$.
%
We claim that $\mu(\symRead{sm}{se}) = \mu(sm)(\exprEval{e}{\mu(sa)})$.
From this, $\mu(sa)[\pc \mapsto \mu(sa)(\pc)+1, x \mapsto \mu(sm)(\exprEval{e}{\mu(sa)})] = \mu(sa')$.
Therefore, $\tup{\mu(sm),\mu(sa)} \eval{p}{\mu(\tau)} \tup{\mu(sm'),\mu(sa')}$.

We now show that $\mu(\symRead{sm}{se}) = \mu(sm)(\exprEval{e}{\mu(sa)})$.
From the evaluation of symbolic expressions, we have that $\mu(\symRead{sm}{se})$ is equivalent to $\mu(sm)(\mu(se))$.
From this and  $se = \exprEval{e}{sa}$, we have $\mu(\symRead{sm}{se}) = \mu(sm)(\mu(\exprEval{e}{sa}))$.
From this and Lemma~\ref{lemma:symbolic:useful}, we get $\mu(\symRead{sm}{se}) = \mu(sm)(\exprEval{e}{\mu(sa)})$.

\item[\textbf{Rule \textsc{Store-Symb}.}]
Assume that $\tup{sm,sa} \symEval{p}{\tau} \tup{sm',sa'}$ has been derived using the \textsc{Store-Symb} rule in the symbolic semantics.
Then, $p(a(\pc)) = \pstore{x}{e}$, $se = \exprEval{e}{sa}$,  $sm' = \symWrite{sm}{se}{sa(x)}$, $\tau = \storeObs{se}$, and $sa' = sa[\pc \mapsto sa(\pc)+1]$.
Let $\mu : \SymVal \to \Val$ be an arbitrary valuation (since $\pathCond{\tau} = \top$ all valuations are valid models).
Observe that $\mu(sa)(\pc) = sa(\pc)$ since the program counter is always concrete.
Hence, $p(\mu(sa)(\pc)) =  \pstore{x}{e}$.
Thus, we can apply the rule \textsc{Store} of the concrete semantics starting from $\tup{\mu(sm),\mu(sa)}$.
Hence, $\tup{\mu(sm),\mu(sa)} \eval{p}{\storeObs{ \exprEval{e}{\mu(sa)} } } \tup{\mu(sm)[\exprEval{e}{\mu(sa)} \mapsto \mu(sa)(x)],\mu(sa)[\pc \mapsto \mu(sa)(\pc)+1]}$.
Observe that (1) $\mu(\exprEval{e}{sa}) = \exprEval{e}{\mu(sa)}$ (from Lemma~\ref{lemma:symbolic:useful}) and therefore $\mu(\storeObs{\exprEval{e}{sa}}) = \storeObs{ \exprEval{e}{\mu(sa)} }$, and (2) $\mu(sa)[\pc \mapsto \mu(sa)(\pc)+1] = \mu(sa[\pc \mapsto sa(\pc)+1])$.
%
Moreover, we claim that $\mu(sm)[\exprEval{e}{\mu(sa)} \mapsto \mu(sa)(x)] = \mu(\symWrite{sm}{se}{sa(x)})$.
Therefore, $\tup{\mu(sm),\mu(sa)} \eval{p}{\mu(\tau)} \tup{\mu(sm'),\mu(sa')}$.

We now prove our claim that  $\mu(sm)[\exprEval{e}{\mu(sa)} \mapsto \mu(sa)(x)] = \mu(\symWrite{sm}{se}{sa(x)})$.
From the evaluation of symbolic expressions, we have that $\mu(\symWrite{sm}{se}{sa(x)})$ is equivalent to $\mu(sm)[\mu(se) \mapsto \mu(sa(x))]$.
From $se = \exprEval{e}{sa}$, we have  $\mu(\symWrite{sm}{se}{sa(x)}) = \mu(sm)[\mu(\exprEval{e}{sa}) \mapsto \mu(sa(x))]$.
From Lemma~\ref{lemma:symbolic:useful}, we have  $\mu(\symWrite{sm}{se}{sa(x)}) = \mu(sm)[\exprEval{e}{\mu(sa)} \mapsto \mu(sa)(x)]$.

\item[\textbf{Rule \textsc{Beqz-Concr-Sat}.}]
Assume that $\tup{sm,sa} \symEval{p}{\tau} \tup{sm',sa'}$ has been derived using the \textsc{Beqz-Concr-1} rule in the symbolic semantics.
Then, $p(a(\pc)) = \pjz{x}{\lbl}$, $sa(x) = 0$, $sa(x) \in \Val$, $sm' = sm$, $\tau = \symPcObs{\top} \concat \pcObs{\lbl}$, and $sa' = sa[\pc \mapsto \lbl]$.
Let $\mu : \SymVal \to \Val$ be an arbitrary valuation (since $\pathCond{\tau} = \top$ all valuations are valid models).
Observe that $\mu(sa)(\pc) = sa(\pc)$ since the program counter is always concrete.
Hence, $p(\mu(sa)(\pc)) =  \pjz{x}{\lbl}$.
Thus, we can apply one of the rules \textsc{Beqz-Sat} or \textsc{Beqz-Unsat} of the concrete semantics starting from $\tup{\mu(sm),\mu(sa)}$.
We claim that $\mu(sa)(x) = 0$.
Therefore, we can apply the \textsc{Beqz-Sat} rule.
Hence, $\tup{\mu(sm),\mu(sa)} \eval{p}{\pcObs{ \lbl } } \tup{\mu(sm),\mu(sa)[\pc \mapsto \lbl]}$.
Therefore, $\tup{\mu(sm),\mu(sa)} \eval{p}{\mu(\tau)} \tup{\mu(sm'),\mu(sa')}$ (since $\mu(\tau) = \mu(\symPcObs{\top} \concat \pcObs{\lbl}) = \pcObs{\lbl}$).

We now prove our claim that $\mu(sa)(x) = 0$.
We know that $sa(x) = 0$.
Hence, $\mu(sa(x)) = 0$ (since $sa(x)$ is a concrete value).
From this and Lemma~\ref{lemma:symbolic:useful}, $\mu(sa)(x) = 0$.

\item[\textbf{Rule \textsc{Beqz-Symb-Sat}.}]
Assume that $\tup{sm,sa} \symEval{p}{\tau} \tup{sm',sa'}$ has been derived using the \textsc{Beqz-Symb-Sat} rule in the symbolic semantics.
Then, $p(a(\pc)) = \pjz{x}{\lbl}$, $sa(x) \not\in \Val$, $sm' = sm$, $\tau = \symPcObs{sa(x) = 0} \concat \pcObs{\lbl}$,  and $sa' = sa[\pc \mapsto \lbl]$.
Let $\mu : \SymVal \to \Val$ be an arbitrary valuation satisfying $\pathCond{\tau} = sa(x) = 0$.
Observe that $\mu(sa)(\pc) = sa(\pc)$ since the program counter is always concrete.
Hence, $p(\mu(sa)(\pc)) =  \pjz{x}{\lbl}$.
Thus, we can apply one of the rules \textsc{Beqz-Sat} or \textsc{Beqz-Unsat} of the concrete semantics starting from $\tup{\mu(sm),\mu(sa)}$.
We claim that $\mu(sa)(x) = 0$.
Therefore, we can apply the \textsc{Beqz-Sat} rule.
Hence, $\tup{\mu(sm),\mu(sa)} \eval{p}{\pcObs{ \lbl } } \tup{\mu(sm),\mu(sa)[\pc \mapsto \lbl]}$.
Therefore, $\tup{\mu(sm),\mu(sa)} \eval{p}{\mu(\tau)} \tup{\mu(sm'),\mu(sa')}$ (since $\mu(\tau) = \mu(\symPcObs{sa(x) = 0} \concat \pcObs{\lbl}) = \pcObs{\lbl}$).

We now prove our claim that $\mu(sa)(x) = 0$.
We know that $sa(x) = se$ and $\mu \models se = 0$.
Hence, $\mu(se) = 0$.
Therefore, $\mu(sa(x)) = 0$.
From this and Lemma~\ref{lemma:symbolic:useful}, we have $\mu(sa)(x) = 0$.

\item[\textbf{Rule \textsc{Beqz-Concr-Unsat}.}]
The proof of this case is similar to that of the \textsc{Beqz-Concr-Sat} rule.

\item[\textbf{Rule \textsc{Beqz-Symb-Unsat}.}]
The proof of this case is similar to that of the \textsc{Beqz-Symb-Sat} rule.

\item[\textbf{Rule \textsc{Jmp-Concr}.}]
Assume that $\tup{sm,sa} \symEval{p}{\tau} \tup{sm',sa'}$ has been derived using the \textsc{Jmp-Concr} rule in the symbolic semantics.
Then, $p(a(\pc)) = \pjmp{e}$,  $\lbl = \exprEval{e}{sa}$, $\lbl \in \Val$, $sm' = sm$, $\tau = \symPcObs{\top} \concat \pcObs{\lbl}$, and $sa' = sa[\pc \mapsto \lbl]$.
Let $\mu : \SymVal \to \Val$ be an arbitrary valuation (since $\pathCond{\tau} = \top$ all valuations are valid models).
Observe that $\mu(sa)(\pc) = sa(\pc)$ since the program counter is always concrete.
Hence, $p(\mu(sa)(\pc)) =  \pjmp{e}$.
Thus, we can apply the \textsc{Jmp} rule of the concrete semantics starting from $\tup{\mu(sm),\mu(sa)}$.
Observe that  $\exprEval{e}{\mu(sa)} = \mu(\exprEval{e}{sa}) = \lbl$.
Hence, $\tup{\mu(sm),\mu(sa)} \eval{p}{\pcObs{ \lbl } } \tup{\mu(sm),\mu(sa)[\pc \mapsto \lbl]}$.
Therefore, $\tup{\mu(sm),\mu(sa)} \eval{p}{\mu(\tau)} \tup{\mu(sm'),\mu(sa')}$ (since $\mu(\tau) = \mu(\symPcObs{\top} \concat \pcObs{\lbl}) = \pcObs{\lbl}$).

\item[\textbf{Rule \textsc{Jmp-Symb}.}]
Assume that $\tup{sm,sa} \symEval{p}{\tau} \tup{sm',sa'}$ has been derived using the \textsc{Jmp-Symb} rule in the symbolic semantics.
Then, $p(a(\pc)) = \pjmp{e}$, $\exprEval{e}{sa} \not \in \Val$, $\lbl \in \Val$, $sm' = sm$, $\tau = \symPcObs{\exprEval{e}{sa} = \lbl} \concat \pcObs{\lbl}$,  and $sa' = sa[\pc \mapsto \lbl]$.
Let $\mu : \SymVal \to \Val$ be an arbitrary valuation satisfying $\pathCond{\tau}$, that is, $\exprEval{e}{sa} = \lbl$.
Observe that $\mu(sa)(\pc) = sa(\pc)$ since the program counter is always concrete.
Hence, $p(\mu(sa)(\pc)) =   \pjmp{e}$.
Thus, we can apply the \textsc{Jmp} rule of the concrete semantics starting from $\tup{\mu(sm),\mu(sa)}$.
We claim $\exprEval{e}{\mu(sa)} = \lbl$.
Hence, $\tup{\mu(sm),\mu(sa)} \eval{p}{\pcObs{ \lbl } } \tup{\mu(sm),\mu(sa)[\pc \mapsto \lbl]}$.
Therefore, $\tup{\mu(sm),\mu(sa)} \eval{p}{\mu(\tau)} \tup{\mu(sm'),\mu(sa')}$ (since $\mu(\tau) = \mu(\symPcObs{\exprEval{e}{sa} = \lbl} \concat \pcObs{\lbl}) = \pcObs{\lbl}$).

We now prove our claim that $\exprEval{e}{\mu(sa)} = \lbl$.
We know that $\exprEval{e}{sa} = se$ and $\mu \models se = \lbl$.
Hence, $\mu(\exprEval{e}{sa}) = \lbl$.
From this and Lemma~\ref{lemma:symbolic:useful}, we have $\exprEval{e}{\mu(sa)} = \lbl$.

\item[\textbf{Rule \textsc{Terminate}.}]
The proof of this case is similar to that of the \textsc{Skip} rule.

\end{compactitem}
This completes the proof of our claim.
\end{proof}

In Lemma~\ref{lemma:symbolic:one-step-soundness-speculative}, we prove the soundness of a single step of the symbolic always mispredict semantics.

\begin{restatable}{lem}{}\label{lemma:symbolic:one-step-soundness-speculative}
Let $p$ be a program. 
Whenever $\tup{\ctr,\tup{sm,sa},s} \symSpeval{p}{\mbp}{\tau} \tup{\ctr',\tup{sm',sa'},s'}$, then for all mappings $\mu: \SymVal \to \Val$ such that $\mu \models \pathCond{s} \wedge \pathCond{\tau}$,  $\tup{\ctr, \tup{\mu(sm),\mu(sa)}, \mu(s)} \ameval{p}{\mbp}{\mu(\tau)} \tup{\ctr',\tup{\mu(sm'),\mu(sa')},\mu(s')}$.
\end{restatable}

\begin{proof}
Let $p \in \Prg$ be a program.
Moreover, assume that $\tup{\ctr,\tup{sm,sa},s} \symSpeval{p}{\mbp}{\tau} \tup{\ctr',\tup{sm',sa'},s'}$.
We proceed by case distinction on the rules defining $ \symSpeval{p}{\mbp}{}$.
\begin{compactitem}
\item[\textbf{Rule \textsc{Am-NoBranch}.}]
Assume that  $\tup{\ctr,\tup{sm,sa},s} \symSpeval{p}{\mbp}{\tau} \tup{\ctr',\tup{sm',sa'},s'}$ has been derived using the \textsc{Se-NoBranch} rule in the symbolic semantics.
Then, $\tup{sm,sa} \symEval{p}{\tau} \tup{sm',sa'}$, $p(sa(\pc)) \neq \pjz{x}{\lbl}$, $\exhaustedTop{s}$, 
$\ctr' = \ctr$, and $s' = \decrementTop{s}$.
Let $\mu : \SymVal \to \Val$ be an arbitrary valuation that satisfies $\pathCond{s} \wedge \pathCond{\tau}$.
Observe that $\mu$ also satisfies $\pathCond{\tau}$.
Therefore, we can apply Lemma~\ref{lemma:symbolic:one-step-soundness-non-speculative} to $\tup{sm,sa} \symEval{p}{\tau} \tup{sm',sa'}$ to derive $\tup{\mu(sm),\mu(sa)} \eval{p}{\mu(\tau)} \tup{\mu(sm'),\mu(sa')}$.
Moreover, we immediately have that $\exhaustedTop{s}$ iff $\exhaustedTop{\mu(s)}$ and  $p(sa(\pc)) \neq \pjz{x}{\lbl}$, iff $p(\mu(sa)(\pc)) \neq \pjz{x}{\lbl}$,.
Thus, we can apply the rule \textsc{Am-NoBranch} of the concrete semantics starting from  $\tup{\ctr,\tup{\mu(sm),\mu(sa)},\mu(s)}$.
Therefore, $\tup{\ctr, \tup{\mu(sm),\mu(sa)}, \mu(s)} \ameval{p}{\mbp}{\mu(\tau)} \tup{\ctr,\tup{\mu(sm'),\mu(sa')},\mu(s')}$.

\item[\textbf{Rule \textsc{Am-Branch}.}]
Assume that  $\tup{\ctr,\tup{sm,sa},s} \symSpeval{p}{\mbp}{\tau} \tup{\ctr',\tup{sm',sa'},s'}$ has been derived using the \textsc{Se-Branch} rule in the symbolic semantics.
Then, $p(sa(\pc)) = \pjz{x}{\lbl''}$,  $\exhaustedTop{s}$, $\tup{sm,sa} \symEval{p}{\symPcObs{se} \concat \pcObs{\lbl'}} \sigma'$, $\lbl = \begin{cases}
 	sa(\pc)+1 & \text{if}\ \lbl' \neq sa(\pc) + 1\\
 	\lbl '' & \text{if}\ \lbl' = sa(\pc) + 1
 \end{cases}$, 
$\ctr'= \ctr+1$, $sm' = sm$, $sa' = sa [\pc \mapsto \lbl]$, $s' = s \concat \tup{\tup{sm,sa}, \ctr, \mathit{min}(w, \window{s} -1), \lbl  }$, and $\tau = \symPcObs{se} \concat \startObs{\ctr} \concat \pcObs{\lbl}$.
Let $\mu : \SymVal \to \Val$ be an arbitrary valuation satisfying $\pathCond{s} \wedge \pathCond{\tau}$.
Observe that $\mu$ also satisfies $\pathCond{\tau}$, i.e., $\mu \models se$.
We can apply  Lemma~\ref{lemma:symbolic:one-step-soundness-non-speculative} to $\tup{sm,sa} \symEval{p}{\symPcObs{se} \concat \pcObs{\lbl'}} \sigma'$ to derive  $\tup{\mu(sm),\mu(sa)} \eval{p}{\mu(\tau)} \mu(\sigma')$.
Observe that $\mu(\sigma')(\pc) = \lbl'$ and $\exhaustedTop{s}$ holds iff $\exhaustedTop{\mu(s)}$ does.
Therefore, we can apply the rule \textsc{Am-Branch} of the concrete semantics starting from  $\tup{\ctr,\tup{\mu(sm),\mu(sa)},\mu(s)}$.
We have $\tup{\ctr, \tup{\mu(sm),\mu(sa)}, \mu(s)} \ameval{p}{\mbp_n}{\startObs{\ctr} \concat \pcObs{\lbl'''}} \tup{\ctr+1,\tup{\mu(sm),\mu(sa)[\pc \mapsto \lbl''']},\mu(s) \concat \tup{\tup{\mu(sm),\mu(sa)}, \ctr, \mathit{min}(w, \window{\mu(s)} -1), \lbl'''}}$.
We claim that $\lbl = \lbl'''$.
Hence, $\tup{\ctr, \tup{\mu(sm),\mu(sa)}, \mu(s)} \ameval{p}{\mbp}{\mu(\tau)} \tup{\ctr',\tup{\mu(sm'),\mu(sa')},\mu(s')}$.

We now show that $\lbl = \lbl'''$.
There are two cases:
\begin{compactitem}
\item[$\bm{\lbl = sa(\pc)+1}:$]
Then, $ \lbl' \neq sa(\pc) + 1$.
There are two cases:
\begin{compactitem}
\item The step $\tup{sm,sa} \symEval{p}{\symPcObs{se} \concat \pcObs{\lbl'}} \sigma'$ has been derived using the \textsc{Beqz-Concr-Sat} rule from the symbolic non-speculative semantics.
Then, $sa(x) = 0$. 
Therefore, $\mu(sa)(x) = 0$ and $\lbl''' = \mu(sa)(\pc) + 1$ by construction.
Hence, $\lbl = \lbl'''$.

\item The step $\tup{sm,sa} \symEval{p}{\symPcObs{se} \concat \pcObs{\lbl'}} \sigma'$ has been derived using the \textsc{Beqz-Symb-Sat} rule from the symbolic non-speculative semantics.
Then, $se$ is $sa(x) = 0$.
From this and $\mu \models se$, we have that  $\mu(sa)(x) = 0$ and, therefore, $\lbl''' = \mu(sa)(\pc) + 1$ by construction.
Hence, $\lbl = \lbl'''$.
\end{compactitem}
This completes the proof of this case.

\item[$\bm{\lbl = \lbl''}:$]
Then, $ \lbl' = sa(\pc) + 1$.
There are two cases:
\begin{compactitem}
\item The step $\tup{sm,sa} \symEval{p}{\symPcObs{se} \concat \pcObs{\lbl'}} \sigma'$ has been derived using the \textsc{Beqz-Concr-Unsat} rule from the symbolic non-speculative semantics.
Then, $sa(x) \neq 0$. 
Therefore, $\mu(sa)(x) \neq 0$ and $\lbl''' = \lbl''$ by construction.
Hence, $\lbl = \lbl'''$.

\item The step $\tup{sm,sa} \symEval{p}{\symPcObs{se} \concat \pcObs{\lbl'}} \sigma'$ has been derived using the \textsc{Beqz-Symb-Unsat} rule from the symbolic non-speculative semantics.
Then, $se$ is $sa(x) \neq 0$.
From this and $\mu \models se$, we have that  $\mu(sa)(x) \neq 0$ and, therefore, $\lbl''' = \lbl''$ by construction.
Hence, $\lbl = \lbl'''$.
\end{compactitem}
This completes the proof of this case.

\end{compactitem}
This completes the proof of our claim.

\item[\textbf{Rule \textsc{Am-Rollback}.}]
Assume that  $\tup{\ctr,\tup{sm,sa},s} \symSpeval{p}{\mbp}{\tau} \tup{\ctr',\tup{sm',sa'},s'}$ has been derived using the \textsc{Se-Rollback} rule in the symbolic semantics.
Then, $s = s_B \concat \tup{\sigma,\id, 0, \lbl}$, $\sigma \symEval{p}{\tau'} \sigma'$, $\ctr'= \ctr$, $
\sigma' = \tup{sm',sa'}$, $s' = s_B$, $sa'(\pc) \neq \lbl$, and $\tau = \rollbackObs{\id} \concat \pcObs{\sigma'(\pc)}$.
Let $\mu : \SymVal \to \Val$ be an arbitrary valuation satisfying $\pathCond{s} \wedge  \pathCond{\tau}$.
We claim that $\mu \models \pathCond{\tau'}$.
Therefore, Lemma~\ref{lemma:symbolic:one-step-soundness-non-speculative} to  $\sigma \symEval{p}{\tau'} \sigma'$ to derive $\mu(\sigma) \eval{p}{\mu(\tau')} \mu(\sigma')$.
By applying the rule \textsc{Am-Rollback} of the concrete semantics starting from  $\tup{\ctr,\tup{\mu(sm),\mu(sa)},\mu(s)}$, we obtain $\tup{\ctr, \tup{\mu(sm),\mu(sa)}, \mu(s_B) \concat \tup{\sigma,\id, 0, \lbl}} \ameval{p}{\mbp_n}{\rollbackObs{\id} \concat \pcObs{\sigma'(\pc)}} \tup{\ctr,\mu(\sigma'), \mu(s_B)}$.
Since $\sigma'(\pc) = \mu(\sigma')(\pc) = \lbl$ (since the program counter is always concrete), we have that  $\tup{\ctr, \tup{\mu(sm),\mu(sa)}, \mu(s)} \ameval{p}{\mbp}{\mu(\tau)} \tup{\ctr',\tup{\mu(sm'),\mu(sa')},\mu(s')}$.

We now prove our claim that $\mu \models \pathCond{\tau'}$.
There are four cases:
\begin{compactitem}
\item The step $\sigma \symEval{p}{\tau'} \sigma'$ has been derived using the \textsc{Beqz-Concr-Sat} rule from the symbolic non-speculative semantics.
Then, $\pathCond{\tau'} = \top$ and the claim trivially holds.

\item The step $\sigma \symEval{p}{\tau'} \sigma'$ has been derived using the \textsc{Beqz-Concr-Unsat} rule from the symbolic non-speculative semantics.
Then, $\pathCond{\tau'} = \top$ and the claim trivially holds.

\item The step $\sigma \symEval{p}{\tau'} \sigma'$ has been derived using the \textsc{Beqz-Symb-Sat} rule from the symbolic non-speculative semantics.
Then, $\pathCond{\tau'}$ is $sa(x) = 0$.
Observe that $sa(x) = 0$ is one of the conjuncts of $\pathCond{s}$ by construction.
Therefore, $\mu \models \pathCond{\tau'}$ since $\mu \models \pathCond{s}$.

\item The step $\sigma \symEval{p}{\tau'} \sigma'$ has been derived using the \textsc{Beqz-Symb-Unsat} rule from the symbolic non-speculative semantics.
Then, $\pathCond{\tau'}$ is $sa(x) \neq 0$.
Observe that $sa(x) \neq 0$ is one of the conjuncts of $\pathCond{s}$ by construction.
Therefore, $\mu \models \pathCond{\tau'}$ since $\mu \models \pathCond{s}$.
\end{compactitem}
This completes the proof of our claim.

\end{compactitem}
This completes the proof of our claim.
\end{proof}

In Proposition~\ref{propostion:alwasy-mispred-symbolic:soundness}, we finally prove the soundness of the symbolic semantics.

\begin{restatable}{prop}{}\label{propostion:alwasy-mispred-symbolic:soundness}
Let $p$ be a program and $\mu$ be a valuation.
Whenever $\tup{0,\tup{sm_0,sa_0}, \emptysequence} \symSpeval{p}{\mbp}{\tau}^* \tup{\ctr,\tup{sm,sa}, \emptysequence}$, $\tup{sm,sa} \in \Final$, and $\mu \models \pathCond{\tau}$, then $\tup{0,\tup{\mu(sm_0),\mu(sa_0)}, \emptysequence} \ameval{p}{\mbp}{\mu(\tau)}^* \tup{\ctr,\tup{\mu(sm),\mu(sa)}, \emptysequence}$.
\end{restatable}

\begin{proof}
Let $p$ be a program and $\mu$ be a valuation.
Assume that, for some $n \in \Nat$, $\tup{\ctr_0, \tup{sm_0, sa_0}, s_0} \symSpeval{p}{\mbp}{\tau_0} \tup{\ctr_1, \tup{sm_1, sa_1}, s_1} \symSpeval{p}{\mbp}{\tau_1} \ldots \tup{\ctr_n, \tup{sm_n, sa_n}, s_n} \symSpeval{p}{\mbp}{\tau_n} \tup{\ctr_{n+1}, \tup{sm_{n+1}, sa_{n+1}}, s_{n+1}}$, where $\ctr_0 = 0$, $s_0 = s_{n+1} = \emptysequence$, and  $\tup{sm_n, sa_n} \in \Final$.
Finally, assume that $\mu \models \pathCond{\tau_0 \concat \tau_1 \concat \ldots \concat \tau_n}$.
We claim that, for all $1 \leq i \leq n+1$, $\mu \models \pathCond{s_{i-1}}$.
We can obtain the concrete execution $\tup{\ctr_0, \tup{\mu(sm_0), \mu(sa_0)}, \mu(s_0)} \speval{p}{\mbp}{\mu(\tau_0 \concat \tau_1 \concat \ldots \concat \tau_n)}^n  \tup{\ctr_{n+1}, \tup{\mu(sm_{n+1}), \mu(sa_{n+1})}, \mu(s_{n+1})}$ by repeatedly applying Lemma~\ref{lemma:symbolic:one-step-soundness-speculative} (since for all $1 \leq i \leq n+1$, $\mu \models \pathCond{s_{i-1}} \wedge \pathCond{\tau_i}$).

We now show that for all $0 \leq i \leq n+1$, $\mu \models \pathCond{s_{i}}$.
We prove this by induction.
For the base case of $i = 0$, the claim trivially holds since $s_0 = \emptysequence$ and therefore $\pathCond{s_{0}} = \top$.
For the induction step, assume that the claim holds for all $j < i$.
We now show that it holds for $i$ as well.
There are three cases:
\begin{compactitem}
\item The $i$-th configuration has been derived using the \textsc{Se-NoBranch} rule.
Then, $s_{i-1}$ and $s_i$ only differ in the length of the remaining speculation windows.
Therefore, $\pathCond{s_{i-1}} = \pathCond{s_i}$.
From the induction hypothesis, $\mu \models \pathCond{s_{i-1}}$.
Therefore, $\mu \models \pathCond{s_i}$.

\item The $i$-th configuration has been derived using the \textsc{Se-Branch} rule.
Then, $s_i = \decrementTop{s_{i-1}} \concat \tup{\tup{sm_{i-1},sa_{i-1}}, \id, \mathit{min}(w, \window{s}-1), \lbl}$, where $p(sa_{i-1}(\pc)) = \pjz{x}{\lbl''}$, $\tup{sm_{i-1}, sa_{i-1}} \symEval{p}{\symPcObs{se} \concat \pcObs{\lbl'}} \sigma'$, 
$\lbl = 
\begin{cases}
sa_{i-1}(\pc) + 1 & \text{if}\ \lbl' \neq sa_{i-1}(\pc) + 1 \\
\lbl'' 	 	& \text{if}\ \lbl' = sa_{i-1}(\pc) + 1 \\
\end{cases}$, and $\pathCond{\tau_{i-1}} = se$.
From this, it follows that $\pathCond{\tup{\sigma, \id, 0, \lbl}} = se$.
From this, $\pathCond{\tau_{i-1}} = se$, and $\mu \models \pathCond{\tau_0 \concat \tau_1 \concat \ldots \concat \tau_n}$, we have that $\mu \models se$.
Moreover, from $\pathCond{s_i} = \pathCond{s_{i-1}} \wedge se$, $\mu \models se$, and the induction hypothesis, we have $\mu \models \pathCond{s_i}$.

\item The $i$-th configuration has been derived using the \textsc{Se-Rollback} rule.
Then, $s_{i-1} = s_i \concat \tup{\sigma, \id, 0, \lbl}$  for some $\sigma$, $\id$, and $\lbl$.
Therefore, $\pathCond{s_{i-1}} = \pathCond{s_i} \wedge \varphi$ for some $\varphi$.
From the induction hypothesis, $\mu \models \pathCond{s_{i-1}}$.
From this and $\pathCond{s_{i-1}} = \pathCond{s_i} \wedge \varphi$ for some $\varphi$, we have $\mu \models \pathCond{s_i}$.
	
\end{compactitem}
This completes the proof of our claim.
\end{proof}

\subsection{Completeness}\label{apx:alwasy-mispred-symbolic:completeness}

In Lemma~\ref{lemma:symbolic:one-step-completeness-non-speculative}, we prove the completeness of the non-speculative symbolic semantics.

\begin{restatable}{lem}{}\label{lemma:symbolic:one-step-completeness-non-speculative}
Let $p$ be a program, $sa$ be a symbolic assignment, and $sm$ be a symbolic memory.
Whenever  $\tup{m,a} \eval{p}{\tau} \tup{m',a'}$, then 
for all assignments $\mu : \SymVal \to \Val$ such that $\mu(sa) = a$ and $\mu(sm) = m$,
then there exist $\tau', sm',sa'$ such that 
$\tup{sm,sa} \symEval{p}{\tau'} \tup{sm',sa'}$,  $\mu(sm') = m'$, $\mu(sa') = a'$, $\mu(\tau') = \tau$, and $\mu \models \pathCond{\tau}$.
\end{restatable}

\begin{proof}
Let $p$ be a program, $sa$ be a symbolic assignment, and $sm$ be a symbolic memory.
Moreover, assume that  $\tup{m,a} \eval{p}{\tau} \tup{m',a'}$ holds.
Finally, let $\mu : \SymVal \to \Val$ be an arbitrary assignment such that  $\mu(sa) = a$ and $\mu(sm) = m$.
We proceed by case distinction on the rules defining $\eval{p}{}$.
\begin{compactitem}
\item[\textbf{Rule \textsc{Skip}.}]
Assume that $\tup{m,a} \eval{p}{\tau} \tup{m',a'}$ has been derived using the \textsc{Skip} rule in the concrete semantics.
Then, $p(a(\pc)) = \pskip$, $m' = m$, $a' = a[\pc \mapsto a(\pc)+1]$, and $\tau = \emptysequence$.
From $p(a(\pc)) = \pskip$, $\mu(sa) = a$, and $\pc$ being always concrete, we have that $p(sa(\pc)) = \pskip$.
Thus, we can apply the rule \textsc{Skip} of the symbolic semantics starting from $\tup{sm,sa}$.
Hence, $\tup{sm,sa} \symEval{p}{\tau'} \tup{sm',sa'}$, where $\tau' = \emptysequence$, $sm' = sm$, $sa' = sa[\pc \mapsto sa(\pc) +1]$, and $\varphi' = \varphi$.
From $m' = m$, $\mu(sm) = m$, and $sm' = sm$, we have that $\mu(sm') = m'$.
From Lemma~\ref{lemma:symbolic:useful}, $a' = a[\pc \mapsto a(\pc)+1]$, $\mu(sa) =a$,  and $sa' = sa[\pc \mapsto sa(\pc) +1]$, we have that $\mu(sa') = \mu(sa[\pc \mapsto sa(\pc) +1]) = \mu(sa)[\pc \mapsto \mu(sa(\pc))+1] = a[\pc \mapsto \mu(sa(\pc))+1] = a[\pc \mapsto \mu(sa)(\pc)+1] = a[\pc \mapsto a(\pc)+1] = a'$.
From $\tau = \emptysequence$ and $\tau' = \emptysequence$, we have that $\mu(\tau') = \tau$.
Finally, from $\pathCond{\tau'} = \top$, we have that $\mu \models \pathCond{\tau}$.

\item[\textbf{Rule \textsc{Barrier}.}]
The proof of this case is similar to that of  \textbf{Rule \textsc{Skip}.}

\item[\textbf{Rule \textsc{Assign}.}]
Assume that $\tup{m,a} \eval{p}{\tau} \tup{m',a'}$ has been derived using the \textsc{Assign} rule in the concrete semantics.
Then, $p(a(\pc)) = \passign{x}{e}$, $x \neq  \pc$, $m' = m$, $a' = a[\pc \mapsto a(\pc)+1, x \mapsto \exprEval{x}{a}]$, and $\tau = \emptysequence$.
From $p(a(\pc)) = \passign{x}{e}$, $\mu(sa) = a$, and $\pc$ being always concrete, we have that $p(sa(\pc)) = \passign{x}{e}$.
Thus, we can apply the rule \textsc{Assign} of the symbolic semantics starting from $\tup{sm,sa}$.
Hence, $\tup{sm,sa} \symEval{p}{\tau'} \tup{sm',sa'}$, where $\tau' = \emptysequence$, $sm' = sm$, and $sa' = sa[\pc \mapsto sa(\pc) +1, x \mapsto \exprEval{e}{sa}]$.
From $m' = m$, $\mu(sm) = m$, and $sm' = sm$, we have that $\mu(sm') = m'$.
From Lemma~\ref{lemma:symbolic:useful}, $a' = a[\pc \mapsto a(\pc)+1, x \mapsto \exprEval{e}{a}]$, $\mu(sa) =a$,  and $sa' = sa[\pc \mapsto sa(\pc) +1,x \mapsto \exprEval{e}{a}]$, we have that $\mu(sa') = \mu(sa[\pc \mapsto sa(\pc) +1,x \mapsto \exprEval{e}{sa}]) = \mu(sa)[\pc \mapsto \mu(sa(\pc)) +1,x \mapsto \mu(\exprEval{e}{sa})] = \mu(sa)[\pc \mapsto \mu(sa)(\pc) +1,x \mapsto \exprEval{e}{\mu(sa)}] = a[\pc \mapsto a(\pc) +1,x \mapsto \exprEval{e}{a}] = a'$.
From $\tau = \emptysequence$ and $\tau' = \emptysequence$, we have that $\mu(\tau') = \tau$.
Finally, from $\pathCond{\tau'} = \top$, we have that $\mu \models \pathCond{\tau}$.

\item[\textbf{Rule \textsc{ConditionalUpdate-Sat}.}]
Assume that $\tup{m,a} \eval{p}{\tau} \tup{m',a'}$ has been derived using the \textsc{ConditionalUpdate-Sat} rule in the concrete semantics.
Then, $p(a(\pc)) = \pcondassign{x}{e}{e'}$, $x \neq  \pc$, $\exprEval{e'}{a} = 0$, $m' = m$, $a' = a[\pc \mapsto a(\pc)+1, x \mapsto \exprEval{x}{a}]$, and $\tau = \emptysequence$.
From $p(a(\pc)) =  \pcondassign{x}{e}{e'}$, $\mu(sa) = a$, and $\pc$ being always concrete, we have that $p(sa(\pc)) =  \pcondassign{x}{e}{e'}$.
There are two cases:
\begin{compactitem}
\item[$\bm{\exprEval{e'}{sa} \in \Val:}$]
We claim that $\exprEval{e'}{sa} = 0 $.
Therefore,  we can apply the rule \textsc{ConditionalUpdate-Concr-Sat} of the symbolic semantics starting from $\tup{sm,sa}$.
Hence, $\tup{sm,sa} \symEval{p}{\tau'} \tup{sm',sa'}$, where $\tau' = \emptysequence$, $sm' = sm$, and $sa' = sa[\pc \mapsto sa(\pc) +1, x \mapsto \exprEval{e}{sa}]$.
From $m' = m$, $\mu(sm) = m$, and $sm' = sm$, we have that $\mu(sm') = m'$.
From Lemma~\ref{lemma:symbolic:useful}, $a' = a[\pc \mapsto a(\pc)+1, x \mapsto \exprEval{e}{a}]$, $\mu(sa) =a$,  and $sa' = sa[\pc \mapsto sa(\pc) +1,x \mapsto \exprEval{e}{a}]$, we have that $\mu(sa') = \mu(sa[\pc \mapsto sa(\pc) +1,x \mapsto \exprEval{e}{sa}]) = \mu(sa)[\pc \mapsto \mu(sa(\pc)) +1,x \mapsto \mu(\exprEval{e}{sa})] = \mu(sa)[\pc \mapsto \mu(sa)(\pc) +1,x \mapsto \exprEval{e}{\mu(sa)}] = a[\pc \mapsto a(\pc) +1,x \mapsto \exprEval{e}{a}] = a'$.
From $\tau = \emptysequence$ and $\tau' = \emptysequence$, we have that $\mu(\tau') = \tau$.
Finally, from $\pathCond{\tau'} = \top$, we have that $\mu \models \pathCond{\tau}$.

We now show that $\exprEval{e'}{sa} = 0 $.
From $\exprEval{e'}{sa} \in \Val$, it follows that all variables defining $e'$ have concrete values in $sa$.
From this and $\mu(sa) = a$, we have that all these variables agree in $a$ and $sa$.
Hence, $\exprEval{e'}{sa} = \exprEval{e'}{a}= 0 $.

\item[$\bm{\exprEval{e'}{sa} \not\in \Val:}$]
We can apply the rule \textsc{ConditionalUpdate-Symb} of the symbolic semantics starting from $\tup{sm,sa}$.
Hence, $\tup{sm,sa} \symEval{p}{\tau'} \tup{sm',sa'}$, where $\tau' = \emptysequence$, $sm' = sm$, and $sa' = sa[\pc \mapsto sa(\pc) +1, x \mapsto \ite{\exprEval{e'}{sa} = 0}{\exprEval{e}{sa}}{sa(x)} ]$.
From $m' = m$, $\mu(sm) = m$, and $sm' = sm$, we have that $\mu(sm') = m'$.
From $\tau = \emptysequence$ and $\tau' = \emptysequence$, we have that $\mu(\tau') = \tau$.
Finally, from $\pathCond{\tau'} = \top$, we have that $\mu \models \pathCond{\tau}$.
Finally, to show that $a' = \mu(sa')$, we need to show that $\mu(\ite{\exprEval{e'}{sa} = 0}{\exprEval{e}{sa}}{sa(x)})$ is equivalent to $\exprEval{e}{sa}$.
From the evaluation of symbolic expressions, $\mu(\ite{\exprEval{e'}{sa} = 0}{\exprEval{e}{sa}}{sa(x)})$ is $\mu(\exprEval{e}{sa})$ if $\mu(\exprEval{e'}{sa} = 0) \neq 0$, i.e., if $\mu(\exprEval{e'}{sa}) = 0$, and $\mu(sa(x))$ otherwise.
We now show that $\mu(\exprEval{e'}{sa}) = 0$, from this we immediately get that $a' = \mu(sa')$.
From  Lemma~\ref{lemma:symbolic:useful}, we have $\mu(\exprEval{e'}{sa})= \exprEval{e'}{\mu(sa)}$.
From this and $\mu(sa) = a$, we have $\mu(\exprEval{e'}{sa})= \exprEval{e'}{a}$.
From this and $\exprEval{e'}{a} = 0$, we finally get  $\mu(\exprEval{e'}{sa}) = 0$.
\end{compactitem}
This concludes the proof of this case.

\item[\textbf{Rule \textsc{ConditionalUpdate-Unsat}.}]
The proof of this case is similar to that of \textbf{Rule \textsc{ConditionalUpdate-Sat}.}

\item[\textbf{Rule \textsc{Load}.}]
Assume that $\tup{m,a} \eval{p}{\tau} \tup{m',a'}$ has been derived using the \textsc{Load} rule in the concrete semantics.
Then, $p(a(\pc)) = \pload{x}{e}$, $x \neq  \pc$, $m' = m$, $a' = a[\pc \mapsto a(\pc)+1, x \mapsto m(\exprEval{e}{a})]$, and $\tau = \loadObs{\exprEval{e}{a}}$.
From $p(a(\pc)) = \pload{x}{e}$, $\mu(sa) = a$, and $\pc$ being always concrete, we have that $p(sa(\pc)) = \pload{x}{e}$.
We proceed by applying the \textsc{Load-Symb} rule of the symbolic semantics starting from $\tup{sm,sa,\varphi}$.
Hence, $\tup{sm,sa} \symEval{p}{\tau'} \tup{sm',sa'}$, where $se = \exprEval{e}{sa}$, $se' = \symRead{sm}{se}$, $\tau' = \loadObs{se}$, $sm' = sm$, and $sa' = sa[\pc \mapsto sa(\pc) +1, x \mapsto se']$.
We claim that $\mu(se') = m(\exprEval{e}{a})$.
From $m' = m$, $\mu(sm) = m$, and $sm' = sm$, we have that $\mu(sm') = m'$.
From $a' = a[\pc \mapsto a(\pc)+1, x \mapsto m(\exprEval{e}{a})]$, $sa' = sa[\pc \mapsto sa(\pc) +1, x \mapsto se']$, and $\mu(se') = m(\exprEval{e}{a})$, we have that 
$\mu(sa') 
	= \mu(sa[\pc \mapsto sa(\pc) +1, x \mapsto se']) 
	= \mu(sa)[\pc \mapsto \mu(sa(\pc)) +1, x \mapsto \mu(se')] 
	= \mu(sa)[\pc \mapsto \mu(sa)(\pc) +1,x \mapsto m(\exprEval{e}{a})] 
	= a[\pc \mapsto a(\pc) +1,x \mapsto m(\exprEval{e}{a})] 
	= a'$ (by relying on $\mu(sa) = a$, $\mu(sm) = m$, our claim that $\mu(se') = m(\exprEval{e}{a})$, and Lemma~\ref{lemma:symbolic:useful}).
From $\tau = \loadObs{\exprEval{e}{a}}$ and $\tau' = \loadObs{\exprEval{e}{sa}}$, we have that $\mu(\tau') =  \loadObs{\mu(\exprEval{e}{sa})} = \loadObs{\exprEval{e}{\mu(sa)}} = \loadObs{\exprEval{e}{a}} = \tau$ (by applying Lemma~\ref{lemma:symbolic:useful} and relying on $\mu(sa) = a$).
Finally, from $\pathCond{\tau} = \top$, we have that $\mu \models \pathCond{\tau}$.

We now prove our claim that $\mu(se') = m(\exprEval{e}{a})$.
From  $se' = \symRead{sm}{se}$, we have that $\mu(se') = \mu(sm)(\mu(se))$ (from the evaluation of symbolic expressions).
From this and $\mu(sm) = m$, we have that $\mu(se') = m(\mu(se))$.
From $se = \exprEval{e}{sa}$, $\mu(se) = \mu(\exprEval{e}{sa}) = \exprEval{e}{\mu(sa)} = \exprEval{e}{a}$ (from $\mu(sa) = a$ and Lemma~\ref{lemma:symbolic:useful}).
Therefore, $\mu(se') = m(\exprEval{e}{a})$.

\item[\textbf{Rule \textsc{Store}.}]
Assume that $\tup{m,a} \eval{p}{\tau} \tup{m',a'}$ has been derived using the \textsc{Store} rule in the concrete semantics.
Then, $p(a(\pc)) = \pstore{x}{e}$, $m' = m[\exprEval{e}{a} \mapsto a(x) ]$, $a' = a[\pc \mapsto a(\pc)+1]$, and $\tau = \storeObs{\exprEval{e}{a}}$.
From $p(a(\pc)) = \pstore{x}{e}$, $\mu(sa) = a$, and $\pc$ being always concrete, we have that $p(sa(\pc)) = \pstore{x}{e}$.
We proceed by applying the \textsc{Store-Symb} rule of the symbolic semantics starting from $\tup{sm,sa}$.
Hence, $\tup{sm,sa} \symEval{p}{\tau'} \tup{sm',sa'}$, where $se = \exprEval{e}{sa}$, $\tau' = \storeObs{se}$, $sm' = \symWrite{sm}{se}{sa(x)}$, and $sa' = sa[\pc \mapsto sa(\pc) +1]$.
For showing that $\mu(sm') = m'$, we proceed as follows.
From the evaluation of symbolic expressions, 
$\mu(sm') 
	= \mu(\symWrite{sm}{se}{sa(x)})
	= \mu(sm)[\mu(se) \mapsto \mu(sa(x))]$.
From this, $\mu(sm) = m$, $\mu(sa) = a$, and Lemma~\ref{lemma:symbolic:useful}, we have
$\mu(sm)[\mu(se) \mapsto \mu(sa(x))]
	= m[\mu(se) \mapsto \mu(sa)(x)]
	= m[\mu(se) \mapsto a(x)]
$.
From this, $\mu(sa) = a$, $se = \exprEval{e}{sa}$, and Lemma~\ref{lemma:symbolic:useful}, we have
$m[\mu(se) \mapsto a(x)] 
	= m[\mu(\exprEval{e}{sa}) \mapsto a(x)]
	= m[\exprEval{e}{\mu(sa)} \mapsto a(x)]
	= m[\exprEval{e}{a} \mapsto a(x)]
	$.
Therefore, $\mu(sm') = m[\exprEval{e}{a} \mapsto a(x)]$.
Since $m' =m[\exprEval{e}{a} \mapsto a(x) ]$, we proved that $\mu(sm') = m'$.
From $a' = a[\pc \mapsto a(\pc)+1]$  and $sa' = sa[\pc \mapsto sa(\pc) +1]$, we have that 
$\mu(sa') 
	= \mu(sa[\pc \mapsto sa(\pc) +1]) 
	= \mu(sa)[\pc \mapsto \mu(sa(\pc)) +1] 
	= a[\pc \mapsto a(\pc) +1] 
	= a'$ (by relying on $\mu(sa) = a$ and Lemma~\ref{lemma:symbolic:useful}).
From $\tau = \storeObs{\exprEval{e}{a}}$ and $\tau' = \storeObs{\exprEval{e}{sa}}$, we have that $\mu(\tau') =  \storeObs{\mu(\exprEval{e}{sa})} = \storeObs{\exprEval{e}{\mu(sa)}} = \storeObs{\exprEval{e}{a}} = \tau$ (by applying Lemma~\ref{lemma:symbolic:useful} and relying on $\mu(sa) = a$).
Finally, from $\pathCond{\tau} = \top$, we have $\mu \models \pathCond{\tau}$.

We now prove our claim that $\mu(se') = m(\exprEval{e}{a})$.
From  $se' = \ite{se = 0}{sm(0)}{\ite{se = 1}{sm(1)}{\ldots}}$, we have that $\mu(se') = \mu(\ite{se = 0}{sm(0)}{\ite{se = 1}{sm(1)}{\ldots}})$.
Hence, $\mu(se') = \mu(sm(\mu(se)))$ (by relying on the semantics of $\ite{se}{se'}{se''}$).
From this and $se = \exprEval{e}{sa}$, we have that $\mu(se') = \mu(sm(\mu(\exprEval{e}{sa})))$.
By applying Lemma~\ref{lemma:symbolic:useful}, we have $\mu(se') = \mu(sm)(\exprEval{e}{\mu(sa)})$.
From this, $\mu(sm) = m$, and $\mu(sa) = a$, we have $\mu(se') = m(\exprEval{e}{a})$.

\item[\textbf{Rule \textsc{Beqz-Sat}.}]
Assume that $\tup{m,a} \eval{p}{\tau} \tup{m',a'}$ has been derived using the \textsc{Beqz-Sat} rule in the concrete semantics.
Then, $p(a(\pc)) = \pjz{x}{\lbl}$, $a(x) = 0$, $m' = m$, $a' = a[\pc \mapsto \lbl]$, and $\tau = \pcObs{\lbl}$.
From $p(a(\pc)) = \pjz{x}{\lbl}$, $\mu(sa) = a$, and $\pc$ being always concrete, we have that $p(sa(\pc)) = \pjz{x}{\lbl}$.
There are two cases:
\begin{compactitem}
\item Assume that $sa(x) \in \Val$.
From $a(x) = 0$ and $\mu(sa) = a$, we have that $sa(x) = 0$.
Hence, we proceed by applying the \textsc{Beqz-Concr-Sat} rule of the symbolic semantics starting from $\tup{sm,sa}$.
Hence, $\tup{sm,sa} \symEval{p}{\tau'} \tup{sm',sa'}$, where $\tau' = \symPcObs{\top} \concat \pcObs{\lbl}$, $sm' = sm$, and $sa' = sa[\pc \mapsto \lbl]$.
From $m' = m$, $\mu(sm) = m$, and $sm' = sm$, we have that $\mu(sm') = m'$.
From $a' = a[\pc \mapsto \lbl]$  and $sa' = sa[\pc \mapsto \lbl]$, we have that 
$\mu(sa') 
	= \mu(sa[\pc \mapsto \lbl]) 
	= \mu(sa)[\pc \mapsto \lbl] 
	= a[\pc \mapsto \lbl] 
	= a'$ (by relying on $\mu(sa) = a$  and Lemma~\ref{lemma:symbolic:useful}).
From $\tau = \pcObs{\lbl}$ and $\tau' = \symPcObs{\top} \concat \pcObs{\lbl}$, we have that $\mu(\tau') =  \tau$.
Finally, from $\pathCond{\tau'} = \top$, we have that $\mu \models \pathCond{\tau'}$.

\item Assume that $sa(x) \not\in \Val$ and let $se = sa(x)$.
We proceed by applying the \textsc{Beqz-Symb-Sat} rule of the symbolic semantics starting from $\tup{sm,sa}$.
Hence, $\tup{sm,sa} \symEval{p}{\tau'} \tup{sm',sa'}$, where $\tau' = \symPcObs{sa(x) = 0} \concat \pcObs{\lbl}$, $sm' = sm$, and $sa' = sa[\pc \mapsto \lbl]$.
From $m' = m$, $\mu(sm) = m$, and $sm' = sm$, we have that $\mu(sm') = m'$.
From $a' = a[\pc \mapsto \lbl]$  and $sa' = sa[\pc \mapsto \lbl]$, we have that 
$\mu(sa') 
	= \mu(sa[\pc \mapsto \lbl]) 
	= \mu(sa)[\pc \mapsto \lbl] 
	= a[\pc \mapsto \lbl] 
	= a'$ (by relying on $\mu(sa) = a$  and Lemma~\ref{lemma:symbolic:useful}).
From $\tau = \pcObs{\lbl}$ and $\tau' =  \symPcObs{sa(x) = 0} \concat \pcObs{\lbl}$, we have that $\mu(\tau') =  \tau$.
We now show that $\mu \models \pathCond{\tau'}$.
Since $\pathCond{\tau'}$ is $sa(x) = 0$, we need to show that $\mu(sa(x)) = 0$.
From Lemma~\ref{lemma:symbolic:useful}, $\mu(sa(x)) = \mu(sa)(x)$.
From this and $\mu(sa) = a$, $\mu(sa(x)) = a(x)$.
Therefore, $\mu(sa(x)) = 0$ since $a(x) = 0$.
Hence, $\mu \models sa(x) = 0$, i.e., $\mu \models \pathCond{\tau'}$.
\end{compactitem}

\item[\textbf{Rule \textsc{Beqz-Unsat}.}]
The proof of this case is similar to that of \textbf{Rule \textsc{Beqz-Unsat}.}

\item[\textbf{Rule \textsc{Jmp}.}]
Assume that $\tup{m,a} \eval{p}{\tau} \tup{m',a'}$ has been derived using the \textsc{Jmp} rule in the concrete semantics.
Then, $p(a(\pc)) = \pjmp{e}$, $\lbl = \exprEval{e}{a}$, $m' = m$, $a' = a[\pc \mapsto \lbl]$, and $\tau = \pcObs{\lbl}$.
From $p(a(\pc)) =  \pjmp{e}$, $\mu(sa) = a$, and $\pc$ being always concrete, we have that $p(sa(\pc)) =  \pjmp{e}$.
There are two cases:
\begin{compactitem}
\item Assume that $\exprEval{e}{sa} \in \Val$.
We proceed by applying the \textsc{Jmp-Concr} rule of the symbolic semantics starting from $\tup{sm,sa}$.
Hence, $\tup{sm,sa} \symEval{p}{\tau'} \tup{sm',sa'}$, where $\lbl' = \exprEval{e}{sa}$, $\tau' = \symPcObs{\top} \concat \pcObs{\lbl}'$, $sm' = sm$, and $sa' = sa[\pc \mapsto \lbl']$.
Observe that $\lbl = \lbl'$ (since $\lbl' = \exprEval{e}{sa}$ and $\mu(sa) = a$).
From $m' = m$, $\mu(sm) = m$, and $sm' = sm$, we have that $\mu(sm') = m'$.
From $a' = a[\pc \mapsto \lbl]$, $sa' = sa[\pc \mapsto \lbl']$, and $\lbl = \lbl'$, we have that 
$\mu(sa') 
	= \mu(sa[\pc \mapsto \lbl]) 
	= \mu(sa)[\pc \mapsto \lbl] 
	= a[\pc \mapsto \lbl] 
	= a'$ (by relying on $\mu(sa) = a$  and Lemma~\ref{lemma:symbolic:useful}).
From $\tau = \pcObs{\lbl}$, $\tau' = \symPcObs{\top} \concat \pcObs{\lbl'}$, and $\lbl = \lbl'$, we have that $\mu(\tau') =  \tau$.
Finally, from $\pathCond{\tau'} = \top$, we have that $\mu \models \pathCond{\tau'}$.

\item Assume that $\exprEval{e}{sa} \not\in \Val$ and let $se = \exprEval{e}{sa}$.
We proceed by applying the \textsc{Jmp-Symb} rule of the symbolic semantics starting from $\tup{sm,sa}$, where we pick $\lbl$ as target label.
Hence, $\tup{sm,sa} \symEval{p}{\tau'} \tup{sm',sa'}$, where $\tau' = \symPcObs{\exprEval{e}{sa} = \lbl} \concat \pcObs{\lbl}$, $sm' = sm$, and $sa' = sa[\pc \mapsto \lbl]$.
From $m' = m$, $\mu(sm) = m$, and $sm' = sm$, we have that $\mu(sm') = m'$.
From $a' = a[\pc \mapsto \lbl]$  and $sa' = sa[\pc \mapsto \lbl]$, we have that 
$\mu(sa') 
	= \mu(sa[\pc \mapsto \lbl]) 
	= \mu(sa)[\pc \mapsto \lbl] 
	= a[\pc \mapsto \lbl] 
	= a'$ (by relying on $\mu(sa) = a$  and Lemma~\ref{lemma:symbolic:useful}).
From $\tau = \pcObs{\lbl}$ and $\tau' =  \symPcObs{\exprEval{e}{sa} = \lbl} \concat \pcObs{\lbl}$, we have that $\mu(\tau') =  \tau$.
We now show that $\mu \models \pathCond{\tau'}$.
Since $\pathCond{\tau'}$ is $\exprEval{e}{sa} = \lbl$, we need to show that $\mu(\exprEval{e}{sa}) = \lbl$.
From Lemma~\ref{lemma:symbolic:useful}, $\mu(\exprEval{e}{sa}) = \exprEval{e}{\mu(sa)}$.
From this and $\mu(sa) = a$,  $\mu(\exprEval{e}{sa}) = \exprEval{e}{a}$.
Therefore, $\mu(\exprEval{e}{sa}) = \lbl$ since $\exprEval{e}{a} = \lbl$.
Hence, $\mu \models \exprEval{e}{sa} = \lbl$, i.e., $\mu \models \pathCond{\tau'}$.
\end{compactitem}

\item[\textbf{Rule \textsc{Terminate}.}]
The proof of this case is similar to that of \textbf{Rule \textsc{Skip}.}

\end{compactitem}
This completes the proof of our claim.
\end{proof}

In Lemma~\ref{lemma:symbolic:one-step-completeness-speculative}, we prove the completeness of one step of the symbolic semantics.

\begin{restatable}{lem}{}\label{lemma:symbolic:one-step-completeness-speculative}
Let $p$ be a program, $sa$ be a symbolic assignment, $sm$ be a symbolic memory,  and $ss$ be a sequence of symbolic speculative states.
Whenever  $\tup{\ctr,\tup{m,a},s} \ameval{p}{\mbp_n}{\tau} \tup{\ctr',\tup{m',a'},s'}$ and $\wf{s}$, then 
for all valuations $\mu : \SymVal \to \Val$ such that $\mu(sa) = a$, $\mu(sm) = m$, and $\mu(ss) = s$,
then there are $\tau', \ctr', sm',sa'$ such that 
$\tup{\ctr, \tup{sm,sa}, ss} \symSpeval{p}{\mbp_n}{\tau'} \tup{\ctr', \tup{sm',sa'}, ss'}$, $\mu(sm') = m'$, $\mu(sa') = a'$, $\mu(ss') = s'$, $\mu(\tau') = \tau$, and $\mu \models \pathCond{\tau'}$.
\end{restatable}

\begin{proof}
Let $p$ be a program, $sa$ be a symbolic assignment, $sm$ be a symbolic memory, $ss$ be a symbolic speculative state.
Moreover, assume that $\tup{\ctr,\tup{m,a},s} \ameval{p}{\mbp_n}{\tau} \tup{\ctr',\tup{m',a'},s'}$ holds and $\wf{s}$.
Finally, let $\mu : \SymVal \to \Val$ be an arbitrary valuation such that $\mu(sa) = a$, $\mu(sm) = m$, and $\mu(ss) = s$.
We proceed by case distinction on the rules defining $\ameval{p}{\mbp_n}{}$.
\begin{compactitem}
\item[\textbf{Rule \textsc{Am-NoBranch}.}]
Assume that $\tup{\ctr,\tup{m,a},s} \ameval{p}{\mbp_n}{\tau} \tup{\ctr',\tup{m',a'},s'}$ has been derived using the \textsc{Se-NoBranch} rule in the concrete semantics.
Then, $p(a(\pc)) \neq \pjz{x}{\lbl}$, $\tup{m,a} \eval{p}{\tau} \tup{m',a'}$, $\exhaustedTop{s}$, $\ctr' = \ctr$, and $s' = \decrementTop{s}$.
Observe that $p(sa(\pc)) \neq \pjz{x}{\lbl}$, and $\exhaustedTop{ss}$ immediately follow from   $p(a(\pc)) \neq \pjz{x}{\lbl}$ and $\exhaustedTop{s}$ respectively.
We can apply Lemma~\ref{lemma:symbolic:one-step-completeness-non-speculative} to $\tup{m,a} \eval{p}{\tau} \tup{m',a'}$.
Therefore, for all valuation $\mu$ such that $\mu(sa) = a$ and $\mu(sm) = m$, there are $sm', sa', \tau'$ such that  $\tup{sm,sa} \symEval{p}{\tau'} \tup{sm',sa'}$,  $\mu(sm') = m'$, $\mu(sa') = a'$,  $\mu(\tau') = \tau$, and $\mu \models \pathCond{\tau'}$.
Since $\mu$ is a valuation for which $\mu(sa) = a$ and $\mu(sm) = m$, we therefore have that $\tup{sm,sa} \symEval{p}{\tau'} \tup{sm',sa'}$,  $\mu(sm') = m'$, $\mu(sa') = a'$,  $\mu(\tau') = \tau$, and $\mu \models \pathCond{\tau'}$.
Hence, we can apply the \textsc{Am-NoBranch} rule in the symbolic semantics to $\tup{\ctr,\tup{sm,sa},ss}$.
Thus, $\tup{\ctr, \tup{sm,sa}, ss} \symSpeval{p}{\mbp_n}{\tau'} \tup{\ctr'', \tup{sm',sa'}, ss'}$, where $\ctr'' = \ctr$.
We have already established that $\mu(sm') = m'$, $\mu(sa') = a'$, and $\mu(\tau') = \tau$, and $\mu \models \pathCond{\tau'}$ (which all follow from the application of Lemma~\ref{lemma:symbolic:one-step-completeness-non-speculative}).
We still have to show $\mu(ss') = s'$.
This immediately follows from $\mu(ss) = s$, $s' =
 	\begin{cases}
 		\decrementTop{s} & \text{if}\ p(a(\pc)) \neq \pbarrier\\
		\zeroesTop{s}  & 	\text{otherwise}
 	\end{cases}$, $ss' =
 	\begin{cases}
 		\decrementTop{ss} & \text{if}\ p(sa(\pc)) \neq \pbarrier\\
		\zeroesTop{ss}  & 	\text{otherwise}
 	\end{cases}$, and $sa(\pc) = a(\pc)$.

\item[\textbf{Rule \textsc{Am-Branch}.}]
Assume that $\tup{\ctr,\tup{m,a},s} \ameval{p}{\mbp_n}{\tau} \tup{\ctr',\tup{m',a'},s'}$ has been derived using the \textsc{Se-Branch} rule in the concrete semantics.
Then, $p(a(\pc)) = \pjz{x}{\lbl'}$, $\exhaustedTop{s}$, $\lbl = \begin{cases} a(\pc) + 1 & \text{if}\ a(x) =0 \\ \lbl' & \text{if}\ a(x) \neq 0 \end{cases}$, $\tau = \startObs{\ctr} \concat \pcObs{\lbl}$, $m' = m$, $a' = a[\pc \mapsto \lbl]$, $\ctr' = \ctr+1$, and $s' = \decrementTop{s} \concat \tup{ \tup{m,a}, \ctr, min(w, \window{s} -1), \lbl}$.
Observe that $\exhaustedTop{ss}$ directly follows from $\mu(ss) = s$ and $\exhaustedTop{s}$ and $p(sa(\pc)) = \pjz{x}{\lbl'}$ follows from $\pc$ being always concrete and $\mu(sa) = a$.
Observe also that $\tup{m,a} \eval{p}{\pcObs{\lbl''}} \tup{m,a[\pc \mapsto \lbl'']}$ holds where $\lbl'' \in \{ \lbl', a(\pc)+1 \} \setminus \{ \lbl \}$.
We can apply  Lemma~\ref{lemma:symbolic:one-step-completeness-non-speculative} to $\tup{m,a} \eval{p}{\pcObs{\lbl''}} \tup{m,a[\pc \mapsto \lbl'']}$.
Therefore, for all valuations $\mu$ such that $\mu(sa) =a$ and $\mu(sm) = m$, $\tup{sm,sa} \symEval{p}{\tau'} \tup{sm',sa'}$, $\mu(sm') = m$, $\mu(sa') = a[\pc \mapsto \lbl'']$, $\mu(\tau') = \pcObs{\lbl''}$, and $\mu \models \pathCond{\tau'}$.
Since $\mu(sa) = a$ and $\mu(sm) = m$, we thus have that  $\tup{sm,sa} \symEval{p}{\tau''} \tup{sm',sa'}$, $\mu(sm') = m$, $\mu(sa') = a[\pc \mapsto \lbl'']$, $\mu(\tau'') = \pcObs{\lbl''}$, and $\mu \models \pathCond{\tau''}$.
Observe that, since  $p(sa(\pc)) = \pjz{x}{\lbl'}$, $\tau'$ is of the form $\symPcObs{se} \concat \pcObs{\lbl''}$.
Therefore, we can apply the \textsc{Am-Branch-Symb} rule in the symbolic semantics to $\tup{\ctr,\tup{sm,sa},ss}$ to derive $\tup{\ctr,\tup{sm,sa},ss} \ameval{p}{w}{\tau'} \tup{\ctr+1, \tup{sm',sa'}, ss'}$, where 
$\tau' = \symPcObs{se} \concat \startObs{\id} \concat \pcObs{\lbl'''}$,
$sm' = sm$,
$sa' = sa[\pc \mapsto \lbl''']$,
$
\lbl''' = 
\begin{cases}
\sigma(\pc) + 1 & \text{if}\ \lbl'' \neq sa(\pc) + 1 \\
\lbl' 	 	& \text{if}\ \lbl'' = sa(\pc) + 1 \\
\end{cases}$, 
$ss' = \decrementTop{ss} \concat \tup{\tup{sm,sa}, \ctr, \mathit{min}(w, \window{ss}-1), \lbl'''}$, and $\id = \ctr$.
We claim that $\lbl''' = \lbl$.
From $\mu(sm) = m$ and $m' = m$, we get that $\mu(sm') = \mu(sm) = m = m'$.
From $\mu(sa) = a$, $a' = a[\pc \mapsto \lbl]$, $sa' = sa[\pc \mapsto \lbl''']$, and $\lbl = \lbl'''$, we have that 
$
\mu(sa') = \mu(sa[\pc \mapsto \lbl''']) = \mu(sa)[\pc \mapsto \lbl'''] = a[\pc \mapsto \lbl'''] = a[\pc \mapsto \lbl] = a'$.
From $\tau = \startObs{\ctr} \concat \pcObs{\lbl}$,  $\tau' = \symPcObs{se} \concat \startObs{\id} \concat \pcObs{\lbl'''}$, and $\lbl = \lbl'''$, we have $\mu(\tau') = \tau$.
From $\mu(ss) = s$, $\mu(sa) = a$, $\mu(sm) = m$, $s' = \decrementTop{s} \concat \tup{ \tup{m,a}, \ctr, min(w, \window{s} -1), \lbl}$, $ss' = \decrementTop{ss} \concat \tup{\tup{sm,sa}, \ctr, \mathit{min}(w, \window{ss}-1), \lbl'''}$, and $\lbl = \lbl'''$, we have 
$\mu(ss') 
	= \mu(\decrementTop{ss} \concat \tup{\tup{sm,sa}, \ctr, \mathit{min}(w, \window{ss}-1), \lbl'''}) 
	= \decrementTop{\mu(ss)} \concat \tup{\tup{\mu(sm),\mu(sa)}, \ctr, \mathit{min}(w, \window{\mu(ss)}-1), \lbl'''} 
	= \decrementTop{s} \concat \tup{\tup{m,a}, \ctr, \mathit{min}(w, \window{s}-1), \lbl'''}
	= \decrementTop{s} \concat \tup{\tup{m,a}, \ctr, \mathit{min}(w, \window{s}-1), \lbl}
	= s'$.
Finally, $\mu \models \pathCond{\tau'}$ immediately follows from $\tau' = \symPcObs{se} \concat \startObs{\id} \concat \pcObs{\lbl'''}$, $\tau'' = \symPcObs{se} \concat \pcObs{\lbl''}$, and $\mu \models \pathCond{\tau''}$.

We now prove our claim that $\lbl = \lbl'''$.
There are two cases:
\begin{compactitem}
\item[$\bm{a(x) = 0}: $]
Then, $\lbl = a(\pc) + 1$.
Therefore, $\lbl'' = \lbl'$.
From this and $p$'s well-formedness, it follows that $\lbl'' \neq sa(\pc) +1$ (since we consider only \lang{} programs such that $p(\lbl) = \pjz{x}{\lbl'}$ and $\lbl' \neq \lbl+1$).
Hence, $\lbl'''= sa(\pc) +1$.
From this and $sa(\pc) = a(\pc)$, we have that $\lbl = \lbl'''$.

\item[$\bm{a(x) \neq 0}: $]
Then, $\lbl = \lbl'$.
Therefore, $\lbl'' = a(\pc) +1$.
From this and $p$'s well-formedness, $\lbl'''= \lbl$.
Hence, $\lbl = \lbl'''$.

\end{compactitem}

\item[\textbf{Rule \textsc{Am-Rollback}.}]
Assume that $\tup{\ctr,\tup{m,a},s} \ameval{p}{\mbp_n}{\tau} \tup{\ctr',\tup{m',a'},s'}$ has been derived using the \textsc{Am-Rollback} rule in the concrete semantics.
Then, $s = s_B \concat \tup{\id, 0, \lbl, \tup{m'',a''}}$, $\tau = \rollbackObs{\id} \concat \pcObs{a'(\pc)}$, $\tup{m'',a''} \eval{p}{\tau''} \tup{m''',a'''}$,  $\ctr' = \ctr$, $m' = m'''$, $a' = a'''$, and $s' = s_B$. 
From $\mu(ss) = s$ and $s = s_B \concat \tup{\id, 0, \lbl, \tup{m'',a''}}$, it follows that $ss =  ss_B \concat \tup{\id, 0, \lbl, \tup{sm'',sa''}}$ and $\mu(sa'') = a''$ and $\mu(sm'') = m''$.
Since $\mu(sa'') = a''$ and $\mu(sm'') = m''$, we can apply Lemma~\ref{lemma:symbolic:one-step-completeness-non-speculative} to $\tup{m'',a''} \eval{p}{\tau''} \tup{m''',a'''}$.
We, therefore, obtain that there are $sm''', sa''', \tau'''$ such that $\tup{sm'',sa''} \symEval{p}{\tau'} \tup{sm''',sa'''}$,  $\mu(sm''') = m'''$, $\mu(sa''') = a'''$, $\mu(\tau''') = \tau''$, and $\mu \models \pathCond{\tau'''}$.
We claim that $sa'''(\pc) \neq \lbl$.
Hence, we can apply the \textsc{Am-Rollback} rule in the symbolic semantics to $\tup{\ctr,\tup{sm,sa},ss}$.
Thus, $\tup{\ctr, \tup{sm,sa}, ss} \symSpeval{p}{\mbp_n}{\tau'} \tup{\ctr'', \tup{sm',sa'}, ss'}$, where $\ctr'' = \ctr$, $sm' = sm'''$, $sa' = sa'''$, and $ss' = ss_B$.
From $\ctr''  = \ctr$ and $\ctr' = \ctr$, we have $\ctr' = \ctr''$.
From $\mu(sa''') =a'''$, $a' = a'''$, and $sa' = sa'''$, we have $\mu(sa') = a'$.
From $\mu(sm''') =m'''$, $m' = m'''$, and $sm' = sm'''$, we have $\mu(sm') = m'$.
From $\tau = \rollbackObs{\id} \concat \pcObs{a'(\pc)}$ and $\tau' =  \rollbackObs{\id} \concat \pcObs{sa'(\pc)}$, we have $\mu(\tau') = \tau$ and $\mu \models \pathCond{\tau'}$.
Finally, from $s' = s_B$ and $ss' = ss_B$, we have
$\mu(ss') 
	= \mu(ss_B)
	= s_B$ (thanks to $\mu(ss) = s$) and $\mu \models \pathCond{ss'}$ (thanks to $\mu \models \pathCond{ss})$.

We now prove our claim that $sa'''(\pc) \neq \lbl$.
From $s = s_B \concat \tup{\id, 0, \lbl, \tup{m'',a''}}$, $\tup{m'',a''} \eval{p}{\tau''} \tup{m''',a'''}$, and $\wf{s}$, we immediately have that $a'''(\pc) \neq \lbl$.
From this and $sa'''(\pc) = a'''(\pc)$, we have $sa'''(\pc) \neq \lbl$.
 
\end{compactitem}
This completes the proof of our claim.
\end{proof}

In Proposition~\ref{propostion:alwasy-mispred-symbolic:completeness} we finally prove the completeness of our symbolic semantics.

\begin{restatable}{prop}{}\label{propostion:alwasy-mispred-symbolic:completeness}
Let $p$ be a program.
Whenever $\tup{0, \sigma, \emptysequence} \ameval{p}{\bp}{\tau'}^* \tup{\ctr,\sigma',\emptysequence}$, $\sigma \in \Init$, and $\sigma' \in \Final$, then there is a valuation $\mu$, a symbolic trace $\tau'$, and a final symbolic configuration $\tup{sm,sa}$ such that $\tup{0,\tup{sm_0,sa_0}, \emptysequence, \top} \symSpeval{p}{\mbp}{\tau'}^* \tup{\ctr,\tup{sm,sa}, \emptysequence}$, $\mu( \tup{sm_0,sa_0} ) = \sigma$, $\mu(\tup{sm,sa}) = \sigma'$, $\mu(\tau') = \tau$, and $\mu \models \pathCond{\tau'}$.
\end{restatable}

\begin{proof}
Let $p$ be a program and $\tup{\ctr_0, \sigma_0, s_0} \ameval{p}{\bp}{\tau_0} \tup{\ctr_1, \sigma_1, s_1} \ameval{p}{\bp}{\tau_1} \ldots \tup{\ctr_n, \sigma_n, s_n} \ameval{p}{\bp}{\tau_n} \tup{\ctr_{n+1}, \sigma_{n+1}, s_{n+1}}$ be a concrete execution, where $\ctr_0 = 0$, $s_0 = s_{n+1} = \emptysequence$, $\sigma_0 \in \Init$, and $\sigma_{n+1} \in \Final$.
We claim that $\wf{s_i}$ holds for all $0 \leq i \leq n+1$.
Furthermore, let $\mu$ be a valuation such that $\mu(\tup{sm_0,sa_0}) = \sigma_0$.
We can construct the symbolic execution $\tup{0,\tup{sm_0,sa_0}, \emptysequence, \top} \symSpeval{p}{\mbp}{\tau'}^* \tup{\ctr,\tup{sm,sa}, \emptysequence}$  by repeatedly applying Lemma~\ref{lemma:symbolic:one-step-completeness-speculative}.
From the same Lemma, we also obtain that $\mu( \tup{sm_0,sa_0} ) = \sigma_0$, $\mu(\tup{sm,sa}) = \sigma_{n+1}$, $\mu(\tau') = \tau_0 \concat \ldots \concat \tau_n$, and $\mu \models \pathCond{\tau'}$.

We prove by induction that  $\wf{s_i}$ holds for all $0 \leq i \leq n+1$.
For the base case $i=0$, $s_0 = \emptysequence$ and $\wf{s_0}$ holds trivially.
For the induction step, we assume that $\wf{s_j}$ holds for all $j < i$ and we show that it holds also for $s_j$.
There are three cases:
\begin{compactitem}
\item The $i$-th configuration has been derived using the \textsc{Am-NoBranch} rule.
Then, $s_{i-1}$ and $s_i$ only differ in the length of the remaining speculation windows.
Therefore, $\wf{s_{i}}$ holds iff $\wf{s_{i-1}}$ does.
From the induction hypothesis, $\wf{s_{i-1}}$ holds and therefore  $\wf{s_{i}}$ holds.

\item 
The $i$-th configuration has been derived using the \textsc{Am-Branch} rule.
Then, $s_i = \decrementTop{s_{i-1}} \concat \tup{\sigma_{i-1},\ctr, \mathit{min}(w, \window{s}-1), \lbl}$, where $p(\sigma_{i-1}(\pc)) = \pjz{x}{\lbl''}$, $\sigma_{i-1} \eval{p}{\pcObs{\lbl'}} \sigma'$, and
$\lbl = 
\begin{cases}
\sigma_{i-1}(\pc) + 1 & \text{if}\ \sigma_{i-1}(x) = 0  \\
\lbl' 	 	& \text{if}\ \sigma_{i-1}(x) \neq 0 \\
\end{cases}$.
From this, $\wf{s_i}$ holds iff both $\wf{s_{i-1}}$ and $\wf{\tup{\sigma_{i-1},\ctr, \mathit{min}(w, \window{s}-1), \lbl}}$ hold.
The former follows from the induction hypothesis, the latter follows from 
$\sigma_{i-1} \eval{p}{\pcObs{\lbl'}} \sigma'$ and $\lbl = 
\begin{cases}
\sigma_{i-1}(\pc) + 1 & \text{if}\ \sigma_{i-1}(x) = 0  \\
\lbl' 	 	& \text{if}\ \sigma_{i-1}(x) \neq 0 \\
\end{cases}$.

\item The $i$-th configuration has been derived using the \textsc{Am-Rollback} rule.
Then, $s_{i-1} = s_i \concat \tup{\sigma, \id, 0, \lbl}$  for some $\sigma$, $\id$, and $\lbl$, i.e.,  if $\wf{s_{i-1}}$ holds, then $\wf{s_i}$ holds as well.
From the induction hypothesis, $\wf{s_{i-1}}$ holds and, thus,  $\wf{s_{i}}$ holds.\looseness=-1	
\end{compactitem}
This completes the proof of our claim.

\end{proof}

%% file: proofs.tex
\section{\tool{}'s soundness and completeness (Theorem~\ref{theorem:soundness-and-completeness})}\label{appendix:proofs}

Below, we provide the proof of Theorem~\ref{theorem:soundness-and-completeness}, which we restate here for simplicity:
\spectectorSoundnessAndCompleteness*

\begin{proof}
	Theorem~\ref{theorem:soundness-and-completeness} immediately follows from \tool{}'s soundness (see Lemma~\ref{theorem:soundness}) and completeness (see Lemma~\ref{theorem:completeness}).
\end{proof}

Lemma~\ref{theorem:soundness} states that if $\tool{}$ determines a program to be secure, the program is speculatively non-interferent.\looseness=-1

\begin{restatable}{lem}{spectectorSoundness}
\label{theorem:soundness}
Whenever $\tool{}(p,\policy,w) = \textsc{Secure}$,  the program $p$ satisfies speculative non-interference w.r.t. the  policy~$\policy$ and any oracle $\bp$ with speculative window at most~$w$.
\end{restatable}

\begin{proof}
	If $\tool{}(p,\policy,w) = \textsc{Secure}$, then for all symbolic traces $\tau$, $\memcheck(\tau) = \bot$ and $\pccheck(\tau) = \bot$.
	Let $\sigma$ and $\sigma'$ be two arbitrary $\policy$-indistinguishable initial configurations producing the same concrete non-speculative trace $\tau$.
	From Proposition~\ref{proposition:always-mispredict:speculative-and-non-speculative}, $\sigma$ and $\sigma'$ result in two concrete speculative traces $\tau_c$ and $\tau_c'$  with the same non-speculative projection.
	From Propositions~\ref{proposition:symbolic-execution}, $\tau_c$ and $\tau_c'$ correspond to two symbolic traces $\tau_s$ and $\tau_s'$.
	Since $\pccheck(\tau_s) = \bot$, $\pccheck(\tau_s') = \bot$, and $\nspecProject{\tau_c} = \nspecProject{\tau_c'}$, speculatively executed control-flow instructions produce the same outcome in $\tau_c$ and $\tau_c'$.
	Hence, the same code is executed in both traces and $\tau_s = \tau_s'$.
	From $\memcheck(\tau_s) = \bot$, the observations produced by speculatively executed $\loadKywd$ and $\storeKywd$ instructions are the same.
	Thus, $\tau_c = \tau_c'$.
	Hence, whenever two initial configurations result in the same non-speculative traces, then they produce the same speculative traces.
	Therefore, $p$ satisfies speculative non-interference w.r.t. the always-mispredict semantics with speculative window $w$.
	From this and Theorem~\ref{theorem:always-mispredict}, $p$ satisfies speculative non-interference w.r.t. any prediction oracle with speculative window at most $w$.
\end{proof}

Lemma~\ref{theorem:completeness} states that the leaks found by $\tool{}$ are valid counterexamples to speculative non-interference.

\begin{restatable}{lem}{spectectorCompleteness}
\label{theorem:completeness}
Whenever $\tool{}(p,\policy,w) = \textsc{Insecure}$, there is an oracle $\bp$ with speculative window at most $w$ such that program $p$ does not satisfy speculative non-interference w.r.t. $\bp$ and the policy~$\policy$.
\end{restatable}

\begin{proof}
	If $\tool{}(p,\policy) = \textsc{Insecure}$, there is a symbolic trace $\tau$ for which either $\memcheck(\tau) = \top$ or $\pccheck(\tau) = \top$.
	In the first case, $\pathCond{\tau} \wedge \policyEqv{\policy} \wedge \cstrs{\nspecProject{\tau}} \wedge \neg \cstrs{\specProject{\tau}}$ is satisfiable.
	Then, there are two models for the symbolic trace $\tau$ that (1) satisfy the path condition encoded in $\tau$, (2) agree on the non-sensitive registers and memory locations in $\policy$, (3) produce the same non-speculative projection, and (4) the speculative projections differ on a $\loadKywd$ or $\storeKywd$ observation.
	From Proposition~\ref{proposition:symbolic-execution}, the two concretizations correspond to two concrete runs, with different traces, whose non-speculative projection are the same.
	By combining this with Proposition~\ref{proposition:always-mispredict:speculative-and-non-speculative}, there are two configurations that produce the same non-speculative trace but different speculative traces.
	This is a violation of speculative non-interference.

	In the second case, there is a prefix $\nu \concat \symPcObs{\sexpr}$ of $\specProject{\tau}$ such that $\pathCond{\nspecProject{\tau}\concat \nu} \wedge \policyEqv{\policy} \wedge \cstrs{\nspecProject{\tau}} \wedge  \neg (\sexpr_1 \leftrightarrow \sexpr_2)$ is satisfiable.
	Hence, there are two symbolic traces $\tau$ and $\tau'$ that produce the same non-speculative observations but differ on a program counter observation $\pcObs{n}$ in their speculative projections.
	Again, this implies, through Propositions~\ref{proposition:symbolic-execution} and~\ref{proposition:always-mispredict:speculative-and-non-speculative}, that there are two $\policy$-indistinguishable initial configurations producing the same non-speculative traces but distinct speculative traces, leading to a violation of speculative non-interference.

	In both cases, $p$ does not satisfy speculative non-interference w.r.t. the always-mispredict semantics with speculative window $w$.
	From this and Theorem~\ref{theorem:always-mispredict}, there is a prediction oracle $\bp$ with speculative window at most $w$ such that $p$ does not satisfy speculative non-interference w.r.t. $\bp$.	
\end{proof}